\documentclass[acmsmall,nonacm, screen]{acmart}

\RequirePackage{mathtools}
\RequirePackage{xspace}
\RequirePackage{stmaryrd}
\RequirePackage{mathpartir}

\RequirePackage{bm}

\RequirePackage{xcolor}
\RequirePackage{pifont}
\RequirePackage{float}
\RequirePackage{multirow}

\RequirePackage{cleveref}

\RequirePackage{enumitem}

\RequirePackage{aliascnt}

\theoremstyle{plain}
\newtheorem{theorem}{Theorem}[section]

\newaliascnt{lemma}{theorem}
\newtheorem{lemma}[lemma]{Lemma}
\aliascntresetthe{lemma}

\newaliascnt{corollary}{theorem}

\aliascntresetthe{corollary}

\newaliascnt{proposition}{theorem}

\aliascntresetthe{proposition}

\newtheoremstyle{deflight}{}{}{\normalfont}{0pt}{\scshape}{.}{ }{}\theoremstyle{deflight}
\newaliascnt{definition}{theorem}

\aliascntresetthe{definition}

\makeatletter
\let\example\@undefined
\let\endexample\@undefined
\makeatother
\newaliascnt{example}{theorem}
\newtheorem{example}[example]{Example}
\aliascntresetthe{example}

\crefname{theorem}{Theorem}{Theorems}
\Crefname{theorem}{Theorem}{Theorems}
\crefname{lemma}{Lemma}{Lemmas}
\Crefname{lemma}{Lemma}{Lemmas}
\crefname{corollary}{Corollary}{Corollaries}
\Crefname{corollary}{Corollary}{Corollaries}
\crefname{proposition}{Proposition}{Propositions}
\Crefname{proposition}{Proposition}{Propositions}
\crefname{definition}{Definition}{Definitions}
\Crefname{definition}{Definition}{Definitions}
\crefname{example}{Example}{Examples}
\Crefname{example}{Example}{Examples}

\theoremstyle{remark}
\newaliascnt{remark}{theorem}

\aliascntresetthe{remark}
\crefname{remark}{Remark}{Remarks}
\Crefname{remark}{Remark}{Remarks}

\crefformat{section}{\S#2#1#3}
\Crefformat{section}{\S#2#1#3}
\crefformat{subsection}{\S#2#1#3}
\Crefformat{subsection}{\S#2#1#3}
\crefformat{subsubsection}{\S#2#1#3}
\Crefformat{subsubsection}{\S#2#1#3}
\crefrangeformat{section}{\S\S#3#1#4--#5#2#6}
\Crefrangeformat{section}{\S\S#3#1#4--#5#2#6}
\crefrangeformat{subsection}{\S\S#3#1#4--#5#2#6}
\Crefrangeformat{subsection}{\S\S#3#1#4--#5#2#6}
\crefmultiformat{section}{\S\S#2#1#3}{ and~#2#1#3}{, #2#1#3}{, and~#2#1#3}
\Crefmultiformat{section}{\S\S#2#1#3}{ and~#2#1#3}{, #2#1#3}{, and~#2#1#3}
\crefmultiformat{subsection}{\S\S#2#1#3}{ and~#2#1#3}{, #2#1#3}{, and~#2#1#3}
\Crefmultiformat{subsection}{\S\S#2#1#3}{ and~#2#1#3}{, #2#1#3}{, and~#2#1#3}

\RequirePackage{listings}
\RequirePackage{listings}
\RequirePackage{xcolor}

\definecolor{entity.name.type}{HTML}{267F99}
\definecolor{entity.name.sort}{HTML}{267F99}
\definecolor{entity.name.module}{HTML}{267F99}
\definecolor{entity.name.function}{HTML}{795E26}
\definecolor{entity.name.lifetime}{HTML}{0000FF}
\definecolor{variable.other.enummember}{HTML}{0070C1}
\definecolor{variable.other.constant}{HTML}{0070C1}
\definecolor{keyword}{HTML}{0000FF}
\definecolor{keyword.control}{HTML}{AF00DB}
\definecolor{variable}{HTML}{001080}
\definecolor{constant.numeric}{HTML}{098658}
\definecolor{meta.attribute}{HTML}{000000}
\definecolor{comment}{HTML}{008000}
\definecolor{string}{HTML}{A31515}

\newcommand\codesize{\footnotesize}

\lstdefinelanguage{fluxlang}{
    basicstyle=\ttfamily\codesize,
    numbers=left,
    numberstyle=\color{darkgray},
    numbersep=5pt,
    morecomment=[l]{//},
    commentstyle=\color{comment},
    mathescape=true,
    escapeinside=``,
    identifierstyle={\color{variable}},
    morestring=*[d]{"},
    stringstyle=\color{string},
emph={[1]i8,i16,i32,i64,i128,isize,u8,u16,u32,u64,u128,usize,f32,f64,bool,RVec,RMat,Vec,T,B,Self,Copy,List,Box,Nat,nat,Range,Option,Iterator,Item,HashMap,Particle,String,K,Eq,Hash,Expr,ItemId,ItemOwner,ItemPrice,ItemDescr,E1,E2,Add,Eq,Gt,And,PhantomData,Q,A,Allocator,Global,FieldPolicy,U,V,Changeset,F,Context,User,ItemLevel,UserId,RowPolicy,IndexTrait},
    emphstyle={[1]\color{entity.name.type}},
emph={[2]sim,sim_bad,select_list,select_where,update_where,allow_update,auth_user,insert_row,allow_insert, and,add,or,not,mk_weights,run,wishPrice,eq,lt,gt,maximum,cmp,in_bounds,index,borrow,index_mut,max_pos,max_neg,assert,simulation_bad,len,get_mut,interact,get,new,test_multi_index,push,pop,lookup,contains_key,insert,get_items,eval,show_items,update_item_descr,get_tx,sum_tx,split_off,as_slice,dedup,pop_if},
    emphstyle={[2]\color{entity.name.function}},
keywords={fn,let,mut,strg,impl,ensures,where,requires,self,struct,pub,forall,true,false,enum,unsafe,spec,trait,as,type,reft,refine,refinement,const,hrn,hdl},
    keywordstyle={\color{keyword}},
    emph={[4]sig,flux,refined_by,variant},
    emphstyle={[4]\color{meta.attribute}},
emph={[5]0},
    emphstyle={[5]\color{constant.numeric}},
emph={[6]int},
    emphstyle={[6]\color{entity.name.sort}},
emph={[7]cont},
    emphstyle={[7]\bf},
emph={[8]std,mem,collections,marker},
    emphstyle={[1]\color{entity.name.module}},
emph={[9]Cons,Nil,MIN,None,Some},
    emphstyle={[9]\color{variable.other.enummember}},
emph={[10]if,else,while,match,while,in,return,for},
    emphstyle={[10]\color{keyword.control}},
emph={[11]},
    emphstyle={[11]\color{entity.name.function}},
emph={[12]Query,Field,R,V,Wish,Ord,Ordering,Index,Output,Idx,MultiBorrow,T,Multi,MultiIndex,MultiIdx,C,IndexMut},
    emphstyle={[12]\color{entity.name.type}},
emph={[13]Greater},
    emphstyle={[13]\color{variable.other.enummember}},
    literate=*{'0}{{{\color{constant.numeric}0}}}{1}
        {'1}{{{\color{constant.numeric}1}}}{1}
        {'2}{{{\color{constant.numeric}2}}}{1}
        {'5}{{{\color{constant.numeric}5}}}{1}
        {'255}{{{\color{constant.numeric}255}}}{3}
        {'len}{{{\color{variable}len}}}{3}
        {'index}{{{\color{variable}index}}}{5}
        {'in_bounds}{{{\color{variable}in\_bounds}}}{9},
    emph={[14]PUBLIC},
    emphstyle={[14]\color{variable.other.constant}},
    moredelim=[is][\color{black}]{@@}{@@},
    moredelim=[is][\color{red}]{@!}{!@}
}

\lstnewenvironment{fluxsh}[1][]{\lstset{language=fluxlang, numbers=none,xleftmargin=.05in, escapeinside={(*}{*)}, #1}}{}
\lstnewenvironment{framedflux}[1][]{
    \lstset{
        language=fluxlang,
        frame=lines,
        xleftmargin=15pt,
        #1
    }
}{}

\def\inlineflux{\lstinline[language=fluxlang, basicstyle=\ttfamily, numbers=none]}
\def\iflux{\lstinline[language=fluxlang, basicstyle=\ttfamily\small, numbers=none]}
\RequirePackage{subcaption}
\definecolor{lnkw}{HTML}{D73A49}
\definecolor{lnty}{HTML}{6F42C1}
\definecolor{lncom}{HTML}{6A737D}
\definecolor{lnstr}{HTML}{032F62}
\lstdefinelanguage{Lean4}{
  morekeywords=[1]{def,let,fun,theorem,lemma,example,by,if,then,else,match,with,
                   structure,instance,class,where,do,return,open,namespace,end,
                   import,deriving,set_option,mutual,partial,abbrev,qualif},
  morekeywords=[2]{Int,Nat,Bool,Prop,Type,True,False,Option,List},
  morecomment=[l]{--},
  morecomment=[s]{/-}{-/},
  morestring=[b]",
  sensitive=true,
  basicstyle=\ttfamily\small,
  keywordstyle=[1]\color{lnkw}\bfseries,
  keywordstyle=[2]\color{lnty},
  commentstyle=\color{lncom}\itshape,
  stringstyle=\color{lnstr},
  extendedchars=true,
  literate=
    {≤}{{$\leq$}}1
    {≥}{{$\geq$}}1
    {≠}{{$\neq$}}1
    {≡}{{$\equiv$}}1
    {∧}{{$\wedge$}}1
    {∨}{{$\vee$}}1
    {∀}{{$\forall$}}1
    {∃}{{$\exists$}}1
    {→}{{$\to$}}1
    {↦}{{$\mapsto$}}1
    {λ}{{$\lambda$}}1
    {κ}{{$\kappa$}}1
    {ν}{{$\nu$}}1
    {∈}{{$\in$}}1
    {⟨}{{$\langle$}}1
    {⟩}{{$\rangle$}}1
    {σ}{{$\sigma$}}1
    {¬}{{$\neg$}}1
    {×}{{$\times$}}1
    {ρ}{{$\rho$}}1
    {Γ}{{$\Gamma$}}1
    {⊢}{{$\vdash$}}1,
  showstringspaces=false,
}
\lstdefinestyle{leanhl}{
  language=Lean4,
  keywords={grind,using,with,def,have,let,in,if,qual,then,else,by,example, theorem},
  escapeinside={(*}{*)},
}
\newcommand{\rname}[1]{\textsc{\bfseries\small #1}}

\newcommand{\sub}[2]{[#1 \mathrel{:=} #2]}

\newcommand{\eqdef}{\doteq}

\newcommand{\letin}[2]{\mathsf{let}\ #1\ \mathsf{in}\ #2}

\newcommand{\opname}[1]{\ensuremath{\operatorname{#1}}}
\newcommand{\myopname}[1]{\ensuremath{\mathsf{#1}}\xspace}

\newcommand{\walkop}{\myopname{Cert}}

\newcommand{\rfl}{\mathsf{rfl}}

\newcommand{\domop}{\opname{dom}}

\newcommand{\leantype}[1]{\ensuremath{\mathit{#1}}}
\newcommand{\Prop}{\leantype{Prop}\xspace}
\newcommand{\Nat}{\leantype{Nat}\xspace}
\newcommand{\Int}{\leantype{Int}\xspace}
\newcommand{\Bool}{\leantype{Bool}\xspace}
\newcommand{\BitVec}[1]{\leantype{BitVec}~{#1}}

\newcommand{\Expr}{\leantype{Expr}\xspace}

\newcommand{\lean}[1]{\texttt{#1}}

\newcommand{\zap}{\myopname{Zap}}
\newcommand{\fusiong}{\myopname{Zap\_grind}}
\newcommand{\fusiona}{\myopname{Zap\_aesop}}

\newcommand{\Preds}{\ensuremath{\mathcal{P}}}
\newcommand{\Terms}{\ensuremath{\mathcal{T}}}
\newcommand{\Vars}{\ensuremath{\mathcal{V}}}
\newcommand{\sys}{\textsc{Flex}\xspace}
\newcommand{\sclean}{\textsc{Lean}\xspace}
\newcommand{\scprusti}{\textsc{Prusti}\xspace}
\newcommand{\scverus}{\textsc{Verus}\xspace}
\newcommand{\fstar}{\ensuremath{\textsc{F}^{*}}\xspace}
\newcommand{\stainless}{\textsc{Stainless}\xspace}
\newcommand{\sclh}{\textsc{LiquidHaskell}\xspace}
\newcommand{\scflux}{\textsc{Flux}\xspace}
\newcommand{\flux}{\scflux}
\newcommand{\scthrust}{\textsc{Thrust}\xspace}
\newcommand{\scisabelle}{\textsc{Isabelle}\xspace}
\newcommand{\sccreusot}{\textsc{Creusot}\xspace}
\newcommand{\scWhy}{\textsc{Why3}\xspace}
\newcommand{\scrocq}{\textsc{Rocq}\xspace}

\newcommand{\leanterm}{\ensuremath{t}\xspace}
\newcommand{\termprop}{\ensuremath{\varphi}\xspace}

\newcommand{\pred}{\ensuremath{p}}
\newcommand{\cstr}{\ensuremath{c}}
\newcommand{\Cstr}{\ensuremath{C}}
\newcommand{\kvar}{\ensuremath{\kappa}}
\newcommand{\kvars}{\ensuremath{\mathcal{K}}}
\newcommand{\Cstrs}{\ensuremath{\mathcal{\Cstr}}}

\newcommand{\kvapp}[2]{\ensuremath{#1(#2)}}

\newcommand{\ie}{\textit{i.e.}\xspace}
\newcommand{\eg}{\textit{e.g.}\xspace}
\newcommand{\mypara}[1]{\smallskip\noindent\textbf{#1}}

\newcommand{\cBind}[2]{\ensuremath{{#1}\!:\!{#2}}}
\newcommand{\cAll}[3]{\ensuremath{\forall \cBind{#1}{#2}.\, {#3}}}
\newcommand{\cExi}[3]{\ensuremath{\exists \cBind{#1}{#2}.\, {#3}}}
\newcommand{\cAnd}[2]{\ensuremath{{#1} \wedge {#2}}}
\newcommand{\cImp}[2]{\ensuremath{{#1} \to {#2}}}

\newcommand{\pLeft}{\ensuremath{\mathsf{L}}\xspace}
\newcommand{\pRight}{\ensuremath{\mathsf{R}}\xspace}
\newcommand{\pAsm}[1]{\ensuremath{\mathsf{G}}\xspace}
\newcommand{\pBind}[2]{\ensuremath{\mathsf{B}}\xspace}

\newcommand{\depsname}{\myopname{Dep}}
\newcommand{\asgn}{\ensuremath{\sigma}\xspace}

\newcommand{\cSat}[2]{\ensuremath{#1 \models #2}}
\newcommand{\deps}[1]{\ensuremath{\depsname(#1)}}
\newcommand{\kvs}[1]{\ensuremath{\KVars(#1)}}

\newcommand{\dg}[1]{\ensuremath{\myopname{Graph}(#1)}}
\newcommand{\KVars}{\ensuremath{\myopname{KVars}}}

\newcommand{\kcyc}[1]{\ensuremath{{#1}_{C}}}
\newcommand{\kacy}[1]{\ensuremath{{#1}_{A}}}

\newcommand{\kcycs}[1]{\ensuremath{\overline{{#1}_{C}}}}
\newcommand{\kacys}[1]{\ensuremath{\overline{{#1}_{A}}}}
\newcommand{\pfid}{\ensuremath{\lambda \hyp.\; \hyp}}

\newcommand{\ProofTy}{\ensuremath{\mathit{Proof}}}

\newcommand{\PrefixTy}{\ensuremath{\mathit{Prefix}}}
\newcommand{\CtxTy}{\ensuremath{\mathit{Ctx}}}
\newcommand{\ctx}{\ensuremath{\Gamma}}

\newcommand{\Zapname}{\myopname{Zap}}
\newcommand{\Zapsname}{\myopname{Zaps}}
\newcommand{\Zaponame}{\myopname{Zap1}}
\newcommand{\Zap}[1]{\ensuremath{\Zapname({#1})}}
\newcommand{\Zaps}[1]{\ensuremath{\Zapsname({#1})}}
\newcommand{\Zapo}[2]{\ensuremath{\Zaponame({#1}, {#2})}}
\newcommand{\pf}{\ensuremath{q}}
\newcommand{\hyp}{\ensuremath{h}}
\newcommand{\Elimname}{\myopname{Elim}}
\newcommand{\Elim}[2]{\ensuremath{\Elimname(#1, #2)}}

\newcommand{\ecstr}[2]{\ensuremath{\exists \overline{#1}.\, {#2}}}

\newcommand{\solopI}{\myopname{Sol}}

\newcommand{\Certname}{\ensuremath{\walkop^{\kvar}_{\asgn,\pfx}}\xspace}
\newcommand{\Cert}[3]{\Certname(#1, #2, #3)}

\newcommand{\Partition}[2]{\ensuremath{\myopname{Partition}(#1, #2)}}
\newcommand{\pfx}{\ensuremath{\rho}}
\newcommand{\Headname}{\ensuremath{\downarrow}}
\newcommand{\Ctxname}{\ensuremath{\uparrow}}
\newcommand{\HeadOf}[2]{\ensuremath{{#1}\!\Headname\! {#2}}}
\newcommand{\CtxOf}[2]{\ensuremath{{#1}\!\Ctxname\!{#2}}}
\newcommand{\GoalOf}[2]{\ensuremath{\myopname{Goal}({#1}, {#2})}}
\newcommand{\Solname}{\solopI}
\newcommand{\Sol}[2]{\ensuremath{\Solname}(#1, #2)}
\newcommand{\Heads}[1]{\ensuremath{\myopname{heads}(#1)}}

\newcommand{\scopeop}{\myopname{Scope}}
\newcommand{\Scopename}{\scopeop}
\newcommand{\Scope}[2]{\ensuremath{\Scopename(#1, #2)}}
\newcommand{\Remname}{\myopname{Red}}
\newcommand{\Rem}[3]{\ensuremath{\Remname(#1, #2, #3)}}

\newcommand{\Rank}[1]{\myopname{rank}(#1)}
\newcommand{\Binders}[1]{\myopname{dom}(#1)}
\newcommand{\subasgn}[3]{\ensuremath{#1\sub{#2}{#3}}}
\newcommand{\cLam}[2]{\ensuremath{\lambda{#1}.{#2}}}
\newcommand{\Navname}{\myopname{Nav}}
\newcommand{\Nav}[2]{\ensuremath{\Navname}(#1, #2)}

\newcommand{\tFlex}{\lean{flex}\xspace}
\newcommand{\invisible}[1]{{\color{white}#1}}
\newcommand{\Fixname}{\myopname{Fix}}
\newcommand{\QualTy}{\ensuremath{\mathcal{Q}\xspace}}

\newcommand{\Fix}[2]{\ensuremath{\Fixname(#1, #2)}}

\newcommand{\Certsname}{\ensuremath{\walkop_{\asgns}}\xspace}
\newcommand{\Certs}[3]{\ensuremath{\Certsname(#1, #2, #3)}}

\newcommand{\PredAbsname}{\myopname{PredAbs}}
\newcommand{\PredAbs}[2]{\ensuremath{\PredAbsname(#1, #2)}}

\newcommand{\Fixpointname}{\myopname{Fixpoint}}
\newcommand{\Fixpoint}[2]{\ensuremath{\Fixpointname(#1, #2)}}

\newcommand{\Initname}{\myopname{Init}}
\newcommand{\Init}[2]{\ensuremath{\Initname(#1, #2)}}

\newcommand{\qual}{q}
\newcommand{\quals}{Q}
\newcommand{\asgns}{\ensuremath{\theta}}
\newcommand{\conjoinsol}[1]{\ensuremath{\bigwedge #1}}

\newcommand{\QualAsgnTy}{\ensuremath{\mathcal{A}}}

\newcommand{\Oraclename}{\myopname{Oracle}\xspace}
\newcommand{\Oracle}[2]{\ensuremath{\Oraclename(#1 \vdash #2)}}
\newcommand{\inst}{\ensuremath{\iota}}

\newcommand{\Weakenname}{\myopname{Weaken}}
\newcommand{\Weaken}[3]{\ensuremath{\Weakenname(#1, #2, #3)}}
\newcommand{\kword}[1]{\mbox{\textbf{#1}}}
\newcommand{\Rems}[2]{\ensuremath{\Remname^{*}(#1, #2)}}
\newcommand{\kqinst}[2]{\ensuremath{{#1}{/}{#2}}}
\newcommand{\appasgn}[2]{\ensuremath{#1[#2]}}

\newcommand{\cFRM}[2][gray!20]{\text{\colorbox{#1}{\ensuremath{#2}}}}

\newcommand{\suite}[1]{\texttt{#1}}
\newcommand{\fluxmedium}{\suite{flux-medium}\xspace}
\newcommand{\wavesuite}{\suite{wave}\xspace}
\newcommand{\fluxtests}{\suite{flux-tests}\xspace}
\newcommand{\lfsuite}{\suite{liquid-fix}\xspace}

\definecolor{kInvColor}{HTML}{006D77}
\definecolor{kIdxColor}{HTML}{9C5700}
\newcommand\kinv{\ensuremath{\textcolor{kInvColor}{\kvar_{\mathsf{inv}}}}\xspace}

\newcommand\ctop{\ensuremath{\cstr_{\mathit{top}}}\xspace}
\newcommand\cbody{\ensuremath{{\color{red}\cstr_{\mathit{body}}}}\xspace}

\newcommand\kconst{\ensuremath{\kvar}\xspace}
\newcommand\kconstb{\ensuremath{\kvar'}\xspace}
\newcommand\kwid{\ensuremath{\textcolor{kIdxColor}{\kvar_{\mathsf{w}}}}\xspace}
\newcommand\cjoin{\ensuremath{{\color{blue}\cstr_{\mathit{join}}}}\xspace}
\newcommand\ulen{\mathit{utf8\_len}}
\newcommand\umax{\mathsf{MAX}}
\newcommand\acc{\mathit{acc}}

\newcommand{\figcstrsyntax}{\begin{figure}[t]
\centering
\[
\begin{array}{l@{\quad}r@{\quad}r@{\quad}l@{\quad}l}

\multicolumn{5}{c}{
  \textbf{\textit{Lean Term}} \quad \leanterm,\termprop \in  \Terms
  \quad \quad \quad
  \textbf{\textit{Variables}} \quad x,h \in  \Vars
  \quad \quad \quad
  \textbf{\textit{$\kvar$-variable}} \quad \kvar \in  \kvars\subset \Vars
}\\[0.2em]

\textbf{\textit{Sort}}              & b         & ::=  & \multicolumn{2}{@{}l}{\Prop \mid \Int \mid \Nat \mid
                                                                               \Bool \mid \BitVec{n} \mid
                                                                               \cdots} \\[0.2em]
\textbf{\textit{Atoms}}             & \pred  & ::=  & \termprop    \mid  \kvapp{\kvar}{\overline{t}} \\[0.2em]              

\textbf{\textit{Constraints}}       & \cstr  & ::=  & \pred                       
\mid  \cAnd{\cstr}{\cstr}   
\mid  \cImp{\pred}{\cstr}
\mid  \cAll{x}{b}{\cstr}          
\\[0.2em]

\textbf{\textit{Closed Constraint}} & \Cstr  & ::=  & \exists \overline{\kappa}. \cstr & \\[0.2em]
\textbf{\textit{Proof term}} & \pf & ::= &
                    \multicolumn{2}{@{}l} {\leanterm \mid \lambda x.\,\pf \mid
                                           \pf~\pf \mid
                                           \langle \pf,\, \pf\rangle \mid
                                           \pf.1 \mid \pf.2 \mid \mathsf{inl}\,\pf \mid
                                           \mathsf{inr}\,\pf } \\[0.1em]
    &      & \mid & \langle \leanterm,\, \pf\rangle \mid
             \letin{\langle x,\, \hyp\rangle = \pf}{\pf} \mid
             \rfl \mid \langle\rangle 
\end{array}
\]
\caption{Grammar of the CHC fragment of Lean 4 propositions and their proof-terms as emitted by our solvers.}
\label{fig:chc-grammar}
\label{fig:chc-syntax}
\label{fig:zap-proofterms}
\end{figure}
}

\newcommand{\figscope}{
\begin{figure}[t]
\small
\begin{minipage}[t]{0.40\textwidth}
\vspace{0pt}
$\begin{array}{@{}l@{\;\;}c@{\;\;}l@{}}
\toprule
\multicolumn{3}{@{}l}{\Scopename \;:\; (\kvars \times \Cstrs) \to \PrefixTy} \\[0.1em]
\midrule
\Scope{\kvar}{\cAll{x}{b}{\cstr}}                     & \doteq & \pBind{x}{b}  \cdot \Scope{\kvar}{\cstr}   \\
\Scope{\kvar}{\cAnd{\cstr_1}{\cstr_2}} & & \\
\quad\mid\; \kvar \in \cstr_1,\; \kvar \notin \cstr_2 & \doteq & \pLeft  \cdot \Scope{\kvar}{\cstr_1} \\
\quad\mid\; \kvar \notin \cstr_1,\; \kvar \in \cstr_2 & \doteq & \pRight \cdot \Scope{\kvar}{\cstr_2} \\
\Scope{\kvar}{\cImp{\pred}{\cstr}}     &  &    \\
\quad\mid\; \kvar \notin \pred                        & \doteq & \pAsm{\pred}   \cdot \Scope{\kvar}{\cstr}   \\
\Scope{\kvar}{\cstr}                                  & \doteq & \varepsilon   \\
\bottomrule
\end{array}$
\end{minipage}
\begin{minipage}[t]{0.28\textwidth}
\vspace{0pt}
$\begin{array}{@{}l@{\;\;}c@{\;\;}l@{}}
\toprule
\multicolumn{3}{@{}l}{(\Headname) \;:\; (\Cstrs \times \PrefixTy) \to \Cstrs} \\[0.1em]
\midrule
\HeadOf{(\cAll{x}{b}{\cstr})}{\pBind{x}{b} \cdot \pfx}  & \doteq & \HeadOf{\cstr}{\pfx} \\
\HeadOf{(\cImp{\pred}{\cstr})}{\pAsm{\pred} \cdot \pfx} & \doteq & \HeadOf{\cstr}{\pfx} \\
\HeadOf{(\cAnd{\cstr_1}{\cstr_2})}{\pLeft \cdot \pfx}   & \doteq & \HeadOf{\cstr_1}{\pfx} \\
\HeadOf{(\cAnd{\cstr_1}{\cstr_2})}{\pRight \cdot \pfx}  & \doteq & \HeadOf{\cstr_2}{\pfx} \\
\HeadOf{\cstr}{\varepsilon}                             & \doteq & \cstr \\[2.5em]
\bottomrule
\end{array}$
\end{minipage}
\begin{minipage}[t]{0.30\textwidth}
\vspace{0pt}
$\begin{array}{@{}l@{\;\;}c@{\;\;}l@{}}
\toprule
\multicolumn{3}{@{}l}{(\Ctxname) \;:\; (\Cstrs \times \PrefixTy) \to \CtxTy} \\[0.1em]
\midrule
\CtxOf{(\cAll{x}{b}{\cstr})}{\pBind{x}{b} \cdot \pfx}       & \doteq & \cBind{x}{b}; \CtxOf{\cstr}{\pfx} \\
\CtxOf{(\cImp{\pred}{\cstr})}{\pAsm{\pred} \cdot \pfx}       & \doteq & \cBind{h}{\pred}; \CtxOf{\cstr}{\pfx} \\
\CtxOf{(\cAnd{\cstr_1}{\cstr_2})}{\pLeft \cdot \pfx}  & \doteq & \pLeft;  \CtxOf{\cstr_1}{\pfx} \\
\CtxOf{(\cAnd{\cstr_1}{\cstr_2})}{\pRight \cdot \pfx} & \doteq & \pRight; \CtxOf{\cstr_2}{\pfx} \\
\CtxOf{\cstr}{\varepsilon}                             & \doteq &             \varepsilon \\[2.5em]
\bottomrule
\end{array}$
\end{minipage}
\caption{
$\Scope{\kvar}{\cstr}$ is the scope of a $\kvar$ in $\cstr$;
$\HeadOf{\cstr}{\pfx}$ is the \emph{head of} $\cstr$ at $\pfx$ and
$\CtxOf{\cstr}{\pfx}$ is the \emph{context of} $\cstr$ at $\pfx$.
In \Ctxname, each guard hypothesis gets a fresh name $\hyp$.}
\label{fig:deps}
\label{fig:ctx}
\label{fig:head}
\label{fig:scope}
\end{figure}
}

\newcommand{\figzapelim}{
\begin{figure}[t]
\begin{minipage}[t]{0.50\textwidth}
\vspace{0pt}
\small
$\begin{array}{@{}l@{\;\;}c@{\;\;}l@{}}
\toprule
\multicolumn{3}{@{}l}{\Zapname :  \Cstrs \to (\Cstrs \times \ProofTy)} \\
\midrule
\Zap{\ecstr{\kvar}{\cstr}}   & \doteq & (\ecstr{\kcyc{\kvar}}{\cstr'}, \pf' \textcolor{gray}{: \ecstr{\kcyc{\kvar}}{\cstr'} \to \ecstr{\kvar}{\cstr}}) \\
\quad \kword{where}           &        & \\
\quad \quad (\cstr', \pf)    & \doteq & \Zaps{\ecstr{\kacy{\kvar}}{\cstr}} \\
\quad \quad \pf'  & \doteq & \cLam{\hyp_1}{\letin{\langle \kcycs{\kvar}, \hyp_2\rangle = \hyp_1}{}} \\
\quad \quad                  &        & \invisible{\lambda \hyp_1.} \letin{\langle \kacys{\kvar}, \hyp_3\rangle = \pf\ \hyp_2}{} \\
\quad \quad                  &        & \invisible{\lambda \hyp_1.} \langle \overline{\kvar},\, \hyp_3\rangle \\[0.5em]
\midrule
\multicolumn{3}{@{}l}{\Zapsname : \Cstrs \to (\Cstrs \times \ProofTy)} \\
\midrule
\Zaps{\exists \kvar.\, \cstr}  & \doteq & (\cstr_2, \pf_3 \textcolor{gray}{: \cstr_2 \to \exists \kvar.\,\cstr} ) \\
\quad \kword{where}             &        & \\
\quad \quad (\cstr_1, \pf_1)      & \doteq & \Zaps{\cstr} \\
\quad \quad (\cstr_2, \pf_2)    & \doteq & \Zapo{\kvar}{\cstr_1} \\
\quad \quad \pf_3  & \doteq & \cLam{\hyp_1}{\letin{\langle \kvar, \hyp_2\rangle = \pf_2\ \hyp_1}{}} \\
\quad \quad                    &        & \invisible{\lambda \hyp_1.}\langle \kvar,\, \pf_1\ \hyp_2\rangle \\
\Zaps{\cstr}                   & \doteq & (\cstr, \pfid) \\[0.1em]
\midrule
\multicolumn{3}{@{}l}{\Zaponame : (\kvars \times \Cstrs) \to (\Cstrs \times \ProofTy)} \\
\midrule
\Zapo{\kvar}{\cstr}                               & \doteq & (\cstr', \pf \textcolor{gray}{: \cstr' \to \exists\kvar.\,\cstr} ) \\
\quad \kword{where}                               &        & \\
\quad \quad (\asgn, \pfx, \cstr')                 & \doteq & \Elim{\kvar}{\cstr} \\
\quad \quad \pf  & \doteq & \cLam{\hyp}{\langle \asgn,\ \Cert{\varepsilon}{\hyp}{\cstr} \rangle}\\[2.35em]
\bottomrule
\end{array}$
\end{minipage}
\begin{minipage}[t]{0.49\textwidth}
\vspace{0pt}
\small
$\begin{array}{@{}l@{\;\;}c@{\;\;}l@{}}
\toprule
\multicolumn{3}{@{}l}{\Elimname : (\kvars \times \Cstrs) \to (\Preds \times \PrefixTy \times \Cstrs)} \\
\midrule
\Elim{\kvar}{\cstr}          & \doteq & (\asgn, \pfx, \cstr') \\
\quad \kword{where}           &        & \\
\quad \quad \pfx             & \doteq & \Scope{\kvar}{\cstr} \\
\quad \quad \asgn            & \doteq & \lambda \overline{z}.~\lambda \overline{x}.~\Sol{\kvar}{\HeadOf{\cstr}{\pfx}} \\
\quad \quad \cstr'           & \doteq & \Rem{\kvar}{\asgn}{\cstr}  \\
\quad \quad \overline{z}     & \doteq & z_0, \ldots, z_{n-1} \ \kword{where}\ n = \Rank{\kvar}  \\
\quad \quad \overline{x}     & \doteq & \Binders{\CtxOf{\cstr}{\pfx}}
\\[0.5em]
\midrule
\multicolumn{3}{@{}l}{\Remname : (\kvars \times \Preds \times \Cstrs) \to \Cstrs} \\
\midrule
\Rem{\kvar}{\asgn}{\cAnd{\cstr_1}{\cstr_2}}     & \doteq & \Rem{\kvar}{\asgn}{\cstr_1}\ \wedge \\
                                                &        & \quad \Rem{\kvar}{\asgn}{\cstr_2} \\
\Rem{\kvar}{\asgn}{\cAll{x}{b}{\cstr}}          & \doteq & \cAll{x}{b}{\Rem{\kvar}{\asgn}{\cstr}} \\
\Rem{\kvar}{\asgn}{\cImp{\pred}{\cstr}}         & \doteq & \cImp{\subasgn{\pred}{\kvar}{\asgn}}{\Rem{\kvar}{\asgn}{\cstr}} \\
\Rem{\kvar}{\asgn}{\kvapp{\kvar}{\overline{t}}} & \doteq & \top \\
\Rem{\kvar}{\asgn}{\pred}                       & \doteq & \pred \\
\midrule
\multicolumn{3}{@{}l}{\Solname : (\kvars \times \Cstrs) \to \Terms} \\
\midrule
\Sol{\kvar}{\cstr} & & \\
\quad\mid\; \kvar \notin \Heads{\cstr}    & \doteq & \bot \\
\Sol{\kvar}{\cAnd{\cstr_1}{\cstr_2}}      & \doteq & \Sol{\kvar}{\cstr_1} \vee \Sol{\kvar}{\cstr_2} \\
\Sol{\kvar}{\cAll{x}{b}{\cstr}}           & \doteq & \cExi{x}{b}{\Sol{\kvar}{\cstr}} \\
\Sol{\kvar}{\cImp{\pred}{\cstr}}          & \doteq & \cAnd{\pred}{\Sol{\kvar}{\cstr}} \\
\Sol{\kvar}{\kvapp{\kvar}{\overline{t}}}  & \doteq & \textstyle\bigwedge_{i < \Rank{\kvar}} z_i = t_i \\
\bottomrule
\end{array}$
\end{minipage}
\caption{\Zapname\ eliminates acyclic horn variables (left) using \Elimname\ to compute solutions (right).
}
\label{fig:algo:zap}
\label{fig:algo:elim}
\label{fig:algo:sol}
\label{fig:algo:rem}
\end{figure}
}

\newcommand{\figcert}{
\begin{figure}[t]
\small
\begin{minipage}[t]{0.48\textwidth}
\[
\begin{array}{@{}l@{\;\;}c@{\;\;}l@{}}
\toprule
\Certname & \;:\; & (\CtxTy \times \ProofTy \times \Cstrs) \to \ProofTy \\
\midrule
\Cert{\ctx}{\pf}{\cAll{x}{b}{\cstr}}         & \doteq & \cLam{x}{\Cert{\ctx; \cBind{x}{b}}{\pf\ x}{\cstr}}  \\
\Cert{\ctx}{\pf}{\cImp{\pred}{\cstr}}        & \doteq & \cLam{\hyp_\pred}{\Cert{\ctx; \cBind{\hyp_\pred}{\pred}}{\pf\ \hyp_\pred}{\cstr}}\\
\Cert{\ctx}{\pf}{\cAnd{\cstr_1}{\cstr_2}}    & \doteq & \langle\ \Cert{\ctx; \pLeft}{\pf.1}{\cstr_1} \\
                                             &        & ,\ \Cert{\ctx; \pRight}{\pf.2}{\cstr_2}\ \rangle \\
\Cert{\ctx}{\pf}{\kvapp{\kvar}{\overline{t}}}& \doteq & \Nav{\ctx \setminus \pfx}{\asgn(\overline{t})} \\
\Cert{\ctx}{\pf}{\pred}                      & \doteq & \pf \\[0.2em]
\bottomrule
\end{array}
\]
\end{minipage}
~
\begin{minipage}[t]{0.48\textwidth}
\[
\begin{array}{@{}l@{\;\;}c@{\;\;}l@{}}
\toprule
\Navname & \;:\; & (\CtxTy \times \Preds) \to \ProofTy \\[0.04em]
\midrule
\Nav{\pLeft;  \ctx}{\pred_L \vee \pred_R}     & \doteq & \mathsf{inl}~\Nav{\ctx}{\pred_L}\\
\Nav{\pRight; \ctx}{\pred_L \vee \pred_R}     & \doteq & \mathsf{inr}~\Nav{\ctx}{\pred_R}\\
\Nav{\cBind{x}{b}; \ctx}{\cExi{z}{b}{\pred}}  & \doteq & \big\langle x,\; \Nav{\ctx}{\subasgn{\pred}{z}{x}}\big\rangle \\
\Nav{\cBind{\hyp}{\pred}; \ctx}{\pred \wedge \pred'} & \doteq & \big\langle \hyp,\; \Nav{\ctx}{\pred'}\big\rangle \\
\Nav{\varepsilon}{\big(\textstyle\bigwedge_i z_i = t_i\big)} & \doteq & \langle\, \rfl, \dots, \rfl \,\rangle \\
\Nav{\varepsilon}{\top}                       & \doteq & \langle\rangle \\[0.45em]
\bottomrule
\end{array}
\]

\end{minipage}
\caption{\Certname\ transforms a proof of the reduced constraint into a proof of the
original, invoking \Navname\ at each divergent $\kvar$-head to rebuild a proof of
$\asgn(\overline{t})$ by replaying the truncated context against the structure of $\asgn$.}
\label{fig:walk}
\label{fig:certify}
\label{fig:nav}
\end{figure}
}

\newcommand{\figpiglet}{
\begin{figure}[t]
\small
\begin{minipage}[t]{0.49\textwidth}
\vspace{0pt}
$\begin{array}{@{}l@{\;}c@{\;}l@{}}
\toprule
\multicolumn{3}{@{}l}{\Fixname : (\Cstrs \times \QualTy) \to (\Cstrs \times \ProofTy)} \\
\midrule
\Fix{\ecstr{\kcyc{\kvar}}{\cstr}}{\quals} & \doteq & (\cstr', \pf\textcolor{gray}{: \cstr' \to \ecstr{\kcyc{\kvar}}{\cstr}}) \\
\quad \kword{where}     &        & \\
\quad \quad \asgns      & \doteq & \PredAbs{\ecstr{\kcyc{\kvar}}{\cstr}}{\quals} \\
\quad \quad \cstr'      & \doteq & \Rems{\kcycs{\kvar}}{\cstr} \\
\quad \quad \pf                                  & \doteq & \cLam{\hyp}{\langle \overline{\conjoinsol{\asgns(\kcyc{\kvar})}},\ \Certs{\varepsilon}{\hyp}{\cstr} \rangle} \\
\quad \quad \Rems{\varepsilon}{\cstr}            & \doteq & \cstr \\
\quad \quad \Rems{\kvar;\overline{\kvar}}{\cstr} & \doteq & \Rems{\overline{\kvar}}{\Rem{\kvar}{\conjoinsol{\asgns(\kvar)}}{\cstr}} \\[0.5em]
\toprule
\multicolumn{3}{@{}l}{\Certsname : (\CtxTy \times \ProofTy \times \Cstrs) \to \ProofTy} \\
\midrule
\Certs{\ctx}{\pf}{\cAll{x}{b}{\cstr}}         & \doteq & \cLam{x}{\Certs{\ctx; \cBind{x}{b}}{\pf\ x}{\cstr}}  \\
\Certs{\ctx}{\pf}{\cImp{\pred}{\cstr}}        & \doteq & \cLam{\hyp_\pred}{\Certs{\ctx; \cBind{\hyp_\pred}{\pred}}{\pf\ \hyp_\pred}{\cstr}}\\
\Certs{\ctx}{\pf}{\cAnd{\cstr_1}{\cstr_2}}    & \doteq & \langle\ \Certs{\ctx; \pLeft}{\pf.1}{\cstr_1} \\
                                              &        & ,\ \Certs{\ctx; \pRight}{\pf.2}{\cstr_2}\ \rangle \\
\Certs{\ctx}{\pf}{\kvapp{\kvar}{\overline{t}}}& \doteq & \Oracle{\appasgn{\ctx}{\asgns}}{(\conjoinsol{\asgns(\kvar)})(\overline{t})} \\
\Certs{\ctx}{\pf}{\pred}                      & \doteq & \pf \\
\bottomrule
\end{array}$
\end{minipage}\begin{minipage}[t]{0.50\textwidth}
\vspace{0pt}
$\begin{array}{@{}l@{\;}c@{\;}l@{}}
\toprule
\multicolumn{3}{@{}l}{\PredAbsname : (\Cstrs \times \QualTy) \to \QualAsgnTy} \\
\midrule
\PredAbs{\ecstr{\kcyc{\kvar}}{\cstr}}{\quals} & \doteq & \Fixpoint{\cstr}{\asgns} \\
\quad \kword{where}                    &        &  \\
\quad \quad \asgns                     & \doteq & [ \kvar \mapsto \Init{\kvar}{\quals} \mid \kvar \in \kcycs{\kvar}]  \\
\quad \quad \Init{\kvar}{\quals}       & \doteq & \{ q_\inst \mid q \in \quals, \inst \in \kqinst{\kvar}{q} \ \mbox{instances} \} \\[0.5em]
\midrule
\multicolumn{3}{@{}l}{\Fixpointname : (\Cstrs \times \QualAsgnTy) \to \QualAsgnTy} \\
\midrule
\Fixpoint{\cstr}{\asgns} & \doteq & \kword{if}\ \exists{\pfx}. \asgns' \doteq \Weaken{\cstr}{\asgns}{\pfx} \\
                         &        & \quad \kword{then}\ \Fixpoint{\cstr}{\asgns'}\ \kword{else}\ \asgns \\[0.5em]
\midrule
\multicolumn{3}{@{}l}{\Weakenname : (\Cstrs \times \QualAsgnTy \times \PrefixTy) \to \QualAsgnTy + \bot} \\
\midrule
\Weaken{\cstr}{\asgns}{\pfx}       & \doteq & \kword{if}\ I \subset \asgns(\kvar)\ \kword{then}\ \asgns[\kvar \mapsto I]\ \kword{else}\ \bot\\
\multicolumn{3}{@{}l}{\quad \kword{where}} \\
\multicolumn{3}{@{}l}{\quad \quad
  \begin{array}[t]{@{}l@{\;}c@{\;}l@{}}
  \ctx \vdash \kvapp{\kvar}{\overline{t}} & \doteq & \GoalOf{\cstr}{\pfx} \\
  I                                       & \doteq & \{ q_\inst \mid q_\inst \in \asgns(\kvar), \Oracle{\appasgn{\ctx}{\asgns}}{q_\inst(\overline{t})} \}
  \end{array}} \\[2.0em]
\bottomrule
\end{array}$
\end{minipage}
\caption{The $\Fixname$ procedure for reducing cyclic horn variables via predicate abstraction.}
\label{fig:fix}
\label{fig:predabs}
\label{fig:fixpoint}
\label{fig:certs}
\end{figure}
}

\newcommand{\figoverviewinstantiate}{\begin{figure}[t]
\centering\small
$\begin{array}{c@{\;}c@{\;}c@{\;}c@{\;}c}
\begin{array}{l}
\exists \kinv, \kwid.\ \cdots \\
\invisible{\wedge}\ \cdots \\
\wedge\ (128 \le m[i] < 2048 \rightarrow \\
\quad \quad \kvapp{\kwid}{2, \acc, i, m, n}) \\
\wedge\ \cdots \\
\wedge\ \forall w.\ \kvapp{\kwid}{w, \acc, i, m, n} \\
\quad \quad \rightarrow \ldots \\
\end{array}
&\xRightarrow[{\kwid := \asgn}]{\lean{zap}} &
\begin{array}{l}
\exists \kinv.\ \cdots \\
\invisible{\wedge}\ \cdots \\
\wedge\ (128 \le m[i] < 2048 \rightarrow \\
\quad \quad \cFRM{\kvapp{\asgn}{2, \acc, i, m, n}}) \\
\wedge\ \cdots \\
\wedge\ \forall w.\ \cFRM{\kvapp{\asgn}{w, \acc, i, m, n}} \\
\quad \quad \rightarrow \ldots \\
\end{array}
&\xRightarrow[]{\text{witness}}&
\begin{array}{l}
\exists \kinv.\ \cdots \\
\invisible{\wedge}\ \cdots \\
\wedge\ (128 \le m[i] < 2048 \rightarrow \\
\quad \quad \cFRM{\top}) \\
\wedge\ \cdots \\
\wedge\ \forall w.\ \cFRM{\kvapp{\asgn}{w, \acc, i, m, n}} \\
\quad \quad \rightarrow \ldots \\
\end{array}
\end{array}$
\caption{Instantiating $\kwid$ substitutes occurrences;
  applying proof-terms eliminates the heads.}
\label{fig:overview:instantiatenew}
\label{fig:overview:zap}
\end{figure}}

\newcommand{\figoverviewfix}{
\begin{figure}[t]
\centering\small
$\begin{array}{c@{\;}c@{\;}c@{\;}c@{\;}c}
\begin{array}{@{}l}
\exists \kinv.\ \cdots \\
\invisible{\wedge}\ \kvapp{\kinv}{0, 0, m, n} \\
\wedge\ \forall \acc, i.\ \kvapp{\kinv}{\acc, i, m, n} \rightarrow \\
\quad \quad \invisible{\wedge}\ (n \le i \rightarrow \acc = \kvapp{\ulen}{m, n}) \\
\quad \quad \wedge\ (i < n \rightarrow \cdots \\
\quad \quad \quad \wedge\ \forall w.\ \kvapp{\asgn}{w, \acc, i, m, n} \rightarrow \\
\quad \quad \quad \quad  \quad \invisible{\wedge}\ \acc + w \le \umax \\
\quad \quad \quad \quad  \quad \wedge\ \kvapp{\kinv}{\acc+w, i+1, m, n})
\end{array}
&
\xRightarrow[\substack{\kinv := \asgn'}]{\lean{fix}}
&
\begin{array}{l}
\cdots \\
\invisible{\wedge}\ \cFRM{\top} \\
\wedge\ \forall \acc, i. \cFRM{\kvapp{\asgn'}{\acc, i, m, n}} \rightarrow \\
\quad \quad \invisible{\wedge}\ (n \le i \rightarrow \acc = \kvapp{\ulen}{m, n}) \\
\quad \quad \wedge\ (i < n \rightarrow \cdots \\
\quad \quad \quad \wedge\ \forall w.\ \kvapp{\asgn}{w, \acc, i, m, n} \rightarrow \\
\quad \quad \quad \quad  \quad \invisible{\wedge}\ \acc + w \le \umax \\
\quad \quad \quad \quad  \quad \wedge\  \cFRM{\top})
\end{array}
&\xRightarrow[]{\lean{grind}}&
\acc + w \le \umax
\end{array}$
\caption{(L) The CHC for \inlineflux{encoded_len} after \lean{zap} (\Cref{fig:overview:zap}),
(C) \lean{fix} uses qualifiers to reduce the \emph{cyclic} $\kinv$, replacing head
occurrences with $\top$ and body occurrences with the solution, (R)
\lean{grind} discharges every residual obligation except the overflow check proven by
an inductive lemma in \Cref{fig:overview:proof}.}
\label{fig:overview:fix}
\end{figure}
}
 \setcopyright{none}
\Crefname{section}{\S}{\SS}
\Crefname{subsection}{\S}{\SS}
\Crefname{subsubsection}{\S}{\SS}
\Crefname{figure}{Fig.}{Fig.}
\crefname{figure}{Fig.}{Fig.}
\Crefname{equation}{Eq.}{Eq.}

\begin{document}

\title{\textsf{Flex}: A Foundational Proof-Producing Horn Solver}

\author{Jam Kabeer Ali Khan}
\orcid{0009-0007-8505-9667}
\affiliation{
  \institution{Independent}
  \country{Hong Kong SAR}
}
\email{jamkhan@connect.hku.hk}

\author{Petros Markopoulos}
\orcid{0009-0009-0370-5712}
\affiliation{
  \institution{University of California San Diego}
  \city{San Diego}
  \country{USA}
}
\email{pmarkopoulos@ucsd.edu}

\author{Niki Vazou}
\orcid{0000-0003-0732-5476}
\affiliation{
  \institution{IMDEA Software Institute}
  \city{Madrid}
  \country{Spain}
}
\email{niki.vazou@software.imdea.org}

\author{Nico Lehmann}
\orcid{0009-0003-6838-3714}
\affiliation{
  \institution{University of Chile}
  \city{Santiago}
  \country{Chile}
}
\email{nlehmann@dcc.uchile.cl}

\author{Ranjit Jhala}
\orcid{0000-0002-1802-9421}
\affiliation{
  \institution{University of California San Diego}
  \city{San Diego}
  \country{USA}
}
\email{rjhala@ucsd.edu}

\renewcommand{\shortauthors}{Khan et al.}

\begin{abstract}
Constrained Horn Clauses (CHCs) are a versatile intermediate 
representation for program verification, but leading CHC solvers 
rely on SMT-based reasoning over supported background theories. 
Specifications involving user-defined recursive functions or 
bespoke operators can require additional semantic encodings and 
reasoning beyond these theories.
We observe that CHCs are existentially quantified propositions
in a foundational logic, and so can be proved within an
interactive theorem prover, where specifications can use
native definitions and users can complete proofs when
automation stalls.
However, in a foundational setting, every Horn variable
solution must be justified by a kernel-checked proof.
We present \sys, a proof-producing CHC solver implemented
as a \sclean metaprogram, that combines predicate synthesis
with proof construction.
\sys comprises a suite of tactics:
\Zapname, which synthesizes strongest predicates for acyclic
Horn variables and exploits their structure to construct
elimination proofs without any proof search;
\Fixname, which uses predicate abstraction over templates
built from arbitrary \sclean terms, discovered by \emph{scraping}
the CHCs and associated definitions, to infer inductive invariants
for cyclic Horn variables;
and a certifying pipeline that keeps external SMT-based solvers
like Spacer and Eldarica outside the trusted computing base.
We use \sys to verify the functional correctness of
sorting algorithms, hash tables, ring buffers from the
Tock OS kernel, and Rust \iflux{std} library challenges.
We show that \Zapname constructs elimination proofs
$28.2\times$ and $44.6\times$ faster than \lean{grind}
and \lean{aesop}, respectively.
With \Zapname and \Fixname, \sys automatically proves
98.0\% of the refinement-typechecker CHCs in our evaluation.
\end{abstract}
 

\maketitle

\section{Introduction}
\label{sec:intro}

Constrained Horn Clauses \cite{flanagan2004automatic, podelski2007armc, bjorner2015horn}
have proven to be a versatile intermediate representation
for encoding program verification conditions.
Horn variables can be used to represent the unknown
invariants or intermediate assertions, required to
prove some verification goal, after which SMT-based
algorithms like Spacer's Property Directed Reachability~\cite{Spacer1}
or Eldarica's Lazy Predicate Abstraction~\cite{ELDARICA}
can be used to solve these CHCs to automate verification.
Their versatility and automation has made CHCs the backend
of choice for a variety of verifiers based on program logics
for C/C++ \cite{SeaHorn}, or refinement types for Rust
\cite{FluxRS,matsushita2021rusthorn,Thrust}
and Haskell~\cite{LiquidHaskell}.

However, leading CHC solvers rely on SMT-based reasoning
over supported background theories. Specifications involving
recursive functions or bespoke operators can require additional
semantic encodings and reasoning beyond these theories.
For example, consider the verification of a Rust function
in \cref{fig:encodedlen} that computes the number of bytes
needed to encode a character sequence \iflux{cs}.
To \emph{specify} the correctness of the result,
we require a recursive user-defined function
$\ulen(s, k)$ that computes the UTF-8 size
of the first $k$ characters of the sequence $s$.
To then \emph{verify} the function, we require
first, a loop invariant that states that the
accumulator tracks the size of the prefix of
size \iflux{i}, and second, a separate
monotonicity lemma that establishes
the arithmetic does not overflow.
Sadly, these obligations can exceed the reasoning supported
by SMT-based CHC solvers, leaving programmers to supply
explicit loop invariants and proof blocks that clutter
the source code.

One way to overcome the expressiveness limitation is to
observe that CHCs are \emph{existentially quantified propositions}
in a foundational logic, such as that of
\sclean~\cite{Lean4} or \scrocq~\cite{Coq-refman}
or \scisabelle~\cite{Isabelle}.
Thus, \emph{in principle} we might use an ITP
to prove the CHC, instead of restricting ourselves
to SMT solving.
This would allow specifications to refer to native
ITP definitions, and programmers can inspect the CHCs
and intermediate proof states to complete their proofs
when automation stalls.

\mypara{Challenge: Proof-producing CHC Solvers}
Alas, merely translating the CHCs into ITP propositions
is not terribly useful as the user is still left with the
task of solving the Horn variables to prove the CHC.
Thus our first challenge is to \emph{solve} the Horn variables
within the ITP to be able to instantiate them with predicates
that render the CHC valid.
However, the foundational setting poses a second challenge:
unlike in the SMT-based setting where the solvers can report
predicate solutions that the verifier simply trusts, in a
foundational context, any solution whose substitution
eliminates a Horn variable from the constraint
must be justified by a kernel-checked proof.
That is, any solver must be \emph{proof-producing},
in that it must construct a kernel-checkable certificate
witnessing that the validity of the proposition obtained by
eliminating the Horn variable \emph{implies} the validity
of the original CHC.
Since efficient solving may analyze sub-constraints locally,
certification must be aligned with those local analyses and
compose their proofs across eliminations.

\mypara{\sys}
We present \sys, a foundational CHC solver
implemented as a \sclean metaprogram~\cite{Lean4}
that combines predicate synthesis with proof construction
to address the challenge of proof-producing CHC solving.
Following \cite{Cosman17} we observe that the
Horn variables in a constraint can be partitioned
into acyclic ones, that can be solved exactly,
and cyclic ones that require some form of
fixpoint over-approximation.
Thus, we design \sys as a suite of proof-producing
tactics that separately solve the acyclic and cyclic
variables, yielding the following contributions.

\mypara{1. \emph{Proof-producing} Acyclic Horn Solving (\Cref{sec:zap}).}
Our first contribution is \Zapname, a proof-producing
algorithm that synthesizes strongest predicates for
acyclic Horn variables and constructs proofs of
their elimination.
We identify a crucial correspondence between the structure of
constraints, the synthesized predicates, and the elimination
proofs that \Zapname exploits to construct elimination witnesses
deterministically and efficiently, without requiring any proof search.
The construction preserves variable scopes and bindings,
and composes individual eliminations into a certificate
that the residual constraint implies the original.

\mypara{2. Predicate Abstraction over \sclean Functions (\Cref{sec:fixpoint}).}
Our second contribution is \Fixname, a certifying
predicate-abstraction procedure \cite{Graf97} that
infers inductive invariants for cyclic Horn variables.
Following liquid inference~\cite{Rondon08,Houdini},
\Fixname weakens conjunctions of predicate-template
instances until it reaches a fixpoint.
Unlike traditional SMT-based approaches,
\Fixname operates directly within \sclean's expressive
logic and hence can use templates over arbitrary \sclean
terms, crucially enabling synthesis of invariants
over user-defined functions or quantifiers.
To discover qualifiers, we develop \emph{scraping}
algorithms that traverse the CHCs and associated
\sclean definitions to extract candidate invariant
fragments from equalities or predicates in branch
conditions (\Cref{sec:scraping}), to enable \Fixname
to synthesize invariants beyond the theories supported
by current SMT-based solvers.

\mypara{3. Certifying External CHC Solvers (\Cref{sec:external-oracles}).}
Our third contribution is to show how to integrate the existing
state-of-the-art SMT-based solvers --- Spacer and Eldarica ---
through a certifying pipeline that translates \sclean CHCs
into SMT-LIB, reconstructs inferred invariants as \sclean
predicates, and certifies their elimination with respect to
the original constraints.
The obligations that remain after eliminating the Horn variables
remain as ordinary \sclean goals, allowing partial progress via
the CHC oracles while keeping external solvers outside the trusted
computing base.

\mypara{4. Implementation and Evaluation (\Cref{sec:evaluation}).}
Our final contribution is an implementation
of \sys as a \sclean tactic in approximately
9,200 lines, and an empirical evaluation
of \sys on a variety of case studies.\footnote{\sys is available at
\url{https://github.com/LeanHorn/Flex}; benchmarks are available at
\url{https://github.com/LeanHorn/flex-benchmarks}.}
First, we show that \sys' expressiveness
enables the verification of functional
correctness in a suite of case studies
covering implementations of
sorting algorithms, Hash tables, ring buffers
from the Tock OS kernel, and Rust \iflux{std}
library verification challenges including the
correctness of string matching, string encoding
and the correctness of the \iflux{vec-deque} module
(\Cref{subsec:rq1-expr}).
Second, we show that on a suite of existing,
SMT-checkable, refinement typing and CHC-COMP benchmarks,
\Zapname constructs elimination proofs
$28$--$45\times$ faster than \lean{grind}
and \lean{aesop} on jointly reduced constraints (\Cref{sec:rq2}).
Third, we demonstrate 98.0\% coverage
over refinement typing CHCs (\Cref{subsec:rq3-auto}).
Finally, on selected cyclic CHCs, we show that incremental scraping
improves proof coverage while controlling candidate growth and solving cost
(\Cref{subsec:rq4-scraping}).

 \section{An Overview of \sys}
\label{sec:overview}

We begin with an overview, starting with
the motivation behind \sys: solving CHCs in \sclean
generated by the Rust refinement type checker \scflux
(\Cref{sec:overview:motivation}).
Next, we show how predicate synthesis guides proof-producing
elimination of acyclic Horn variables (\Cref{sec:overview:solvers}).
Finally, we describe cyclic invariant inference using native \sclean
predicates and certifying external oracles
(\Cref{sec:overview:expressiveness}).

\begin{figure}[t]
\begin{minipage}[t]{.40\textwidth}
\vspace{0pt}
\begin{fluxsh}
fn encoded_len(cs: &[char][n,m])
-> u32[utf8_len(m,n)]
requires utf8_len(m,n) <= u32::MAX
{
  let mut acc = 0;
  let mut i = 0;
  while i < cs.len() {
    let c = cs[i] as u32;
    let w = if c < 128 { 1 }
      else if c < 2048 { 2 }
      else if c < 65536 { 3 }
      else { 4 };
    acc += w;
    i += 1;
  }
  acc
}
\end{fluxsh}
\end{minipage}
\hfill
\begin{minipage}[t]{.57\textwidth}
\vspace{2pt}
$\begin{array}{@{}r@{\;}c@{\;}l@{}}
\ctop & \doteq &
    \exists \kinv, \kwid. \forall m, n. \\
& & \kvapp{\ulen}{m, n} \le \umax \rightarrow \\
& & \quad \invisible{\wedge}\ \kvapp{\kinv}{0, 0, m, n} \\
& & \quad \wedge\ \forall \acc, i.\ \kvapp{\kinv}{\acc, i, m, n}\ \rightarrow \\
& & \quad \quad \quad \invisible{\wedge}\ (n \le i \rightarrow \acc = \kvapp{\ulen}{m, n}) \\
& & \quad \quad \quad \wedge\ (i < n \rightarrow \cbody) \\[0.85em]
\cbody & \doteq &
    \invisible{\wedge}\ (m[i] < 128 \rightarrow \kvapp{\kwid}{1, \acc, i, m, n}) \\
& & \wedge\ (128 \le m[i] < 2048 \rightarrow \kvapp{\kwid}{2, \acc, i, m, n}) \\
& & \wedge\ (2048 \le m[i] < 65536 \rightarrow \kvapp{\kwid}{3, \acc, i, m, n }) \\
& & \wedge\ (65536 \le m[i] \rightarrow \kvapp{\kwid}{4, \acc, i, m, n}) \\
& & \wedge\ \cjoin \\[0.85em]
\cjoin & \doteq &
    \forall w.\ \kvapp{\kwid}{w, \acc, i, m, n} \rightarrow \\
& & \quad \invisible{\wedge}\ \acc + w \le \umax \\
& & \quad \wedge\ \kvapp{\kinv}{\acc + w, i + 1, m, n} \\
\end{array}$
\end{minipage}
\caption{(L) Rust with \scflux specification, (R) verification CHC generated by \scflux.}
\label{fig:encodedlen}
\end{figure}

\subsection{Motivating Example}
\label{sec:overview:motivation}

\Cref{fig:encodedlen} shows a Rust function \iflux{encoded_len}
that computes the \emph{number of bytes}, represented as a
\iflux{u32}, \ie unsigned 32-bit integers, needed to encode
\iflux{cs}, a slice of Unicode characters.
The function iterates over the slice, accumulating
the sum of the width of each character's encoding
in \iflux{acc}.

\mypara{Specification: Refinement Types.}
To \emph{specify} the correctness requirement,
the programmer writes a \scflux refinement
type contract where the input type
\iflux{&[char][n, m]} represents
the slice by its length \iflux{n}
and contents \iflux{m} which is a
map from slice-indices to the
corresponding character.
The specification uses a programmer-defined
\sclean function $\kvapp{\ulen}{m, k}$ denoting
the number of bytes needed to encode the first
$k$ elements of $m$, defined by recursion on $k$,
\ie $\kvapp{\ulen}{m, 0} \doteq 0$ and
$\kvapp{\ulen}{m, k+1} \doteq \kvapp{\ulen}{m, k} + \kvapp{\mathit{cw}}{m[k]}$,
where $\kvapp{\mathit{cw}}{c}$ is the width of the encoding of $c$.
The output type says that \iflux{encoded_len}
should return exactly $\kvapp{\ulen}{m, n}$,
subject to the precondition that the length
fits in a \iflux{u32}.

\begin{figure}[t]
\begin{minipage}[t]{.52\textwidth}
\vspace{0pt}
\begin{lstlisting}[style=leanhl,basicstyle=\ttfamily\footnotesize]
theorem encoded_len_proof : EncodedLenCHC
  := by unfold EncodedLenCHC; zap ; fix
        have := utf8_len_mono m (i+1) n
          (by omega)
        grind [cw, utf8_len_succ]
\end{lstlisting}
\end{minipage}
\hfill
\begin{minipage}[t]{.47\textwidth}
\vspace{0pt}
\begin{lstlisting}[style=leanhl,basicstyle=\ttfamily\footnotesize]
theorem utf8_len_mono: ∀ m i j,
  i ≤ j → utf8_len m i ≤ utf8_len m j
  := by
  -- induction on j - i; each step
  -- ... adds cw (m k) ≥ 0
\end{lstlisting}
\end{minipage}
\caption{Verifying \inlineflux{encoded_len} in \sclean with an inductive lemma about $\ulen$.}
\label{fig:overview:induction}
\label{fig:overview:proof}
\end{figure}

\mypara{Verification: CHC for \iflux{encoded_len}.}
To \emph{verify} \iflux{encoded_len}, \scflux generates
the Constrained Horn Clause (CHC) shown on the right of
\Cref{fig:encodedlen}.
This CHC is a \sclean proposition beginning with
existentially quantified \emph{Horn variables}
($\kvar$), which denote unknown refinements
or invariants, followed by nested universally
quantified binders naming statically unknown
program values, conjunctions collecting
requirements, and implications introducing
hypotheses, ending in a \emph{head} obligation
that must hold under the preceding hypotheses.
The CHC for \iflux{encoded_len} $\ctop$
has two Horn variables.
The first, $\kinv$, represents the \iflux{while}-loop
invariant over the accumulator \iflux{acc},
the counter \iflux{i}, and the slice's elements
 \iflux{m} and length \iflux{n}.
The second, $\kwid$, represents the unknown refinement
on the width \iflux{w} obtained after ``joining'' the
different branches of the nested \iflux{if-else} chain.
For readability, we write $\cbody$ and $\cjoin$ to name
two nested sub-constraints.
The first head $\kvapp{\kinv}{0, 0, m, n}$ requires
the invariant to hold at loop entry.
The second conjunct states that whenever the invariant
$\kinv$ holds and the loop exits $(n \le i)$ that the
returned \iflux{acc} equals the specified output
$\kvapp{\ulen}{m, n}$, and that if instead the
loop guard holds $(i < n)$ then the loop-body
sub-constraint $\cbody$ holds.
The loop body reads the character $m[i]$ and has one conjunct
for each branch of the \iflux{if} chain.
Each branch's head, \eg $\kvapp{\kwid}{2, \acc, i, m, n}$
checks that the $\kwid$ join invariant holds for the value
assigned to \iflux{w} on the corresponding branch.
Its first argument is the width being refined;
the remaining arguments record the program context.
The last conjunct, $\cjoin$, encodes the code after the join:
every width $w$ satisfying $\kvapp{\kwid}{w, \acc, i, m, n}$
must not overflow \iflux{acc}, and the head
$\kvapp{\kinv}{\acc + w, i + 1, m, n}$ checks
that the loop invariant is preserved by the loop body.

\mypara{Foundational Solvers for Acyclic and Cyclic Horn Variables.}
To solve the CHC we must provide a proof of $\ctop$ by
supplying \emph{explicit solutions} for both Horn variables
$\kwid$ and $\kinv$, that prove that they satisfy the constraints.
SMT-based approaches can solve CHCs over their supported theories,
but cannot handle arbitrary user-defined functions like $\ulen$.
\sys addresses this challenge in three steps.
First, following \cite{Cosman17} we use the notion of
variable dependencies --- where a Horn variable at the head
of a constraint \emph{depends on} the Horn variables in its hypotheses ---
to partition the variables into those that are \emph{cyclic}, meaning
they can transitively depend on themselves, and those that are \emph{acyclic}.
For example, in $\ctop$ (\Cref{fig:encodedlen}),
$\kinv$ is cyclic: it appears in both the hypothesis
$\kvapp{\kinv}{\acc, i, m, n}$ and the head
$\kvapp{\kinv}{\acc + w, i + 1, m, n}$ of the inductive clause.
In contrast, $\kwid$ is acyclic:
its heads depend on $\kinv$, but no dependency
leads from $\kwid$ to itself.
Second, we eliminate the acyclic variables
like $\kwid$ while constructing a proof
certifying the soundness of the reduction
(\Cref{sec:overview:solvers}).
Third, we use a fixpoint computation to infer
and certify an inductive invariant for the cyclic
variables like $\kinv$ (\Cref{sec:overview:expressiveness}).
Finally, at this point we are left with a plain (Horn-free)
verification condition that can be discharged using standard
\sclean tactics and automation.

\mypara{Verifying \iflux{encoded_len}.}
The left side of \Cref{fig:overview:proof}
shows how a user can prove the CHC for
\iflux{encoded_len} using \sys's tactics.
First we \lean{unfold} the CHC definition,
then we apply the \lean{zap} tactic to solve and eliminate
the acyclic Horn variables.
Next, the \lean{fix} tactic synthesizes an
inductive invariant $0 \le i \le n \wedge \acc = \kvapp{\ulen}{m, i}$
which, after substitution and elimination leaves
the residual proof goal $\acc + w \le \umax$
that checks the addition in the loop body does
not overflow.
The user looks at this residual obligation
to separately write a lemma \lean{utf8\_len\_mono}
about the monotonicity of $\ulen$ that requires
an inductive proof, and then applies the lemma
to discharge the goal to verify \iflux{encoded_len}.
Crucially, with our approach, the user can use
the ITP's proof infrastructure to manage the
entire complexity of the proof.
The user can use \sys to synthesize invariants,
and use the ITP's proof state context to identify
relevant helper lemmas, and then apply them, all
without cluttering the source code (\Cref{fig:encodedlen})
with proof artifacts.
But all of that crucially hinges upon the ability to
solve the Horn variables, which we describe next.

\subsection{Solving \emph{Acyclic} Horn Variables}
\label{sec:overview:solvers}

We solve acyclic variables using \Zapname (\Cref{sec:zap})
which exploits the structure of the constraint to synthesize
a solution whose structure mirrors the constraint
allowing us to construct the corresponding \sclean
witness deterministically, without search-based automation.

\mypara{Computing the Solution.}
First, we compute an exact solution for the \emph{acyclic}
Horn variable focusing on the variable's \emph{scope}: the
smallest sub-constraint containing all the variable's
occurrences.
For $\kwid$, the scope is the body constraint $\cbody$,
and we can compute its strongest solution which is intuitively
the disjunction of the hypotheses under which each head occurs:
\[
  \asgn \doteq \cLam{z, \acc, i, m, n}{
    \begin{array}[t]{@{}l@{}}
      \invisible{\vee}\ (m[i] < 128 \wedge z = 1) \\
      \vee\ (128 \le m[i] < 2048 \wedge z = 2) \\
      \vee\ (2048 \le m[i] < 65536 \wedge z = 3) \\
      \vee\ (65536 \le m[i] \wedge z = 4),
    \end{array}}
\]
Substituting $\asgn$ as the solution for $\kwid$ yields
the instantiated constraint shown
in the middle in \Cref{fig:overview:instantiatenew}.
Each head application becomes the disjunction $\asgn$ applied to the
branch's width, and the hypothesis application
becomes $\kvapp{\asgn}{w, \acc, i, m, n}$.
Instantiating $\kwid$ leaves behind the heads like
$\cFRM{\kvapp{\asgn}{2, \acc, i, m, n}}$ that must
still be proved, \emph{before} we can replace the
head with $\top$ on the right.
We could, in principle, prove those heads using
proof search tactics such as \texttt{grind} or
\texttt{aesop}, which suffice for this simple
example.
Sadly, in practice, proof search can be woefully slow (\Cref{sec:rq2}).

\figoverviewinstantiate{}

\mypara{Proof-producing Elimination.}
Instead, our $\Zapname$ tactic shows how to automatically
construct the proof for each eliminated head implication,
directly from the structure of the synthesized predicate.
Consider the head obligation
${\ldots (128 \le m[i] < 2048 \rightarrow \kvapp{\asgn}{2, \ldots})}$
and let $g$ name the guard hypothesis $128 \le m[i] < 2048$.
Recall the solution $\asgn$ was computed as a
\emph{disjunction} of the different guards under
which $\kwid$ appears.
Hence, we can prove this head via a proof term of the form
$\mathsf{inr} (\mathsf{inl} (\langle g, \rfl\rangle))$ where
the prefix of $\mathsf{inr}$ and $\mathsf{inl}$ record the path to
this head, thereby identifying the specific disjunct of
$\asgn$ that must be proved, after which the pair
$\langle g, \rfl \rangle$ proves the conjunct
$128 \le m[i] < 2048 \wedge z = 2$ using the
guard hypothesis $g$ and reflexivity.
We generalize this intuition in \Cref{fig:duality},
and formalize it in a deterministic proof construction
algorithm that lets $\Zapname$ eliminate the heads
and replace them with $\top$, as shown on the right
in \Cref{fig:overview:instantiatenew}.

In our evaluation (\Cref{sec:rq2}), \Zapname successfully reduces
99.8\% of 1,275 CHCs with eliminable predicates from
refinement-typechecker suites and CHC-COMP.
In contrast, when we instead use \texttt{grind}
or \texttt{aesop} with the same resource limits,
to prove the same synthesized heads, they only
succeed on 32.2\% and 33.4\% of constraints respectively.

\begin{figure}[t]
\centering
\small
$\begin{array}{@{}lll@{}}
\toprule
\textbf{Constraint} & \textbf{Solution} & \textbf{Head-proof construction} \\
\midrule
c_1 \wedge c_2 & \textit{sol}\,c_1 \;\vee\; \textit{sol}\,c_2 & \mathsf{inl} \,/\, \mathsf{inr} \\
\forall x.\ c & \exists x.\ \textit{sol}\,c & \langle x,\,\cdot\rangle \\
\pred \to c & \pred \;\wedge\; \textit{sol}\,c & \langle \hyp,\,\cdot\rangle \\
\kappa(\overline{t}) & \textstyle\bigwedge_i z_i=t_i & \langle\rfl,\dots,\rfl\rangle \\
\bottomrule
\end{array}$
\caption{Constraint structure determines the solution and its proof at an
eliminated head. Here $\textit{sol}$ denotes the synthesized predicate body,
$z_i$ its formal arguments, and $\hyp$ the guard hypothesis.}
\label{fig:duality}
\end{figure}

\subsection{Solving \emph{Cyclic} Horn Variables}
\label{sec:overview:expressiveness}

The Horn variable $\kinv$ is cyclic: it appears
as a hypothesis along the same path where it is
the goal.
Hence, we cannot apply $\Zapname$ as it
would loop indefinitely!
Next, we show how \sys's \Fixname tactic
infers an inductive invariant for $\kinv$
using \emph{predicate-abstraction}~\cite{Graf97,Houdini}
in three steps (\Cref{sec:fixpoint-algorithm}).

\smallskip\emph{Step 1: Qualifiers~}
The \Fixname tactic solves cyclic variables as
\emph{conjunctions of predicates} from a library
of \emph{qualifier} templates \cite{Rondon08}.
In \sys, qualifiers come from a standard library
of arithmetic predicates, \eg $\cLam{a, b}{a \le b}$,
from user-supplied predicates, and from \emph{scraping}
the CHC and the defining equations and branch conditions
of the relevant \sclean definitions (\Cref{sec:scraping}).
For example, the exit obligation of $\ctop$ already
contains a useful predicate, and scraping abstracts
its program variables to yield a qualifier over $\ulen$:
\[
  \acc = \kvapp{\ulen}{m, n}
  \xRightarrow{\text{scrape}}
  \cLam{a, m, k}{a = \kvapp{\ulen}{m, k}}.
\]

\smallskip\emph{Step 2: Initial Assignment~}
Next, \Fixname uses the qualifiers to compute
an initial assignment which instantiates
\emph{all} the qualifiers with all the
possible $\kinv$ parameters that yield
a well-typed \sclean predicate.
For example, the initial assignment for $\kinv$ would be
$$
\kvapp{\kinv}{a, z, m, n} \ \doteq\
\{
0 \le z, z \le n, \ldots,
a = \kvapp{\ulen}{m, z}, a = \kvapp{\ulen}{m, n}, \ldots
\}
$$

\smallskip\emph{Step 3: Weaken till Fixpoint~}
Finally, \Fixname enters an iterative weakening loop
where we use an \emph{oracle} --- here, \lean{grind}
with the unfolding equation of $\ulen$ ---
to \emph{discard} candidates that \emph{cannot} be proven
to hold at some CHC head, after substituting the current
assignment for the Horn variables in the hypotheses.
For example, $a = \kvapp{\ulen}{m, n}$ is discarded
as it fails at loop entry: an arbitrary slice may
have a nonzero encoded length.
In contrast, $a = \kvapp{\ulen}{m, z}$ survives, as the
oracle can use the recursive definition of $\ulen$ to
prove that it holds initially and is preserved by the loop body:\looseness=-1
\[
\begin{aligned}
  &0 = \kvapp{\ulen}{m, 0},\\
  &\acc = \kvapp{\ulen}{m, i} \wedge 0 \le i \wedge \kvapp{\asgn}{w, \acc, i, m, n}
    \;\rightarrow\;
    \acc + w = \kvapp{\ulen}{m, i+1}.
\end{aligned}
\]
Once we have eliminated the impossible, whatever candidates
remain, must be the truth. Hence, \Fixname returns a solution $\asgn'$
for $\kinv$ comprising the conjunction of all remaining candidates
$${\asgn'\ \doteq\ \cLam{a, z, m, n}{0 \le z \le n \wedge a = \kvapp{\ulen}{m, z}}}$$
As with \lean{zap}, we substitute this solution in for the hypotheses,
but eliminate the heads with $\top$ as shown in \Cref{fig:overview:fix}.
Finally, after applying \lean{zap} and \lean{fix}, we are left
with the ordinary \sclean goal corresponding to the overflow
check $\acc + w \le \umax$, which the proof in
\Cref{fig:overview:proof} closes with an inductive lemma
about $\ulen$.

\figoverviewfix{}

\mypara{Certifying the Elimination.}
Again, to eliminate cyclic heads with $\top$,
as shown in \Cref{fig:overview:fix}, \Fixname must construct
\sclean certificates that prove the corresponding implications, as
detailed in our proof-producing predicate abstraction procedure
in \Cref{sec:fixpoint-algorithm}.
In addition to native \sclean predicate abstraction which
is useful for invariants for user-defined functions and
bespoke theories, \sys also builds on the invariant-synthesis
of Spacer~\cite{Spacer1} and Eldarica~\cite{ELDARICA}.
For SMT-LIB constraints these \emph{external} solvers
discover invariants independently of our qualifier library.
\sys reconstructs their candidates as \sclean predicates and
certifies their use against the original constraint
(\Cref{sec:external-oracles}), retaining any unproved
obligations.
We evaluate native expressive inference in
\Cref{subsec:rq1-expr} and the proof coverage of these
complementary approaches in \Cref{subsec:rq3-auto}.
 \section{Constraints and Proofs}
\label{sec:chc}

\sys represents CHCs as a restricted fragment of \sclean propositions.
This fragment preserves the structure needed for Horn solving while
using native connectives such as $\forall$, $\wedge$, and $\to$.
We therefore implement our solvers as metaprograms over \sclean's
\Expr datatype, thus, like~\cite{Loom}, 
getting correctness for free from Lean's kernel.

In this section we start by defining the syntax and semantics of constraints (\Cref{sec:chc-syntax}) and then 
introduce notation for examining the
structure of constraints (\Cref{sec:chc-props}) and the dependencies
between their unknown predicates (\Cref{sec:partition}), preparing
the ground for the solving algorithms that follow.
As a running example throughout this section, consider a
constraint that asserts $x \ge 0$ and an unknown refinement
$\kconst$ relating $x$ to its successor $\nu$:
\begin{equation}
\cstr_0 \ \doteq\ \exists \kconst.\
\forall x.\ 0 \le x \to
(\forall \nu.\ \nu = x+1 \to \kvapp{\kconst}{\nu, x} ) \wedge
(\forall a.\ \kvapp{\kconst}{a, x} \to 1 \le a)
\label{ex:running}
\end{equation}

\subsection{Syntax, Semantics and Proofs}
\label{sec:chc-syntax}
\label{sec:proof-terms}

\figcstrsyntax{}

\Cref{fig:chc-grammar} defines the syntax of CHCs as \sclean propositions and of the proof terms that certify their solutions. 

\mypara{Sorts and Terms}
\emph{Sorts} $b$ range over arbitrary \sclean types, including
integers, bit-vectors, sets, and maps.
We write $x_1, x_2, x_3, \ldots$ for \emph{variables} of sort $b$,
and $t_1, t_2, t_3, \ldots$ for arbitrary \sclean terms of
the appropriate sort.
We use $\varphi$ for terms of sort $\Prop$.
We use $\hyp$ for proof variables whose types are propositions.
We assume terms include the trivial propositions $\top$ and $\bot$
(of sort $\Prop$).

\mypara{Horn Variables}
A \emph{Horn (refinement) variable} is a \sclean variable
$\kvar : b_1 \to \cdots \to b_{n_\kappa} \to \Prop$, and occurs
in a constraint as horn applications $\kvapp{\kvar}{\overline{t}}$.
Horn variables do not appear inside arbitrary terms $\leanterm$.
We existentially quantify these variables to represent \emph{unknown}
refinements.
For any constraint $\cstr$, we write $\KVars(\cstr)$ for the set
of Horn variables occurring in $\cstr$.

\mypara{Atoms}
An \emph{atom} $\pred$ is either
a $\kvar$-free term $\leanterm$ of sort \Prop or
a \emph{horn application} $\kvapp{\kvar}{\overline{t}}$.

\mypara{Constraints}
A \emph{Horn constraint} $\cstr$ is either
(1)~a \emph{head} atom $\pred$,
(2)~a \emph{conjunction} $\cAnd{\cstr_1}{\cstr_2}$,
(3)~an \emph{implication} $\cImp{\pred}{\cstr}$, or
(4)~a \emph{quantification} $\cAll{x}{b}{\cstr}$
over a variable $x$ of sort $b$.
A \emph{closed constraint} $\Cstr$ is a \sclean proposition
of the form $\exists \kvar_1, \ldots, \kvar_n. \cstr$,
where all the refinement variables in $\cstr$ are
existentially bound at the top.
A \emph{verification condition} (VC) is a closed constraint
that has \emph{no} Horn variables.
We write $\cstr$ for both open and closed constraints when clear
from context.

\mypara{Predicates and Interpretations}
A \emph{predicate} in \Preds{}  is a \sclean term of sort
$b_1 \to \ldots \to b_n \to \Prop$ (\ie, $\Preds{} \subset \Terms$).
For a horn variable $\kvar$, an \emph{interpretation} $\asgn \in \Preds$
is a predicate of the same type as $\kvar$.
For $\cstr_0$ of \Cref{ex:running}, the predicate $\asgn \doteq \lambda \nu\, x.\ \nu \ge 1$
is an interpretation for $\kconst$, since $\asgn : \Int \to \Int \to \Prop$.

\mypara{Satisfiability}
We \emph{apply} an interpretation $\asgn$ to a constraint $\cstr$
by replacing each application $\kvapp{\kvar}{\overline{t}}$ with
$\kvapp{\asgn}{\overline{t}}$, writing
$\subasgn{\cstr}{\kvar}{\asgn}$ for the result.
A closed constraint is \emph{satisfiable} if interpretations for its
Horn variables make the resulting $\kvar$-free proposition valid
(\ie, provable in \sclean). Such interpretations form a \emph{solution}.
For example, applying $\asgn$ to $\cstr_0$ yields the following VC:
\[
\forall x.\ 0 \le x \to
(\forall \nu.\ \nu = x+1 \to \nu \ge 1)\wedge
(\forall a.\ a \ge 1 \to 1 \le a).
\]
This VC is valid; hence $\asgn$ is a solution and $\cstr_0$
is satisfiable.

\mypara{Proof Terms}
The solvers construct certificates using ordinary \sclean proof terms
(\Cref{fig:zap-proofterms}): $\lambda$-terms and applications for
$\forall$ and implication; pairs and projections for conjunction;
injections for disjunction; existential witnesses and unpacking;
and $\rfl$ and $\langle\rangle$ for equality and $\top$.
These forms, together with arbitrary \sclean terms $\leanterm$, suffice
for the constraint logic and the solutions computed by our solvers.
Their typing is inherited from \sclean, whose kernel checks each certificate.
The same logical forms occur in \scrocq~\cite{Coq-refman} and
\textsc{Agda}~\cite{Norell2009}, so the certification approach is not
specific to \sclean.

\subsection{Structural Operations: Scope \& Contexts}
\label{sec:chc-props}

\figscope{}

Setting the ground for efficient constraint solving, 
here we define the structural operations that compute the 
 smallest enclosing scope of each Horn variable and the
hypotheses visible there.

\mypara{Prefixes, Contexts, Heads}
A \emph{prefix} $\pfx$ is a list of \emph{steps}
used to navigate a path through a constraint's syntax tree.
A \emph{step} $s$ is either a left \pLeft or right \pRight
choice in a conjunction,
a variable binding \pBind{x}{b},
or an assumption \pAsm{\pred}.
A \emph{context} $\ctx$ is an ordered list
of variable bindings $\cBind{x}{b}$ (one per $\forall$),
named hypotheses $\cBind{\hyp}{\pred}$ (one per $\to$),
or a $\pLeft$ or $\pRight$ choice gathered along a prefix.
$$\begin{array}{l@{\quad}r@{\;}c@{\;}l}
\textbf{\textit{Prefixes}}  & \pfx & ::= & \varepsilon \mid s \cdot \pfx \\[0.2em]
\textbf{\textit{Steps}}     & s    & ::= & \pLeft \mid \pRight \mid \pBind{x}{b} \mid \pAsm{\pred} \\[0.2em]
\textbf{\textit{Contexts}}  & \ctx & ::= & \varepsilon \mid \cBind{x}{b}; \ctx \mid \cBind{\hyp}{\pred}; \ctx \mid \pLeft; \ctx \mid \pRight; \ctx
\end{array}$$
The procedure \CtxOf{\cstr}{\pfx} computes this context, while
\HeadOf{\cstr}{\pfx} returns the sub-constraint at the end of
$\pfx$ (\Cref{fig:ctx}).
The \emph{truncated context} $\ctx \setminus \pfx$ drops the
initial entries of $\ctx$ matching $\pfx$.

\mypara{Goals}
A prefix $\pfx$ is a \emph{head-prefix} of $\cstr$ if
\HeadOf{\cstr}{\pfx} is a head atom $\pred$.
We write $\Heads{\cstr}$ for the set of Horn variables with
a head application in $\cstr$.
A \emph{goal} $\ctx \vdash \pred$ is a context $\ctx$ paired
with a head $\pred$. It is \emph{valid} if $\pred$ is derivable
under the bindings and guards in $\ctx$.
For a head-prefix $\pfx$, we define
\[
\GoalOf{\cstr}{\pfx}
  \doteq \CtxOf{\cstr}{\pfx} \vdash \HeadOf{\cstr}{\pfx}.
\]
If its head is $\kvapp{\kvar}{\overline{t}}$, we call it a
\emph{$\kvar$-goal}; an interpretation $\asgn$ solves it if
$\ctx \vdash \asgn(\overline{t})$ is valid.

\mypara{Scopes}
The procedure $\Scope{\kvar}{\cstr}$ (\Cref{fig:scope}) returns
the prefix of the smallest sub-constraint containing all applications
of $\kvar$, in both guards and heads.
The constraint $\cstr_1$ below extends $\cstr_0$ with $\kconstb$
under the guard $\kconst(a,x)$. Its scopes appear on the right
($\wedge$ is right associative).

\begin{equation}
\begin{array}{rc@{}ll}
\cstr_1 & \doteq &
\begin{array}{l}
\forall x.\ 0 \le x \to \\
\quad \invisible{\wedge}\ \forall \nu.\ \nu = x+1 \to \kvapp{\kconst}{\nu, x} \\
\quad \wedge\             \forall a.\ \kvapp{\kconst}{a, x} \to \\
\quad \quad \invisible{\wedge}\ 1 \le a \\
\quad \quad \wedge\ \kvapp{\kconstb}{a+1, a, x} \\
\quad \quad \wedge\ \forall \nu.\ \kvapp{\kconstb}{\nu, a, x} \to 2 \le \nu
\end{array}
&
\begin{array}{rcl}
\quad \Scope{\kconst}{\cstr_1} & \doteq & \pBind{x}{\Int} \cdot \pAsm{0 \le x} \\
\quad \Scope{\kconstb}{\cstr_1} & \doteq & \pBind{x}{\Int} \cdot \pAsm{0 \le x} \cdot \pRight \cdot \pBind{a}{\Int} \cdot \pAsm{\kvapp{\kconst}{a, x}} \cdot \pRight
\end{array}
\end{array}
\label{ex:scope}
\end{equation}

\mypara{Well-formedness}
Let $\pfx \doteq \Scope{\kvar}{\cstr}$,
$\ctx \doteq \CtxOf{\cstr}{\pfx}$, and
$\Binders{\ctx} \doteq x_0,\ldots,x_n$ list the scope's binders
innermost first.
A constraint $\cstr$ is \emph{well-formed} for $\kvar$ if every
application has the form $\kvapp{\kvar}{\overline{t},x_0,\ldots,x_n}$:
its trailing arguments are exactly these binders, while its leading
arguments $\overline{t}$ are arbitrary terms.
We write $\Rank{\kvar}$ for the number of leading arguments.
A constraint is \emph{well-formed} if this holds for every
$\kvar \in \KVars(\cstr)$. We assume well-formed constraints throughout.
For example, $\cstr_1$ is well-formed: every application of
$\kconst$ ends in $x$, and every application of $\kconstb$
ends in $a,x$, matching their respective scope binders.

\subsection{Partitioning Horn Variables}
\label{sec:partition}

Partitioning is a preprocessing step shared by the solvers; it determines
which variables can be eliminated and which require invariant synthesis.

\mypara{Dependencies}
The \emph{dependencies} of a constraint $\cstr$ form a relation
$\deps{\cstr} \subseteq \kvars \times \kvars$, defined by
\[
\deps{\cstr} \doteq
\{(\kvar,\kvar') \mid
  \exists \pfx.\;
  \kvar \in \CtxOf{\cstr}{\pfx}
  \;\mbox{and}\;
  \kvapp{\kvar'}{\overline{t}} = \HeadOf{\cstr}{\pfx}\}.
\]
We say $\kvar'$ \emph{depends on} $\kvar$ in $\cstr$
if $(\kvar,\kvar') \in \deps{\cstr}$.
This means that $\kvar$ occurs in a guard along some path
to a head application of $\kvar'$.
The \emph{dependency graph} $\dg{\cstr}$ has vertices
$\KVars(\cstr)$ and an edge from $\kvar$ to $\kvar'$
if $\kvar'$ depends on $\kvar$.
Intuitively, we solve for $\kvar$ \emph{before} solving
for its dependents.

\mypara{Cycles, Cuts and Partitions}
A \emph{cut set} is a feedback vertex set of $\dg{\cstr}$:
removing its vertices and incident edges makes the graph acyclic.
The sets $\kcycs{\kvar}, \kacys{\kvar}$ \emph{partition} the
Horn variables of $\cstr$ if
(1) they are disjoint,
(2) $\kcycs{\kvar} \cup \kacys{\kvar} = \KVars(\cstr)$, and
(3) $\kcycs{\kvar}$ is a cut set.
We call $\kacys{\kvar}$ the \emph{acyclic} set.
The procedure $\Partition{\overline{\kvar}}{\cstr}$ computes
such a partition. Within the acyclic set, each variable is bound
outside its dependencies, so innermost-first elimination
(\Cref{sec:zap}) removes dependencies first.
Finding a minimum feedback vertex set is NP-hard~\cite{Karp72},
so $\Partition{\overline{\kvar}}{\cstr}$ uses a greedy heuristic
based on Tarjan's algorithm~\cite{Tarjan72}.
The cut set need not be minimum for soundness.

We pass the acyclic set $\kacys{\kvar}$ to \Zapname{}
for elimination (\Cref{sec:solver}), and synthesize invariants
for the cut set $\kcycs{\kvar}$ using predicate abstraction
(\Cref{sec:fixpoint}) or external oracles
(\Cref{sec:external-oracles}).
 \section{\Zapname: Proof-Producing Acyclic Horn Solving}
\label{sec:sec4}
\label{sec:solver}
\label{sec:zap}

\Zapname eliminates the Horn variables outside the cut set
computed in \Cref{sec:partition}, which we call \emph{acyclic}
by a slight abuse of terminology.
It computes strongest solutions and certifies the resulting rewrites.
We describe its elimination (\Cref{sec:acyclic}) and proof
construction (\Cref{sec:certification}) steps, then establish
that the generated proofs are well-typed (\Cref{sec:zap-correctness}).
The remaining cut variables are handled by predicate abstraction
(\Cref{sec:fixpoint}) or external invariant-synthesis oracles
(\Cref{sec:external-oracles}).

\subsection{Synthesizing Strongest Predicates}
\label{sec:acyclic}

\figzapelim{}

\begin{figure}[t]
\resizebox{\linewidth}{!}{$\begin{array}{@{}l@{}c@{}l@{}c@{}l}
\begin{array}{l}
\exists \kconstb, \kconst. \\
\forall x.\ 0 \le x \to \\
\ \invisible{\wedge}\ \forall \nu.\ \nu = x+1 \to \cFRM{\kvapp{\kconst}{\nu, x}} \\
\ \wedge\ \forall a.\ \cFRM{\kvapp{\kconst}{a, x}} \to 1 \le a \\
\quad\ \wedge\ \kvapp{\kconstb}{a+1,a,x}  \\
\quad\ \wedge\ \forall \nu.\ \kvapp{\kconstb}{\nu,a,x} \to 2 \le \nu\\
\end{array}
&
\xrightarrow[\cLam{z_0\,x}{z_0 = x + 1}]{\ \mbox{reduce}\ \kconst}
&
\begin{array}{l}
\exists \kconstb. \\
\forall x.\ 0 \le x \to \\
\ \invisible{\wedge}\ \forall \nu.\ \nu = x+1 \to \cFRM{\top} \\
\ \wedge\ \forall a.\ \cFRM{a = x+1} \to 1 \le a\\
\quad\ \wedge\ \cFRM{\kvapp{\kconstb}{a+1,a,x}}  \\
\quad\ \wedge\ \forall \nu.\ \cFRM{\kvapp{\kconstb}{\nu,a,x}} \to  2 \le \nu
\end{array}
&
\xrightarrow[\cLam{z_0\,a\,x}{z_0 = a + 1}]{\ \mbox{reduce}\ \kconstb}
&
\begin{array}{l}
\\
\forall x.\ 0 \le x \to \\
\ \invisible{\wedge}\ \forall \nu.\ \nu = x+1 \to \top \\
\ \wedge\ \forall a.\ a = x+1 \to 1 \le a \\
\quad\ \wedge\ \cFRM{\top}  \\
\quad\ \wedge\ \forall \nu.\ \cFRM{\nu = a+1} \to 2 \le \nu
\end{array}
\end{array}$
}
\caption{Eliminating acyclic variables with \Zapname.
(L) initial constraint from \Cref{ex:scope},
(C) result of reducing $\kconst$, and
(R) result of reducing $\kconstb$ yielding a VC with no Horn variables.
}
\label{fig:ex:zap}
\end{figure}

\mypara{Top-level Algorithm (\Zapname{}):}
\Cref{fig:algo:zap} defines the top-level \Zapname{}
algorithm which takes as input a closed constraint
$\ecstr{\kvar}{\cstr}$, and returns a pair
of a \emph{new} constraint $\ecstr{\kcyc{\kvar}}{\cstr'}$
where all occurrences of the acyclic variables have
been removed, and a \emph{proof term} $\pf'$ that
shows that the output constraint \emph{entails}
the original constraint.
The cut set $\kcycs{\kvar}$ and acyclic set $\kacys{\kvar}$
are supplied by the preprocessing of \Cref{sec:partition}
and are left implicit in the notation $\Zap{\cstr}$.
\Zapname{} invokes \Zapsname{} to eliminate the acyclic
variables from the innermost binder outward.
On input $\exists\kvar.\,\cstr$, \Zapsname{} first recursively
eliminates the inner variables, obtaining $\cstr_1$,
then invokes \Zaponame{} to eliminate $\kvar$ from $\cstr_1$.
\Zapsname{} then composes the proofs
$\pf_2 : \cstr_2 \to \exists\kvar.\,\cstr_1$ and $\pf_1 : \cstr_1 \to \cstr$
into a proof $\pf_3 : \cstr_2 \to \exists\kvar.\,\cstr$
by unpacking and repacking the $\kvar$-witness.
At the top, \Zapname{} re-bundles the cut and acyclic witnesses
in the original binder order of $\overline{\kvar}$.
The \Zaponame{} function eliminates a \emph{single}
acyclic variable $\kvar$ from a constraint $\cstr$
in two steps.
In the first step, \Zaponame{} invokes the \Elimname{} procedure
to \emph{eliminate} the variable $\kvar$ by computing a triple comprising
(1) the \emph{strongest solution} $\asgn$ for $\kvar$ in \cstr,
(2) a prefix $\pfx$ corresponding to the scope of $\kvar$ in \cstr, and
(3) a constraint $\cstr'$ where $\kvar$ has been eliminated.
In the second, \Zaponame{} invokes \Certname with the triple
to construct the term
$\cLam{\hyp}{\langle \asgn,\ \Cert{\varepsilon}{\hyp}{\cstr} \rangle}$
that \emph{certifies} the elimination (as explained in~\Cref{sec:certification}):
it maps any proof $\hyp$ of $\cstr'$ to a proof of
$\exists\kvar.\,\cstr$, with $\asgn$ as the witness.
\Cref{fig:ex:zap} illustrates the steps of \Zapname{}
on the constraint $\cstr_1$ from \Cref{ex:scope}.

\mypara{Variable Elimination (\Elimname):}
The procedure $\Elim{\kvar}{\cstr}$, shown
in \Cref{fig:algo:elim}, eliminates $\kvar$ from
$\cstr$ in three steps.
First, it computes the \emph{scope}
\pfx{} of $\kvar$ in $\cstr$.
Next, it calls $\Solname$
on $\HeadOf{\cstr}{\pfx}$ --- the
sub-constraint of $\cstr$ where $\kvar$
occurs --- to compute $\asgn$: the
\emph{strongest solution} for $\kvar$.
Finally, the procedure invokes
$\Rem{\kvar}{\asgn}{\cstr}$,
to compute the \emph{reduced constraint}
where all \emph{head} applications
of $\kvar$ are replaced by $\top$,
since they hold by construction of $\asgn$,
and all other non-head occurrences of
$\kvapp{\kvar}{\overline{t}}$ are replaced by
the solution $\asgn(\overline{t})$, yielding
a constraint \emph{without} $\kvar$.

\mypara{Strongest Solution (\Solname):}
A solution $\asgn$ is \emph{valid} for $\kvar$ in $\cstr$ if it is
valid for every $\kvar$-goal (defined in \Cref{sec:chc}).
The procedure $\Sol{\kvar}{\cstr}$, shown in \Cref{fig:algo:sol},
computes the \emph{body} of such a solution: the
\emph{disjunction} of the contexts at each head-prefix for $\kvar$,
existentially quantifying the variables not in the scope;
$\Elimname$ closes the body under $\overline{z}$ and $\overline{x}$
to obtain $\asgn$.
The guard $\kvar \notin \Heads{\cstr}$ collapses each
sub-constraint without a $\kvar$ head to an inert $\bot$,
keeping dead-branch horn applications out of the solution.
For example, \Cref{fig:ex:zap} shows below each arrow, the solution
computed for the $\kvar$ reduced in the respective step.
In the sub-constraint for $\kconst$ there is exactly one
head application, with the (truncated) context $\cBind{\nu}{\Int}; \nu = x+1$;
the conjunct with the \emph{uses} of $\kconst$ collapses to $\bot$.
Thus, the computed solution is
$\cLam{z_0\, x}{(\exists \nu.\; \nu = x + 1 \wedge z_0 = \nu) \vee \bot}$
which we simplify to $\cLam{z_0\, x}{z_0 = x + 1}$.

\mypara{Scope avoids Exponential Blowup}
The procedure $\Elimname$ is analogous
to the similarly named one from \cite{Cosman17},
which introduces the scope optimization
to sidestep an exponential blowup that
otherwise occurs in practice.
Specifically, if we compute $\Sol{\kvar}{\cstr}$
\emph{without} considering the scope $\pfx$, then
the assumptions from $\CtxOf{\cstr}{\pfx}$ get
(unnecessarily) included in the strongest solution,
which causes an exponential duplication of assumptions
as we eliminate multiple variables.

\subsection{Constructing Elimination Proofs}
\label{sec:certification}

\figcert{}

In a foundational setting,
we cannot just rewrite a constraint,
we also need to prove the original one.
To this end, \Certname \emph{converts} a proof of the rewritten constraint
into a proof of the original one.

\mypara{Certifying an Assignment (\Certname):}
The procedure $\Cert{\ctx}{\pf}{\cstr}$ (\Cref{fig:certify})
takes two sets of inputs.
First, a parametric set comprising $\kvar$, its solution $\asgn$
and its scope $\pfx$ in $\cstr$.
Second, a set with a context $\ctx$ of accumulated bindings,
the original sub-constraint $\cstr$,
and a proof $\pf$ of the $\asgn$-reduced sub-constraint
(\ie, $\cstr' \doteq \Rem{\kvar}{\asgn}{\cstr}$).
The procedure then traverses $\cstr$, transforming $\pf$
at each step into a proof of the corresponding sub-constraint.
For the $\cAll{x}{b}{\cstr}$ (resp. $\cImp{\pred}{\cstr}$)
case, the procedure emits a $\lambda$-term that
applies the supplied input $x$ (resp. $\hyp_\pred$) to $\pf$
to recursively generate the certificate for the sub-constraint $\cstr$.
For a $\cAnd{\cstr_1}{\cstr_2}$ the procedure emits a pair term
comprising the sub-terms obtained from the proofs of $\cstr_1$
and $\cstr_2$ constructed using the proofs $\pf.1$ and $\pf.2$
respectively.
The other heads $\pred$ are unchanged so we can directly reuse
the corresponding proof term $\pf$.

For example, with $\asgn \doteq \cLam{z\,x}{z=x+1}$,
the clause $\nu=x+1 \to \kvapp{\kvar}{\nu,x}$
reduces to $\nu=x+1 \to \top$.
After substituting $\asgn$, \Certname reconstructs its proof as:
$
  \lambda (\hyp : \nu=x+1).\,\hyp
  \;:\; \nu=x+1 \to \nu=x+1.
$

\mypara{Scope Optimization and Certification}
The scope optimization complicates proof construction:
\Solname\ operates on a scoped sub-constraint, omitting its enclosing
guards, while \Certname\ traverses the original constraint and
accumulates the full context.
Certification must therefore align the synthesized predicate with
the corresponding part of this context, preserving the scope's
variable bindings.
The prefix machinery (\Cref{sec:chc-props,fig:ctx}) records this
correspondence and guides the construction of the proof terms.

\mypara{Reconciling divergent heads (\Navname):}
At the divergent heads in $\Certname$ we call
$\Navname$ to furnish us with a proof for the
original goal $\kvapp{\kvar}{\overline{t}}$
post-substitution, \ie, for $\asgn(\overline{t})$.
Recall that the solution $\asgn$ was
a disjunction of contexts at each
head-prefix for $\kvar$ in the
sub-constraint $\HeadOf{\cstr}{\rho}$
where $\kvar$ occurs.
Thus, we use the \emph{truncated} context
($\ctx \setminus \rho$) that led to the
divergent head to
\emph{reconstruct} the proof term for that
disjunction from the context.

\mypara{Example}
In the example in \Cref{fig:ex:zap}, $\kconst$ has the scope prefix
$\rho \doteq \pBind{x}{\Int} \cdot \pAsm{0 \le x}$.
Thus, at the head $\kvapp{\kconst}{\nu, x}$, $\Certname$ has accumulated
the context $\ctx \doteq \cBind{x}{\Int}; \cBind{\hyp_0}{0 \le x}; \pLeft; \cBind{\nu}{\Int}; \cBind{\hyp_1}{\nu = x + 1}$
(naming the guard hypotheses $\hyp_0, \hyp_1$),
and the truncated context
is $\ctx \setminus \rho \doteq \pLeft; \cBind{\nu}{\Int}; \cBind{\hyp_1}{\nu = x + 1}$.
Given the solution
$\asgn \doteq \cLam{z_0\, x}{(\exists \nu'.\, \nu' = x + 1 \wedge z_0 = \nu') \vee \bot}$
(the bound variable renamed to $\nu'$), $\Navname$
replays the path, as follows, consuming one entry per node:
\[\small
\begin{array}{@{}c@{\;\;}l@{\qquad}l@{}}
  & \Nav{\pLeft;\cBind{\nu}{\Int};\cBind{\hyp_1}{\nu = x + 1}}{\asgn(\nu, x)} & \\[3pt]
= & \mathsf{inl}\;\Nav{\cBind{\nu}{\Int};\cBind{\hyp_1}{\nu = x + 1}}{\exists \nu'.\, \nu' = x{+}1 \wedge \nu = \nu'}
    & \pLeft \text{ picks the live disjunct} \\[3pt]
= & \mathsf{inl}\,\big\langle \nu,\; \Nav{\cBind{\hyp_1}{\nu = x + 1}}{(\nu = x{+}1) \wedge \nu = \nu}\big\rangle
    & \nu \text{ witnesses } \exists\nu' \\[3pt]
= & \mathsf{inl}\,\big\langle \nu,\; \big\langle \hyp_1,\; \Nav{\varepsilon}{\nu = \nu}\big\rangle\big\rangle
    & \hyp_1 \text{ discharges } \nu = x{+}1 \\[3pt]
= & \mathsf{inl}\,\big\langle \nu,\; \big\langle \hyp_1,\; \rfl\big\rangle\big\rangle.
    & \text{slot equality is } \rfl \\
\end{array}
\]

\mypara{Equivalence and Soundness}
\label{sec:zap-correctness}
\Zapname inherits satisfiability
results from \citet{Cosman17}.
We proved, in Appendix~\ref{app:zap-proofs},
generated proofs of \Zapname are sound,
with respect to a model of \sclean kernel.
That is, all proofs generated by \Zapname
will type check under \sclean kernel.

\begin{theorem}[Soundness of Certificates] \label{thm:zap:sound}
If $(\cstr', \pf) = \Zap{\cstr}$ then $\vdash \pf : \cstr' \to \cstr$.
\end{theorem}
 \section{\Fixname: Proof-Producing Predicate Abstraction}

\label{sec:fixpoint}
\label{sec:cyclic}
\label{sec:sec5}

After \Zapname (\Cref{sec:sec4}) eliminates acyclic Horn variables,
\Fixname synthesizes inductive invariants for the remaining variables
using \emph{predicate abstraction}~\cite{Graf97}.
Each invariant is a conjunction of instances drawn from a library
of predicate templates, called \emph{qualifiers}.
While traditional SMT-based implementations~\cite{Houdini,Rondon08}
restrict qualifiers to theories supported by their solvers,  
in our foundational setting, qualifiers can be arbitrary \sclean predicates.
In this section, we define qualifiers and their instantiation (\Cref{sec:predicate-abstraction})
and present the fixpoint algorithm (\Cref{sec:fixpoint-algorithm}).
Finally, we explain proof certification and establish soundness
and relative completeness (\Cref{sec:pa-certification}).
Qualifier discovery is described in the implementation (\Cref{sec:scraping}).

\subsection{Qualifiers, Instances and Candidates}
\label{sec:predicate-abstraction}
Informally, qualifiers are predicate \emph{templates}.
Formally, a \emph{qualifier} $\qual$ is a \sclean predicate
of the form\[
  \qual \doteq
  \cLam{\cBind{z_0}{b_0},\ldots,\cBind{z_n}{b_n}}{\leanterm}.
\]
Let $\kvar : b'_0 \to \ldots \to b'_m \to \Prop$ be a Horn variable.
We say that $\inst$ is a $\kqinst{\kvar}{\qual}$ \emph{instance}
if $\inst : [0,n] \to [0,m]$ and
$b_j = b'_{\inst(j)}$ for each $0 \le j \le n$.
That is, an instance $\inst$ maps each parameter of the qualifier
to a parameter of $\kvar$ of the same sort.
A \emph{candidate} is a pair of a qualifier $\qual$
and an instance $\inst$, written $\qual_\inst$.
We overload $\qual_\inst$ to denote the predicate induced
by instantiating $\qual$ along $\inst$:
\[
  \qual_\inst \doteq
  \cLam{z_0 \cdots z_m}
    {\qual(z_{\inst(0)},\ldots,z_{\inst(n)})}.
\]
Given a \emph{candidate set} $\overline{\qual_\inst}$, we write
$\conjoinsol{\overline{\qual_\inst}}$ for its pointwise conjunction.

\mypara{Example}
\newcommand{\kex}{\kvar_{\mathsf{ex}}}
Consider a Horn variable $\kex : \Int \to \Int \to \Prop$ together with
the binary qualifier
\[
  \qual_{\le} \doteq
  \cLam{\cBind{z_0}{\Int},\,\cBind{z_1}{\Int}}{z_0 \le z_1}.
\]
Each of its two parameters of sort $\Int$ can be mapped to either
of the parameters of $\kex$.
Hence, there are four $\kqinst{\kex}{\qual_{\le}}$ instances,
yielding the candidate set
\[
\left\{
  \cLam{z_0\,z_1}{z_0 \le z_0},\quad
  \cLam{z_0\,z_1}{z_0 \le z_1},\quad
  \cLam{z_0\,z_1}{z_1 \le z_0},\quad
  \cLam{z_0\,z_1}{z_1 \le z_1}
\right\}.
\]

\mypara{Assignments, Application and Validity}
An \emph{assignment} $\asgns \in \QualAsgnTy$ maps each $\kvar$
to a candidate set.
We \emph{apply} $\asgns$ to a Horn application,
written $\appasgn{\kvapp{\kvar}{\overline{t}}}{\asgns}$,
by applying the pointwise conjunction of the candidate
set $\asgns(\kvar)$.
We lift application to contexts $\appasgn{\ctx}{\asgns}$
by applying it to each hypothesis in $\ctx$.
\[\small
\begin{array}{lclr}
\appasgn{\kvapp{\kvar}{\overline{t}}}{\asgns}
  & \doteq & (\conjoinsol{\asgns(\kvar)})(\overline{t})
  & \mbox{Horn applications}\\
\appasgn{\leanterm}{\asgns}
  & \doteq & \leanterm
  & \mbox{other atoms}\\
\appasgn{\cBind{h}{\pred}; \ctx}{\asgns}
  & \doteq & \cBind{h}{\appasgn{\pred}{\asgns}}; \appasgn{\ctx}{\asgns}
  & \mbox{hypotheses}\\
\appasgn{\cBind{x}{b}; \ctx}{\asgns}
  & \doteq & \cBind{x}{b}; \appasgn{\ctx}{\asgns}
  & \mbox{binders}\\
\appasgn{\pLeft; \ctx}{\asgns}
  & \doteq & \pLeft; \appasgn{\ctx}{\asgns}
  & \mbox{marks (\pRight\ alike)}
\end{array}
\]
An assignment $\asgns$ is \emph{valid} for a $\kvar$-goal
$\ctx \vdash \kvapp{\kvar}{\overline{t}}$ if the $\kvar$-free goal
$\appasgn{\ctx}{\asgns} \vdash
 \appasgn{\kvapp{\kvar}{\overline{t}}}{\asgns}$
is valid.
We write $\cSat{\asgns}{\cstr}$ when $\asgns$ is valid for
every $\kvar$-goal of $\cstr$.

\subsection{Fixpoint Algorithm}
\label{sec:fixpoint-algorithm}
\Cref{fig:fix} summarizes $\Fixname$, which takes a constraint
$\cstr$ and qualifiers $\quals$ and returns a pair $(\cstr', \pf)$:
a constraint $\cstr'$ with all Horn variables removed, and a proof
$\pf$ that $\cstr' \rightarrow \cstr$.
It first computes an assignment $\asgns$,
invoking $\PredAbsname$ on the constraint and qualifiers.
It solves \emph{all} Horn variables simultaneously
to account for their cyclic dependencies.
Next, it folds $\Remname$ over $\cstr$ and all the horn variables,
to reduce the head applications of $\kvar$ with $\top$ and
the body applications with the computed candidate set $\asgns(\kvar)$.
Finally, it invokes $\Certsname$ to compute the proof $\pf$ of
$\cstr' \rightarrow \cstr$ (\Cref{sec:pa-certification}).

\figpiglet{}

\mypara{Predicate Abstraction ($\PredAbsname$):}
The right column in \Cref{fig:predabs} shows how
\PredAbsname takes Horn variables $\overline{\kvar}$,
a constraint $\cstr$, and qualifiers $\quals$ to compute
an assignment valid for $\cstr$.
First, $\Initname$ computes an \emph{initial assignment}
mapping each $\kvar$ to all candidates obtained from
$\kqinst{\kvar}{\qual}$ instances of each qualifier in $\quals$.
Second, it invokes $\Fixpointname$ to distill these candidates
down to those that can be proven valid at every head
$\kvar$-application, allowing $\Fixname$ to replace these heads
with $\top$.

\mypara{Fixpoint Computation ($\Fixpointname$):}
$\Fixpointname$  takes a constraint $\cstr$
and assignment $\asgns$, and iteratively invokes $\Weakenname$
to whittle away candidates that \emph{cannot be proven} valid.
For a head-prefix $\pfx$, $\Weakenname$ uses $\Oraclename$
to determine the subset $I$ of $\asgns(\kvar)$ that can be
proven valid at the corresponding head.
If $I$ is strictly smaller than $\asgns(\kvar)$, \ie, some whittling
occurred, $\Weakenname$ returns the assignment with $\kvar$ updated to $I$;
otherwise, it returns $\bot$.
$\Fixpointname$ checks whether \emph{any} head-prefix weakens
$\asgns$ to $\asgns'$.
If so, it recurses on $\asgns'$.
Otherwise, $\asgns$ is valid for $\cstr$ and is returned.
$\Fixpointname$ terminates: the initial assignment is finite, and each
successful $\Weakenname$ strictly decreases the total number of candidates.
$\Weakenname$, and hence $\Fixname$, is parameterized by
$\Oraclename$, which attempts to \emph{prove} a goal.
This lets users plug in theory-specific oracles to synthesize
solutions for CHCs over arbitrary \sclean propositions.

\subsection{Certification and Correctness}
\label{sec:pa-certification}

\mypara{Certification ($\Certsname$):}
Finally, $\Fixname$ invokes $\Certsname$ to justify the rewriting
done by $\Remname_{\asgns}$, producing a proof that the reduced
constraint $\cstr'$ \emph{implies} the original $\cstr$.
Like $\Zapname$'s certification (\Cref{sec:certification}),
it traverses the structure of $\cstr$, replaying the proof of $\cstr'$.
The only difference is at head applications
$\kvapp{\kvar}{\overline{t}}$, where, instead of navigating
the structure of the solution, it invokes $\Oraclename$
to prove the applied conjunction
$(\conjoinsol{\asgns(\kvar)})(\overline{t})$.

\mypara{Correctness}
$\Fixname$ requires that on any goal,
\myopname{Oracle} either returns a
valid proof term, or fails:
\begin{mathpar}
\inferrule*[left=\rname{Oracle-Soundness}]
  {\Oracle{\ctx}{\varphi} = \pf}
  {\ctx \vdash \pf : \varphi}
\and
\inferrule*[left=\rname{Oracle-Completeness}]
  {\ctx \vdash \varphi}
  {\Oracle{\ctx}{\varphi}\ \mbox{succeeds}}
\end{mathpar}

\begin{theorem}[\Fixname-Soundness] \label{thm:fix:sound}
If \Oraclename~ is sound and $\Fix{\cstr}{\quals} = (\cstr', \pf)$, then $\vdash \pf : \cstr' \rightarrow \cstr$.
\end{theorem}

\mypara{Relative Completeness}
It would be trivially sound, but not terribly useful
to replace each horn application with $\top$.
Fortunately, $\Fix{\cstr}{\quals}$ does better: it returns
the \emph{strongest} valid solution expressible as conjunctions
of predicates from $\quals$.
Given two $\quals$-assignments, we say $\asgns$
\emph{is stronger than} $\asgns'$, written
$\asgns \preceq \asgns'$, if for every $\kvar$,
we have $\asgns'(\kvar) \subseteq \asgns(\kvar)$.

\begin{theorem}[\Fixname-Completeness \cite{Rondon08}]
\label{thm:fix:complete}
If \Oraclename~ is complete and $\asgns \doteq \PredAbs{\ecstr{\kvar}{c}}{\quals}$,
then $\asgns \preceq \asgns'$ for every $\quals$-assignment $\asgns'$
with $\cSat{\asgns'}{c}$.
\end{theorem}
 \section{Implementation}
\label{sec:implementation}

We implemented \sys in approximately 9,200 lines of \sclean.
\sys integrates native solvers (\Zapname and \Fixname) with external
solvers (Spacer \citep{Spacer1} and Eldarica \citep{ELDARICA}) in a
unifying \tFlex tactic.
Horn solving is divided into \emph{synthesis}, whereby a solver
proposes solutions for Horn variables, and \emph{certification},
whereby \sys proves that the residual constraint implies the original
CHC.

\mypara{Native Solvers.}
The implementations of \Zapname and \Fixname closely follow
the algorithms described in \cref{sec:zap} and \cref{sec:fixpoint}.
We instantiate the \Fixname oracle in the weakening loop with a query
to Z3~\cite{z3} on an SMT-LIB encoding of each goal, falling back
to native \sclean automation when Z3 is inconclusive.
The oracle need not be trusted as it only selects candidates,
which certification must still prove.

\mypara{External Solvers.}
\phantomsection\label{sec:external-oracles}
\sys invokes Spacer~\cite{Spacer1} and Eldarica~\cite{ELDARICA}
on SMT-LIB exports of supported constraints, and reconstructs
the computed solutions as \sclean predicates.
The kernel checks the resulting proofs, so neither the solvers
nor the translation are trusted.

\mypara{Discharging residual VC.}
After eliminating $\kappa$-variables in a CHC, we are left
with a residual VC to prove in \sclean.
\tFlex tries to construct a proof for the residual with a combination of
native \sclean tactics (\lean{simp}, \lean{omega}, \lean{grind}, \lean{aesop})
and \texttt{lean-smt}~\cite{lean-smt}, which reconstructs cvc5~\cite{cvc5}
proofs as \sclean terms.
When it succeeds, \tFlex yields a \emph{full} proof of the original
constraint.
Otherwise, it may still make \emph{partial} progress if a Horn solution
is proven strong enough.
This is true by construction for solutions found by \Zapname and the
external solvers when the SMT-LIB encoding is exact.
Because \Fixname's completeness is relative to qualifiers and oracle,
its solution may be too weak, so \tFlex discards any \Fixname solution
that cannot be certified.

\mypara{Scraping Qualifiers}
\label{sec:scraping}
The performance of \Fixname depends on the \emph{qualifiers}
supplied to it. \sys supports user-defined qualifiers over arbitrary
\sclean definitions and automates discovery by \emph{scraping} CHCs
and selected native definitions.
Scraping abstracts predicates into templates, combines arithmetic atoms
into guarded qualifiers, and inspects defining equations for further candidates.
\PredAbsname checks their instances during inference.
To control the search space, \sys starts with the built-in bank and
progressively deepens scraping when a level fails to yield an accepted invariant.
Detailed rules appear in \Cref{app:scraping};
RQ4 (\hyperref[subsec:rq4-scraping]{\S\ref*{subsec:rq4-scraping}})
evaluates the scraping strategies. \section{Evaluation}
\label{sec:evaluation}

We evaluate \sys across four research questions:

\begin{itemize}
  \item \textbf{RQ1:} How effective is \sys on expressive case studies? (\hyperref[subsec:rq1-expr]{\S\ref*{subsec:rq1-expr}})
  \item \textbf{RQ2:} How efficient is \Zapname{} elimination compared with search-based tactics in \sclean? (\hyperref[sec:rq2]{\S\ref*{sec:rq2}})
  \item \textbf{RQ3:} How much automated proof coverage does \sys provide? (\hyperref[subsec:rq3-auto]{\S\ref*{subsec:rq3-auto}})
  \item \textbf{RQ4:} How do different \emph{qualifier scraping} strategies affect solver performance? (\hyperref[subsec:rq4-scraping]{\S\ref*{subsec:rq4-scraping}})
\end{itemize}
To evaluate RQ2 and RQ3 we draw CHCs from two sources.
First, we use 540 nested CHCs of the kind refinement type checkers generate:
337 from Flux's regression tests, 96 from the Flux guide including
implementation of vector-manipulating algorithms, 36 from the
verification of WaVe~\cite{wave,FluxRS}, and 71 from Liquid
Fixpoint's regression tests. To evaluate RQ4, we select set of fully proved
CHCs by \Fixname to enable ablation on different scraping strategies.
Second, to complement the nested constraints, we use the 1,559 distinct conventional
clause-based CHCs that CHC-COMP 2026~\citep{chccomp26} labels satisfiable,
covering linear
and nonlinear Horn clauses over linear integer arithmetic.

We run our experiments on a workstation running Ubuntu 24.04.4 LTS,
with an AMD Ryzen 7 8845HS processor (8 cores, 16 threads) and 27\,GB of RAM.

\subsection{\texorpdfstring{RQ1: Expressive Invariant Inference and Elimination by \sys}{RQ1: Expressive Invariant Inference and Elimination by Flex}}
\label{subsec:rq1-expr}

\begin{table}[t]
\centering\footnotesize
\setlength{\tabcolsep}{2.5pt}
\begin{tabular}{@{}l rrr r rrrr rrr rrr@{}}
\toprule
 & & & &
 & \multicolumn{4}{c}{\Zapname{} + \Fixname{}} & \multicolumn{3}{c}{Spacer} & \multicolumn{3}{c}{Eldarica} \\
\cmidrule(l{2pt}r{2pt}){6-9}\cmidrule(l{2pt}r{2pt}){10-12}\cmidrule(l{2pt}r{2pt}){13-15}
Case study & Rust & Spec. & Proof & CHCs
 & Full & Partial & Med. & Total & Full & Med. & Total & Full & Med. & Total \\
\midrule
Fibonacci        & 147 (84) & 50 (11) &  119 &  4 &  3 &  0 &  7.3\,s &   76\,s &  0 &  $<$0.1\,s &  60\,s &  1 &     0.7\,s &    4\,s \\
Sorting          & 273 (175) & 103 (37) &  620 & 11 &  6 &  5 & 17.0\,s &  196\,s &  0 &     0.2\,s &  63\,s &  6 &     4.7\,s &  269\,s \\
RingBuffer       & 108 (65) & 152 (29) &  501 &  2 &  0 &  2 & 19.8\,s &   40\,s &  0 &     0.4\,s &   1\,s &  0 &    31.4\,s &   63\,s \\
HashTable        & 260 (168) & 302 (156) & 1361 & 12 &  7 &  4 &  5.6\,s &  250\,s &  1 &     0.2\,s &   4\,s &  0 &     3.0\,s &  258\,s \\
VecDeque         & 3147 (215) & 1269 (109) & 1060 & 11 &  0 & 11 &  9.5\,s &  199\,s &  0 &     2.2\,s &  89\,s &  0 &    60.2\,s &  550\,s \\
CStr             & 1105 (64) & 247 (19) &  245 &  2 &  2 &  0 & 10.2\,s &   20\,s &  2 &     0.4\,s &   1\,s &  1 &    31.8\,s &   64\,s \\
String searching & 1685 (513) & 659 (157) & 2949 & 17 &  7 & 10 & 22.9\,s &  458\,s &  5 &     0.8\,s &  87\,s &  5 &     6.1\,s &  385\,s \\
\midrule
\textbf{All} & 6725 (1284) & 2782 (518) & 6855 & 59 & 25 & 32 & 14.3\,s & 1238\,s
 & 8 & 0.4\,s & 305\,s & 13 & 6.1\,s & 1592\,s \\
\bottomrule
\end{tabular}
\caption{Case-study sizes and results. Rust: lines of Rust code;
Spec.: lines of \flux specifications; Proof: lines of \sclean proof.
The parenthesized numbers give the Rust and Spec lines
relevant to the code verified in \sclean.
Full: the proof is fully automated;
Partial: the proof is not fully automated, but \Zapname{} eliminates at least one
predicate, leaving a residual proof with fewer predicates. Each solution found by Spacer and Eldarica ended in a
full proof. Med.: median time per CHC including failures and timeouts;
Total: the total time to run \sys.}
\label{tab:case_studies}
\label{tab:rq1-inference}
\end{table}

To evaluate how well \sys infers rich invariants, we extend \flux,
a refinement type checker for Rust, with a \sclean backend in about
1,800 lines of Rust that emits its verification conditions as \sclean
CHCs.
We then study Rust programs whose correctness requires
reasoning about recursive functions, quantified invariants
over collections, or nonlinear arithmetic, which is beyond what \flux's
existing SMT backend can express or automate.
\Cref{tab:case_studies} summarizes the case studies, which total
6,725 lines of Rust, of which 1,284 are relevant to the code verified in \sclean.

\begin{itemize}[leftmargin=*]
  \item \textbf{Fibonacci.} We verify Fibonacci number computations,
  including imperative and memoized versions, against pure recursive
  specifications.

  \item \textbf{Sorting.} We verify that quicksort, merge sort, and
  insertion sort return sorted permutations of their inputs.

  \item \textbf{RingBuffer.} We verify that a ring buffer adapted from
  Tock, a production embedded operating system~\citep{tock-proper},
  avoids reads from uninitialized memory and implements a queue interface.

  \item \textbf{HashTable.} We verify that insertion, lookup, removal,
  and resizing of a hash table with chaining are correct against
  functional \sclean models.

  \item \textbf{VecDeque.} We verify that vector-backed double-ended
  queue operations adapted from Rust's standard library update the
  contents of the deque as expected.

  \item \textbf{CStr.} The safety of C-string operations is a target of
  the Rust standard library verification
  initiative~\citep{10.1007/978-3-032-28079-4_19}. We verify that
  operations adapted from the standard library maintain the invariant
  that a string contains exactly one null byte, at the end.

  \item \textbf{String searching.} String searching is also a target of
  the same initiative. We verify that
  string searchers adapted from Rust's standard library return matches
  that are aligned with UTF-8 character boundaries.
\end{itemize}

\noindent
These case studies produce 126 VCs, including 61 CHCs. We exclude two
trivial tests, \texttt{SortTest1} and \texttt{VectorsTest1}, evaluating 59 CHCs.
All the CHCs fall outside the standard CHC fragment, in which constraints
are quantifier-free formulas over linear arithmetic, possibly with arrays
or algebraic data types~\cite{bjorner2015horn}.

\mypara{Experiment design.}
We run \sys in three configurations that differ only in how solutions
are synthesized: \Zapname followed by \Fixname, Spacer
alone, and Eldarica alone.
\Fixname uses the built-in bank and qualifiers scraped from the CHC
and its referenced \sclean definitions (\Cref{sec:scraping}); we supply
no qualifiers or invariants by hand.
In every configuration, \sys then
certifies the solutions and discharges the
residual in \sclean.
For each CHC, synthesis is limited to 60\,s and a 180\,s
time budget for the entire proof.
Eldarica supports some of the CHCs through
its extensions for quantifiers and arrays, so
we pass it quantified definitions directly.
Natively, Spacer rejects all CHCs but one, which uses only
nonlinear arithmetic, since it supports neither quantifiers in clause
bodies nor recursive definitions.
We therefore encode every definition a solver cannot handle, such as a
recursive function, as an uninterpreted function, in the hope that the
invariant does not require unfolding it.
Following \citet{bjorner2015horn}, each such function becomes
an array over its parameters, universally
quantified in every clause and passed to every Horn variable as an
extra parameter.

\mypara{Results.}
Of the 59 CHCs, \Zapname + \Fixname fully proves 25, while Spacer
proves 8 and Eldarica 13.
The solved sets largely coincide: every CHC proved by Spacer or
Eldarica is also proved by \Zapname + \Fixname, with a single
exception that only Eldarica proves.
In addition, \Zapname + \Fixname makes partial progress on 32 more.
For instance, on VecDeque, \Zapname + \Fixname makes partial
progress on all 11 CHCs; leaving a residual proof with strictly fewer
Horn variables which we then prove interactively.
Spacer and Eldarica make no partial progress: each solution they find
ends in a full proof.
\Zapname + \Fixname takes a median of 14.3\,s per CHC (1,238\,s in
total), including certification and residual discharge.
Spacer is significantly faster, with a median of 0.4\,s mostly
because it fails quickly.
Eldarica's median time is faster than \Zapname + \Fixname, but in total
Eldarica takes 29\% longer than \Zapname + \Fixname.

\subsection{RQ2: Efficiency of \Zapname~}
\label{sec:rq2}

\begin{table}[t]
\centering\footnotesize
\setlength{\tabcolsep}{2pt}
\newcommand{\rqTwoReduced}[2]{#1\,(#2\%)}
\makebox[\linewidth][c]{\begin{tabular}{@{}ll*{12}{r}@{}}
\toprule
 & & \multicolumn{3}{c}{CHCs} & \multicolumn{3}{c}{\zap} & \multicolumn{3}{c}{\fusiong} & \multicolumn{3}{c}{\fusiona} \\
\cmidrule(lr){3-5}\cmidrule(lr){6-8}\cmidrule(lr){9-11}\cmidrule(lr){12-14}
 & Suite & Total & Acyclic & Mixed & Reduced & hb & ms & Reduced & $\times$hb & $\times$ms & Reduced & $\times$hb & $\times$ms \\
\midrule
& \fluxtests & 295 & 240 & 55 & \rqTwoReduced{295}{100} & 42.1 & 6.8 & \rqTwoReduced{216}{73.2} & 44.9$\times$ & 17.0$\times$ & \rqTwoReduced{250}{84.7} & 85.1$\times$ & 31.6$\times$ \\
& \fluxmedium & 80 & 36 & 44 & \rqTwoReduced{80}{100} & 185 & 61.7 & \rqTwoReduced{46}{57.5} & 141$\times$ & 54.1$\times$ & \rqTwoReduced{48}{60} & 162$\times$ & 63.7$\times$ \\
& \wavesuite & 35 & 26 & 9 & \rqTwoReduced{35}{100} & 174 & 36.2 & \rqTwoReduced{19}{54.3} & 94.6$\times$ & 32.6$\times$ & \rqTwoReduced{20}{57.1} & 125$\times$ & 44.4$\times$ \\
& \lfsuite & 57 & 50 & 7 & \rqTwoReduced{57}{100} & 125 & 42.2 & \rqTwoReduced{40}{70.2} & 33.8$\times$ & 12.8$\times$ & \rqTwoReduced{49}{86} & 65.2$\times$ & 26.7$\times$ \\
& Case studies & 47 & 27 & 20 & \rqTwoReduced{47}{100} & 1{,}167 & 563 & \rqTwoReduced{9}{19.1} & 432$\times$ & 133$\times$ & \rqTwoReduced{9}{19.1} & 628$\times$ & 181$\times$ \\
\midrule
\multirow{8}{*}{\rotatebox{90}{CHC-COMP}} & Eldarica & 97 & 0 & 97 & \rqTwoReduced{97}{100} & 6{,}628 & 2{,}215 & \rqTwoReduced{16}{16.5} & 128$\times$ & 66.7$\times$ & \rqTwoReduced{4}{4.1} & 297$\times$ & 176$\times$ \\
& HCAI & 104 & 0 & 104 & \rqTwoReduced{104}{100} & 203 & 31.4 & \rqTwoReduced{26}{25} & 228$\times$ & 118$\times$ & \rqTwoReduced{14}{13.5} & 1{,}461$\times$ & 549$\times$ \\
& HoPV & 78 & 32 & 46 & \rqTwoReduced{78}{100} & 533 & 101 & \rqTwoReduced{28}{35.9} & 189$\times$ & 81.0$\times$ & \rqTwoReduced{25}{32.1} & 266$\times$ & 128$\times$ \\
& Kind2 & 395 & 0 & 395 & \rqTwoReduced{393}{99.5} & 1{,}516 & 238 & \rqTwoReduced{0}{0} & -- & -- & \rqTwoReduced{0}{0} & -- & -- \\
& LIA-XS & 6 & 0 & 6 & \rqTwoReduced{6}{100} & 37.3 & 4.7 & \rqTwoReduced{6}{100} & 155$\times$ & 73.4$\times$ & \rqTwoReduced{6}{100} & 892$\times$ & 398$\times$ \\
& LLReve & 11 & 2 & 9 & \rqTwoReduced{11}{100} & 1{,}375 & 283 & \rqTwoReduced{2}{18.2} & 534$\times$ & 282$\times$ & \rqTwoReduced{0}{0} & -- & -- \\
& RustHorn & 10 & 2 & 8 & \rqTwoReduced{10}{100} & 254 & 35.9 & \rqTwoReduced{0}{0} & -- & -- & \rqTwoReduced{0}{0} & -- & -- \\
& Synthesis & 60 & 0 & 60 & \rqTwoReduced{59}{98.3} & 10{,}149 & 4{,}078 & \rqTwoReduced{2}{3.3} & 371$\times$ & 202$\times$ & \rqTwoReduced{1}{1.7} & 398$\times$ & 161$\times$ \\
\midrule
& \textbf{All} & 1{,}275 & 415 & 860 & \rqTwoReduced{1{,}272}{99.8} & 1{,}583 & 472 & \rqTwoReduced{410}{32.2} & 71.7$\times$ & 28.2$\times$ & \rqTwoReduced{426}{33.4} & 116$\times$ & 44.6$\times$ \\
\bottomrule
\end{tabular}}
\caption{Reduced: count (percentage of row total).
\emph{hb}, \emph{ms}: mean heartbeats and milliseconds over successful reductions.
$\times$hb, $\times$ms: respective cost ratios to \zap{} on jointly reduced CHCs.}
\label{tab:rq2-flux-agg}
\end{table}

\Zapname{} (\Cref{sec:zap}) instantiates acyclic $\kvar$-variables
with their strongest solutions and constructs proofs justifying
replacement of their head occurrences with $\top$.
We compare this deterministic proof construction with proof search
using Lean's \texttt{grind} and \texttt{aesop} tactics.

\mypara{Experiment design.}
We implemented two variants, \fusiong{} and \fusiona{}, which compute
the same strongest solution as \zap{} using $\Solname$, but replace
$\Certname$'s proof-term construction with calls to \texttt{grind}
and \texttt{aesop}, respectively, each using its default settings.
We run all three variants on 1,275 CHCs with eliminable variables:
415 acyclic-only and 860 mixed.
These comprise 514 CHCs from the refinement-typechecker suites
and case studies introduced above, and 761 from CHC-COMP.
All variants use identical resource limits.

\mypara{Results.}
\Cref{tab:rq2-flux-agg} reports successful reductions and, over those,
mean wall-clock time and heartbeats---Lean's deterministic proxy
for proof effort.
\zap{} reduces 1,272 constraints (99.8\%), averaging 472\,ms.
Its three failures exceed the 60\,s timeout; \fusiong{} and
\fusiona{} also fail on these constraints.
\fusiong{} reduces only 410 constraints (32.2\%); on shared successes
with \zap{}, it consumes $71.7\times$ as many heartbeats and takes
$28.2\times$ as long.
\fusiona{} fares slightly better in coverage, reducing 426 constraints
(33.4\%), but consumes $116\times$ as many heartbeats and takes
$44.6\times$ as long on its shared successes with \zap{}.
Both variants produce proof terms 40\% deeper than \zap{}'s,
taking $3.7\times$ and $2.5\times$ as long for the kernel to check,
respectively
(\Cref{tab:rq2-flux-agg-term} in \Cref{app:rq2-proof-metrics}).
These proof-term comparisons use geometric means over jointly reduced
constraints and include generated helper proofs.
Thus, \zap{}'s direct proof construction achieves higher coverage
and greater efficiency than search-based automation.

\subsection{RQ3: Automated Proof Coverage}
\label{subsec:rq3-auto}

Refinement type checkers and CHC solvers traditionally rely on SMT
for automation. We evaluate how much of this automation \sys retains
as a \emph{foundational} solver in \sclean, using refinement-typechecker
CHCs and CHC-COMP.

\mypara{Experiment Design.}
We run three configurations of \sys separately: \zap{} + \Fixname{},
and the certifying pipeline with Spacer or Eldarica as its external
solver. \Fixname{} progressively deepens scraping when synthesis
fails (\S\ref{sec:scraping}).
All configurations use identical preprocessing and discharge the remaining
goals with the same combination of \lean{simp}, \lean{grind}, \lean{omega},
\lean{aesop} and \texttt{lean-smt}.
Each configuration has a 180\,s budget.
Each Spacer or Eldarica call is limited to 90\,s.
Tactic times exclude \sclean start-up.

\phantomsection\label{sec:rq3a}
\mypara{Refinement-Typechecker CHCs.}
We run each configuration on the 540 CHCs
from \fluxtests, \fluxmedium, \wavesuite, and \lfsuite{} introduced above.
We exclude the case studies since their CHCs were never amenable to SMT automation.

\begin{table}[t]
\centering\footnotesize
\setlength{\tabcolsep}{4pt}
\newcommand{\rqThreeFull}[2]{#1\,(#2\%)}
\begin{tabular}{@{}l r rrr rr rr@{}}
\toprule
 & & \multicolumn{3}{c}{\zap{} + \Fixname{}} & \multicolumn{2}{c}{\sys{} Spacer} & \multicolumn{2}{c}{\sys{} Eldarica} \\
\cmidrule(lr){3-5}\cmidrule(lr){6-7}\cmidrule(lr){8-9}
Suite & CHCs & Full & Partial & Med.\ (s) & Full & Med.\ (s) & Full & Med.\ (s) \\
\midrule
\fluxtests  & 337 & \rqThreeFull{333}{98.8} & 2 & 0.05 & \rqThreeFull{310}{92.0} & 0.12 & \rqThreeFull{320}{95.0} & 0.78 \\
\fluxmedium &  96 &  \rqThreeFull{94}{97.9} & 0 & 0.32 &  \rqThreeFull{91}{94.8} & 0.23 &  \rqThreeFull{91}{94.8} & 1.05 \\
\wavesuite  &  36 &  \rqThreeFull{35}{97.2} & 0 & 0.31 &  \rqThreeFull{35}{97.2} & 0.27 &  \rqThreeFull{36}{100} & 1.44 \\
\lfsuite    &  71 &  \rqThreeFull{67}{94.4} & 3 & 0.04 &  \rqThreeFull{65}{91.5} & 0.12 &  \rqThreeFull{65}{91.5} & 0.73 \\
\midrule
\textbf{All} & 540 & \rqThreeFull{\textbf{529}}{98.0} & 5 & 0.08 & \rqThreeFull{\textbf{501}}{92.8} & 0.15 & \rqThreeFull{\textbf{512}}{94.8} & 0.81 \\
\bottomrule
\end{tabular}
\caption{Coverage on 540 refinement-typechecker CHCs: CHCs with a full proof (Full, with
percentage) and the median tactic time over the CHCs each configuration fully proves (Med.).
For \zap{} + \Fixname{}, Partial counts CHCs not fully proved in which some Horn variables are eliminated by \zap{}.}
\label{tab:rq3a-coverage}
\end{table}

\mypara{Results. }
All three configurations retain most of the automation on these CHCs
(\Cref{tab:rq3a-coverage}).
\zap{} + \Fixname{} fully proves 529 of the 540 CHCs (98.0\%),
Spacer 501 (92.8\%), and Eldarica 512 (94.8\%),
each with a median tactic time below one second.
\zap{} + \Fixname{} additionally makes partial progress on five of the remaining 11 CHCs,
leaving behind fewer Horn variables.
The 11 CHCs not proved by \zap{} + \Fixname{} encounter timeouts,
failed invariant inference, or residual obligations that cannot be handled automatically such as
nonlinear arithmetic.
For Spacer and Eldarica, nine cases fail during encoding,
and one further case is not solved after translation.
The remaining failures arise from timeouts, failed
certification, or residual goals that Lean automation
cannot discharge.
\zap{} + \Fixname has the lowest median tactic time, about half of Spacer's and one tenth of Eldarica's.
Over all 540 CHCs it also spends the least time in total, since both solvers spend most of their time on CHCs they do not prove.
On the 491 CHCs that all three prove, however, Spacer is fastest in total, because a few CHCs take Zap + Fix much longer.
\Cref{tab:rq3a-time} in \Cref{app:rq3-details} gives the full timing statistics.

\phantomsection\label{sec:rq3b}
\mypara{CHC-COMP.}
We run the same three configurations on the 1,559 satisfiable CHCs from
CHC-COMP introduced above.
As a baseline, we run Spacer and Eldarica natively on the original
SMT-LIB inputs, with a 90\,s limit each.

\mypara{Results.}
Natively, Spacer solves 1,198 CHCs and Eldarica 937;
together they solve 1,354 (86.9\%; \Cref{tab:rq3b-chccomp}
in \Cref{app:rq3-details}).
\sys retains most of this automation with foundational proofs:
\sys{} Spacer fully proves 983 of the 1,198 CHCs Spacer solves
natively (82.1\%), and \sys{} Eldarica 799 of the 937 Eldarica
solves natively (85.3\%).
In total, \sys{} Spacer fully proves 997 CHCs (64.0\%) and
\sys{} Eldarica 822 (52.7\%), with median tactic times of
2.9\,s and 4.0\,s.
Together they prove 1,129 CHCs (72.4\%): 1,118 of the 1,354
natively solved CHCs and 11 that neither native solver solves.
\zap{} + \Fixname{} fully proves 221 CHCs (14.2\%), with a median
tactic time of 1.2\,s, and makes partial progress on 654 others,
almost all due to \zap{} eliminating Horn variables.
Adding \zap{} + \Fixname{} to the two backends raises coverage to 1,133 CHCs.
Spacer and Eldarica failures on CHC-COMP are analyzed in \Cref{app:rq3-details}.

\subsection{RQ4: Scraping Strategies}
\label{subsec:rq4-scraping}

\begin{figure}[tbp]
\centering
\begin{subfigure}[t]{.49\linewidth}
\centering
\includegraphics[width=\linewidth]{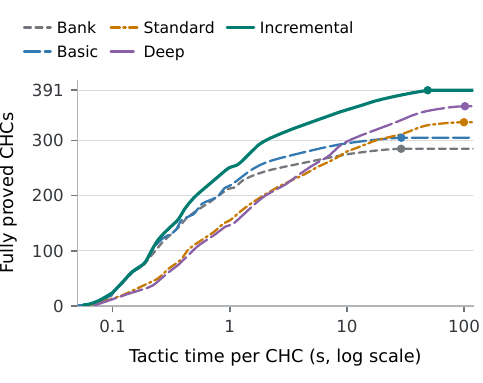}
\caption{CHCs proved within each time budget.}
\label{fig:scraping-coverage}
\end{subfigure}\hfill
\begin{subfigure}[t]{.49\linewidth}
\centering
\includegraphics[width=\linewidth]{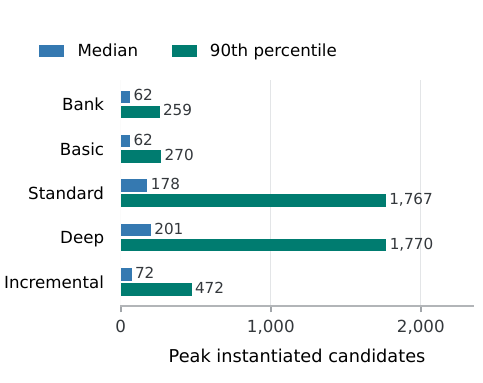}
\caption{Candidate-pool sizes: median and 90th percentile.}
\label{fig:scraping-candidates}
\end{subfigure}
\caption{Scraping configurations on 391 selected CHCs. Coverage curves use
monotone interpolation; dots mark each configuration's last observed proof.
Candidate counts include unsuccessful runs.}
\label{fig:rq4-scraping}
\end{figure}

We evaluate how different scraping strategies affect \Fixname's solving
time and candidate-pool size. Specifically, we ask two questions:
(1) How does \emph{incremental} scraping compare with applying a fixed
scraping level directly, in terms of solving time and candidate-pool size?
(2) How do solving time and candidate-pool size change across the fixed
scraping levels?

\mypara{Experimental Setup.}
We select 391 CHCs: 181 from the refinement-typechecker suites, 192 from
CHC-COMP, and 18 from RQ1. Each selected CHC contains at least one
\emph{cyclic} Horn variable and was fully proved using only \Zapname{}
and \Fixname{} with the default incremental scraping strategy.
The cyclic-variable requirement ensures that \Fixname{} is used, and thus,
allows us to perform ablation across different strategies.
We evaluate four fixed scraping levels (\Cref{app:scraping}).
\emph{Bank} uses 19 built-in templates; \emph{Basic} adds predicates
from CHC heads.
\emph{Standard} also scrapes premises and directly referenced
\sclean{} definitions and builds guarded arithmetic qualifiers.
\emph{Deep} explores helper definitions and adds generalizations
such as quantified bounds and accumulator relations.
We compare these with our default, \emph{Incremental}, which tries
the levels in order, retaining earlier candidates and advancing
only when no invariant is accepted.
Each run applies \Zapname{} followed by \Fixname{}, with a 60\,s
budget for \Zapname{} and a 120\,s budget shared across
\Fixname{}'s scraping levels.
We report the number of fully proved CHCs, tactic time on those CHCs,
and peak candidate-pool size. For each run, the latter is the maximum
over PA stages of the initial candidate count summed across predicates
(\Cref{app:scraping-ablation}).

\mypara{Results.}
Compared with applying Deep directly, Incremental uses smaller candidate
pools and reduces solving time on jointly proved CHCs
(\Cref{fig:rq4-scraping}).
Its median and 90th-percentile peak candidate-pool sizes are 72 and 472,
versus 201 and 1,770 for Deep, with a geometric-mean speedup of
$2.70\times$ on shared successes.
Incremental proves all 391 selected CHCs, compared with 285 for Bank
and 362 for Deep; its median and 90th-percentile tactic times are
0.50\,s and 8.82\,s.
Across the fixed levels, deeper scraping increases candidate-pool size
and coverage, accompanied by higher observed median solving times.
For Bank, Basic, Standard, and Deep, median peak candidate-pool sizes
are 62, 62, 178, and 201, respectively, while the numbers of fully
proved CHCs rise from 285 to 305, 333, and 362.
Median tactic times are 0.37\,s, 0.38\,s, 1.16\,s, and 1.56\,s,
respectively. Timing medians use each configuration's own successful CHCs.
\Cref{tab:scraping-times} in \Cref{app:scraping-ablation}
reports the full timing distributions.
 \section{Related Work}
\label{sec:related}

\mypara{CHC Solvers}
\sys is inspired by \cite{Cosman17} to solve acyclic variables,
and predicate abstraction \cite{Houdini} to compute (overapproximate)
solutions for cyclic ones. \Zapname additionally constructs exact proof terms
certifying acyclic \emph{Horn} elimination.
Spacer \cite{Spacer1,Spacer2}, Eldarica \cite{ELDARICA},
and SeaHorn \cite{SeaHorn} solve CHCs by inferring invariants using
IC3 \cite{IC3} or Lazy Abstraction and Craig Interpolation
\cite{BLAST, IMPACT}, which do not require qualifiers,
but which may diverge.
A direction for future work is \emph{proof-producing} IC3-style or
interpolation-based solving over native \sclean definitions.
Several other SMT-based CHC solvers have been developed
\cite{Dietsch_2019,Holce,Golem2025}.
To our knowledge, \sys is the first CHC solver to operate
directly over native ITP propositions while producing
kernel-checkable proofs.

\mypara{Proving or Certifying VCs}
Work on proof-carrying code
showed how to avoid trusting SMT solvers
by modifying decision procedures to emit
certificates \cite{necula1998compiling}.
This idea was extended to model checkers
\cite{blastpcc}, and refinement types
\cite{fine}, and several modern theorem
provers emit independently checkable
certificates \cite{bártek2025vampirediary,vampire-lean,smt-certificates, SMTCoq,lean-smt}.
Alternatively, interactive proof assistants can discharge VCs:
\textsc{Jahob} \cite{jahob} delegates VCs to backends including \scisabelle,
while the \sccreusot Rust verifier \cite{Creusot} emits VCs in the \scWhy
format for SMT or interactive proof \cite{Why3}.
These techniques apply to plain VCs; \sys uses such provers as
\emph{oracles} to obtain certificates for CHC satisfiability.

\mypara{SMT-based Program Verifiers}
Program-logic-based verifiers~\cite{Floyd1967, Hoare1969, dijkstra1976}
typically emit Horn-variable-free VCs discharged by
SMT solvers~\cite{Luckham1979, escjava, dafny, framac}.
\scprusti \cite{prusti} and \scverus \cite{verus} implement
such Floyd-Hoare style verifiers for Rust.
\fstar generalizes the approach to the higher order setting
using Dijkstra monads \cite{F*-dijkstra-monads}, while
\stainless \cite{stainless} uses refinement types.
In contrast to the above, \sclh \cite{LiquidHaskell},
\scflux \cite{FluxRS} and \scthrust \cite{Thrust} use
CHCs, which generalize VCs with existentially quantified
variables that delegate the synthesis of refinements
or invariants to an SMT-based CHC solver.
When incomplete SMT heuristics \cite{leino2016trigger} are opaque or
hard to use \cite{mariposa}, \sys supports interactive proof while
providing foundational guarantees of CHC satisfiability.
 
\section*{AI Usage}
The project has been ongoing for more than a year.
The high-level architecture of \sys, the designs of \Zapname{}
and predicate abstraction, and the evaluation were human-led.
Over time, we used AI models, including Opus, to assist with
implementation and mechanization, develop evaluation harnesses, and identify
performance bottlenecks.

\section*{Data-Availability Statement}
All \sclean code for \sys is available in the supplementary material.
Rust case studies, benchmarks, and evaluation setup will be submitted
with artifact evaluation.

\newpage
\appendix
\newcommand{\restr}[2]{#1 \setminus #2}            

\section{Zap: Soundness}
\label{app:zap-proofs}

This appendix proves the correctness of the \Zapname\ algorithm of
\Cref{sec:zap}, whose two results are \emph{Soundness}
(\Cref{thm:zap:sound}) and \emph{Equivalence} (\Cref{thm:zap:equiv}).

\mypara{Overview}
Equivalence, together with the fact that \Solname\ computes the strongest
valid solution, is the \textsc{Fusion} result of \citet{Cosman17}.
Soundness we prove as a pure typing result: the proof term \pf\ emitted by
$\Zap{\Cstr}$ is a closed, well-typed term of type $\Cstr' \to \Cstr$, so
the \sclean\ kernel turns any proof of the residual $\Cstr'$ into a proof
of the original $\Cstr$.

\mypara{Setup}
Throughout the appendix, we fix a closed constraint
\[
  \Cstr \;\equiv\; \ecstr{\kvar}{c},
\]
whose Horn variables $\overline{\kvar}$ are partitioned
(\Cref{sec:partition}) into the cut set $\kcycs{\kvar}$ and the acyclic
set $\kacys{\kvar}$.

\mypara{Organization}
\begin{itemize}
  \item \Cref{app:zap-prelim} fixes the conventions and the proof-term calculus;
  \item \Cref{app:zap-sol} establishes the shape of the synthesized solution $\asgn$;
  \item \Cref{app:zap-elim} characterizes \Remname\ and the elimination order;
  \item \Cref{app:zap-soundness} types \Navname, \Certname, and the emitted bridge.
\end{itemize}

\subsection{Conventions}
\label{app:zap-prelim}

\mypara{Definitions and uses}
A \emph{$\kvar$-definition} of a constraint $c$ is a head occurrence of
\kvar, i.e., a sub-constraint $\HeadOf{c}{\pfx} = \kvapp{\kvar}{\overline{t}}$
for some head-prefix $\pfx$; a \emph{$\kvar$-use} is an occurrence of
\kvar\ in a guard.
In input constraints, guards are atoms (\Cref{fig:chc-grammar}), so every
use is a single application $\kvapp{\kvar}{\overline{t}}$ occurring negatively, and Horn
variables occur nowhere else (\Cref{sec:chc-syntax}); in particular they
never occur inside argument terms.

\mypara{Validity}
Validity is \sclean\ provability: we write $\ctx \vdash \varphi$ for
derivability under a context and use the connectives' introduction and
elimination laws without comment.
The judgment reads the binders and hypotheses of $\ctx$ and ignores the
routing marks.

\mypara{Scope, contexts, truncation}
For an acyclic \kvar\ in $c$, let $\pfx = \Scope{\kvar}{c}$
(\Cref{fig:scope}). We use the following facts.
\begin{itemize}
  \item $\HeadOf{c}{\pfx}$ is the sub-constraint at the lowest common
        ancestor (LCA) of \kvar's occurrences, and $\CtxOf{c}{\pfx}$ is
        the context above it.
  \item A context (\Cref{sec:chc-props}) lists binders $\cBind{x}{b}$,
        hypotheses $\cBind{\hyp}{\pred}$, and marks $\pLeft/\pRight$; each
        entry \emph{matches} the prefix step that produced it:
        $\cBind{x}{b}$ matches $\pBind{x}{b}$, $\cBind{\hyp}{\pred}$ matches
        $\pAsm{\pred}$, and a mark matches itself.
  \item Given a context $\ctx$ along a path extending $\pfx$, the
        \emph{truncation} $\restr{\ctx}{\pfx}$ drops the first $|\pfx|$
        entries; this is the context \Certname\ hands \Navname\ at a
        divergent head (\Cref{fig:certify}).
  \item By well-formedness (\Cref{sec:chc-props}), every application of
        \kvar\ has $\Rank{\kvar}$ leading arguments followed by the
        binders $\overline{x} = \Binders{\CtxOf{c}{\pfx}}$, innermost
        first, as trailing arguments.
\end{itemize}

\begin{figure}[t]
\begin{mathpar}
\inferrule{\cBind{x}{b} \in \ctx}{\ctx \vdash x : b}
\and
\inferrule{\cBind{\hyp}{\varphi} \in \ctx}{\ctx \vdash \hyp : \varphi}
\and
\inferrule{\ctx, \cBind{x}{b} \vdash \pf : \varphi}{\ctx \vdash \cLam{x}{\pf} : \cAll{x}{b}{\varphi}}
\and
\inferrule{\ctx \vdash \pf : \cAll{x}{b}{\varphi} \\ \ctx \vdash t : b}{\ctx \vdash \pf\,t : \varphi\sub{x}{t}}
\and
\inferrule{\ctx, \cBind{\hyp}{g} \vdash \pf : \varphi}{\ctx \vdash \cLam{\hyp}{\pf} : \cImp{g}{\varphi}}
\and
\inferrule{\ctx \vdash \pf : \cImp{g}{\varphi} \\ \ctx \vdash \pf' : g}{\ctx \vdash \pf\,\pf' : \varphi}
\and
\inferrule{\ctx \vdash \pf : \varphi_1 \\ \ctx \vdash \pf' : \varphi_2}{\ctx \vdash \langle \pf, \pf'\rangle : \varphi_1 \wedge \varphi_2}
\and
\inferrule{\ctx \vdash \pf : \varphi_1 \wedge \varphi_2}{\ctx \vdash \pf.i : \varphi_i}\;(i \in \{1,2\})
\and
\inferrule{\ctx \vdash \pf : \varphi_1}{\ctx \vdash \mathsf{inl}\,\pf : \varphi_1 \vee \varphi_2}
\and
\inferrule{\ctx \vdash \pf : \varphi_2}{\ctx \vdash \mathsf{inr}\,\pf : \varphi_1 \vee \varphi_2}
\and
\inferrule{\ctx \vdash t : b \\ \ctx \vdash \pf : \varphi\sub{z}{t}}{\ctx \vdash \langle t, \pf\rangle : \cExi{z}{b}{\varphi}}
\and
\inferrule{\ctx \vdash \pf : \cExi{z}{b}{\varphi} \\ \ctx, \cBind{z}{b}, \cBind{\hyp}{\varphi} \vdash \pf' : \psi}{\ctx \vdash \letin{\langle z, \hyp\rangle = \pf}{\pf'} : \psi}\;(z,\hyp \notin \mathrm{FV}(\psi))
\and
\inferrule{\ }{\ctx \vdash \rfl : t = t}
\and
\inferrule{\ }{\ctx \vdash \langle\rangle : \top}
\end{mathpar}
\caption{Typing of the proof-term fragment of \Cref{fig:zap-proofterms}.}
\label{fig:zap-typing}
\end{figure}

\mypara{Typing the proof-term fragment}
\Cref{fig:zap-typing} fixes the typing judgment $\ctx \vdash \pf : \varphi$
for the proof terms of \Cref{fig:zap-proofterms}: the introduction and
elimination forms of the constraint logic.
The following conventions apply:
\begin{enumerate}
    \item In the $\exists$-introduction and $\exists$-elimination rules, the
        sort $b$ annotating the existential binder $\cExi{z}{b}{\varphi}$
        ranges over base sorts and predicate types
        $b_1 \to \cdots \to b_n \to \Prop$.
  \item A bundle $\langle t_1,\ldots,t_k;\, \pf\rangle$ abbreviates
        iterated $\exists$-introduction.
  \item A multi-binder
        $\letin{\langle z_1,\ldots,z_k, \hyp\rangle = \pf}{\pf'}$ abbreviates
        iterated $\exists$-elimination: $k$ nested $\kword{let}$s that
        bind the witnesses $z_1,\ldots,z_k$ one at a time, with $\hyp$
        naming the proof of the fully unpacked body.
  \item The premise $\ctx \vdash t : b$ is \sclean's term-typing
        judgment.
\end{enumerate}
These rules are not the \sclean\ kernel itself but our model of it,
restricted to this fragment and following the kernel formalization of
\citet{lean4lean}; every rule is admissible in the kernel, so a term
well-typed here is accepted by the kernel, and Soundness is proved
against this system.

\subsection{The synthesized solution}
\label{app:zap-sol}

\begin{figure}[t]
\small
\[
\begin{array}{@{}l@{\;\;}c@{\;\;}l@{}}
\toprule
\Solname & \;:\; & (\kvars \times \Cstrs) \to \Terms \\
\midrule
\Sol{\kvar}{\cstr} & & \\
\quad\mid\; \kvar \notin \Heads{\cstr}    & \doteq & \bot \\
\Sol{\kvar}{\cAnd{\cstr_1}{\cstr_2}}      & \doteq & \Sol{\kvar}{\cstr_1} \vee \Sol{\kvar}{\cstr_2} \\
\Sol{\kvar}{\cAll{x}{b}{\cstr}}           & \doteq & \cExi{x}{b}{\Sol{\kvar}{\cstr}} \\
\Sol{\kvar}{\cImp{\pred}{\cstr}}          & \doteq & \cAnd{\pred}{\Sol{\kvar}{\cstr}} \\
\Sol{\kvar}{\kvapp{\kvar}{\overline{t}}}  & \doteq & \textstyle\bigwedge_{i < \Rank{\kvar}} z_i = t_i \\
\bottomrule
\end{array}
\]
\caption{The \Solname\ procedure, reproduced from \Cref{fig:algo:sol}.}
\label{fig:app:sol}
\end{figure}

Fix an acyclic \kvar\ in a constraint $c$, and let $\pfx = \Scope{\kvar}{c}$.
For this \kvar, \Elimname\ (\Cref{fig:algo:elim}) computes the solution
using \Solname\ as in \Cref{fig:app:sol}, with the slot names
$\overline{z}$ fresh for the binders of $c$:
\[
  \asgn \;=\; \cLam{\overline{z}}{\cLam{\overline{x}}{\Sol{\kvar}{\HeadOf{c}{\pfx}}}},
  \qquad
  \overline{z} = z_0, \ldots, z_{\Rank{\kvar}-1},
  \quad
  \overline{x} = \Binders{\CtxOf{c}{\pfx}}.
\]

The soundness proof needs two properties of \asgn:
\begin{enumerate}
  \item \asgn\ must be \kvar-free: it is substituted at the uses of \kvar\
and supplied as the witness for $\exists\kvar$.
    \item The body of \asgn\ must mirror the structure of
    $\HeadOf{c}{\pfx}$: along every path to a \kvar-definition it is
    the same tree with each connective dualized, so \Navname\ can
    walk both in lockstep.
\end{enumerate}
The next lemma establishes both.

\begin{lemma}[Solution shape]
\label{lem:sol-shape}
\leavevmode
\begin{enumerate}
\item[(a)] $\Sol{\kvar}{c^\ast} = \bot$ for every sub-constraint
  $c^\ast$ of $\HeadOf{c}{\pfx}$ containing no $\kvar$-definition.
\item[(b)] On a path from the root of $\HeadOf{c}{\pfx}$ to a
  $\kvar$-definition, \Solname\ dualizes each connective, preserving
  binder names, and emits the slot equalities at the definition:
  \begin{align*}
    \cAnd{c^\ast_1}{c^\ast_2} &\;\mapsto\;
      \Sol{\kvar}{c^\ast_1} \vee \Sol{\kvar}{c^\ast_2} \\
    \cAll{x}{b}{c^\ast}  &\;\mapsto\; \cExi{x}{b}{\Sol{\kvar}{c^\ast}} \\
    \cImp{g}{c^\ast}     &\;\mapsto\; \cAnd{g}{\Sol{\kvar}{c^\ast}} \\
    \kvapp{\kvar}{\overline{t}} &\;\mapsto\;
      \textstyle\bigwedge_{i<\Rank{\kvar}} z_i = t_i
  \end{align*}
\item[(c)] $\mathrm{FV}(\asgn) \subseteq \kvs{c} \setminus \{\kvar\}$;
  in particular, \asgn\ never mentions \kvar.
\end{enumerate}
\end{lemma}

\begin{proof}\leavevmode
\begin{enumerate}
\item[(a)] Immediate from the guarded equation of \Cref{fig:app:sol}:
$\kvar \notin \Heads{c^\ast}$ (\Cref{sec:chc-props}) holds exactly when
$c^\ast$ contains no \kvar-definition.

\item[(b)] Each node on the path has a \kvar-definition in its subtree,
so the guarded equation does not apply: at each node above the
definition, the equation matching the connective applies and copies its
binder or guard verbatim; at the definition itself, the head equation
emits the slot equalities.

\item[(c)] By induction on the sub-constraints $c^\ast$ of
$\HeadOf{c}{\pfx}$, following the equations of \Cref{fig:app:sol},
\[
  \mathrm{FV}(\Sol{\kvar}{c^\ast}) \;\subseteq\;
  (\mathrm{FV}(c^\ast) \setminus \{\kvar\}) \cup \overline{z}.
\]
\begin{itemize}
  \item The guarded equation emits $\bot$, with no free variables.
  \item The $\wedge$ and $\forall$ equations preserve the binding
        structure: the disjunction takes the same union, and the
        $\exists$ binds the same $x$.
  \item The guard equation copies $g$, and $\kvar \notin \mathrm{FV}(g)$:
        a \kvar-use above a \kvar-definition is a self-edge at \kvar\ in
        $\dg{c}$ (\Cref{sec:partition}), so every cut set would contain
        \kvar, contradicting that \kvar\ is acyclic.
  \item The head equation emits the slots $\overline{z}$ and the
        arguments $\overline{t}$, which never contain Horn variables
        (\Cref{app:zap-prelim}).
\end{itemize}
Since $\Cstr$ is closed, every free variable of $\HeadOf{c}{\pfx}$ is a
Horn variable or a binder of $\CtxOf{c}{\pfx}$, i.e., among
$\overline{x}$; the $\cLam{\overline{z}}{\cLam{\overline{x}}{\cdot}}$
closure of \Elimname\ then leaves free only Horn variables other than
\kvar.
\end{enumerate}
\end{proof}

\subsection{Reduction and the elimination order}
\label{app:zap-elim}

We write $\subasgn{c}{\kvar}{\asgn}$ for the substitution of \asgn\ for
every occurrence of \kvar\ in $c$.

Substituting a synthesized solution can produce compound guards
containing cut variables. We admit these formulas as residual guards
and count their Horn-variable occurrences when computing dependencies.
\Solname\ copies each guard as a whole, and \Navname\ uses the
corresponding hypothesis.

\begin{lemma}[Reduction]\label{lem:red}
$\Rem{\kvar}{\asgn}{c}$ (\Cref{fig:algo:rem}) is \kvar-free and
coincides with $\subasgn{c}{\kvar}{\asgn}$, except that each definition
head $\kvapp{\kvar}{\overline{t}}$ becomes $\top$ rather than
$\asgn(\overline{t})$.
\end{lemma}

\begin{proof}
By structural induction on $c$, comparing each equation of
\Cref{fig:algo:rem} with the substitution:
\begin{itemize}
  \item at $\cAnd{c_1}{c_2}$ and $\cAll{x}{b}{c^\ast}$, both sides
        recurse componentwise; at $\cImp{\pred}{c^\ast}$, both also
        substitute the guard $\pred$ identically;
  \item at a definition head $\kvapp{\kvar}{\overline{t}}$, \Remname\
        gives $\top$ and the substitution gives $\asgn(\overline{t})$:
        the sole divergence;
  \item at any other head, both are the identity, since \kvar\ occurs
        neither in a \kvar-free atom nor inside the arguments of a
        foreign application $\kvapp{\kvar'}{\overline{s}}$
        (\Cref{app:zap-prelim}).
\end{itemize}
The result is \kvar-free: heads become $\top$, and uses are replaced by
\asgn, which is \kvar-free (\Cref{lem:sol-shape}).
\end{proof}

The fold of \Cref{fig:algo:zap} eliminates the acyclic variables
innermost first: writing $\kacys{\kvar} = [\kvar_1, \ldots, \kvar_n]$ in
elimination order, $c_0 = c$ and $c_i = \Rem{\kvar_i}{\asgn_i}{c_{i-1}}$,
with $\asgn_i$ synthesized on the residual $c_{i-1}$. The partition
orders the acyclic set topologically (\Cref{sec:partition}): every
variable that $\kvar_i$ depends on is some $\kvar_j$ with $j < i$ or a
cut variable.

\begin{lemma}[Solutions are cut-only]\label{lem:afree}
For each $i$, $\mathrm{FV}(\asgn_i) \subseteq \kcycs{\kvar}$.
\end{lemma}

\begin{proof}
By induction on $i$: assuming the lemma for all $j < i$, we prove
$\mathrm{FV}(\asgn_i) \subseteq \kcycs{\kvar}$ in two steps:

\begin{itemize}
  \item \emph{$\kvar_i$ is acyclic in $c_{i-1}$.} \Remname\ preserves the
  skeleton, and the spliced $\asgn_j$ ($j < i$) are cut-only by induction;
  so every edge of $\dg{c_{i-1}}$ out of an acyclic variable is already in
  $\dg{c}$. Hence $\kcycs{\kvar}$ remains a cut set of $c_{i-1}$, $\kvar_i$
  remains acyclic, and \Cref{lem:sol-shape} applies at step $i$.
  \item \emph{No acyclic variable is free in $\asgn_i$.} By
  \Cref{lem:sol-shape}(c), the free variables of $\asgn_i$ are Horn
  variables of $c_{i-1}$; by (b), each occurs in a guard above a
  $\kvar_i$-definition. An acyclic one is impossible: its use is inherited
  from $c$ (first step), so $\kvar_i$ depends on it in $c$ and, by
  topological order, it is some $\kvar_j$ with $j < i$; but $c_{i-1}$ is
  $\kvar_j$-free (\Cref{lem:red}; the splices are cut-only). Hence all are
  cut variables.
\end{itemize}

\end{proof}

\subsection{Soundness}
\label{app:zap-soundness}
Finally, we bring our attention to Soundness (\Cref{thm:zap:sound}) for \Zapname.
It is stated under the typing relation, shown in \Cref{fig:zap-typing} for the fragment
of \sclean proof terms in \Cref{fig:zap-proofterms}.

\Certname\ accumulates one context entry per node it descends past,
whereas $\pfx = \Scope{\kvar}{c}$ was computed on $c$ up front. The next
lemma aligns the two at every \kvar-definition: truncating $\pfx$ from
the accumulated context recovers the definition's address inside
$\HeadOf{c}{\pfx}$, the directions \Navname\ walks by.

\begin{lemma}[Prefix correspondence]\label{lem:prefix}
Let $\pfx = \Scope{\kvar}{c}$. If the root call $\Cert{\varepsilon}{\hyp}{c}$
(\Cref{fig:certify}) reaches $\Cert{\ctx}{\pf}{\kvapp{\kvar}{\overline{t}}}$,
then $\kvapp{\kvar}{\overline{t}} = \HeadOf{c}{\pfx \cdot \pi}$ for some
path $\pi$, and
\begin{enumerate}
\item[(a)] $\ctx$ matches $\pfx \cdot \pi$;
\item[(b)] $\restr{\ctx}{\pfx}$ matches $\pi$.
\end{enumerate}
\end{lemma}

\begin{proof}
By induction on the recursion of \Certname\ (\Cref{fig:certify}), every
call $\Cert{\ctx}{\pf}{c^\ast}$ reached from $\Cert{\varepsilon}{\hyp}{c}$
satisfies the invariant: $c^\ast = \HeadOf{c}{\tau}$ for some path
$\tau$, and $\ctx$ matches $\tau$. At the root, $\tau = \ctx =
\varepsilon$; each recursive equation extends $\tau$ by one step and
$\ctx$ by the matching entry:
\begin{itemize}
  \item at $\cAll{x}{b}{c_0}$: the step $\pBind{x}{b}$, the entry
        $\cBind{x}{b}$;
  \item at $\cImp{\pred}{c_0}$: the step $\pAsm{\pred}$, the entry
        $\cBind{\hyp}{\pred}$;
  \item at $\cAnd{c_1}{c_2}$: the mark $\pLeft$ for $c_1$ and $\pRight$
        for $c_2$, as both step and entry.
\end{itemize}
At $c^\ast = \kvapp{\kvar}{\overline{t}}$, this \kvar-definition lies
inside $\HeadOf{c}{\pfx}$ (\Cref{app:zap-prelim}), so $\tau$ extends
$\pfx$: $\tau = \pfx \cdot \pi$ for some path $\pi$. The invariant at
this call is (a), and dropping the first $|\pfx|$ entries gives (b).
\end{proof}

When \Remname\ rewrites a definition head $\kvapp{\kvar}{\overline{t}}$
to $\top$, \Navname\ supplies the proof that justifies it: a proof of
the applied solution $\asgn(\overline{t})$. The next lemma states that
the term \Navname\ builds under the truncated context
$\restr{\Gamma}{\pfx}$ is indeed well-typed (valid proof) at $\asgn(\overline{t})$.

\begin{lemma}[\Navname]\label{lem:nav}
Let $\pi$ be a path with $\HeadOf{c}{\pfx \cdot \pi} =
\kvapp{\kvar}{\overline{t}}$ and $\Gamma = \CtxOf{c}{\pfx \cdot \pi}$. Then
\[
  \Gamma \;\vdash\; \Nav{\restr{\Gamma}{\pfx}}{\asgn(\overline{t})}
  \;:\; \asgn(\overline{t}).
\]
\end{lemma}

\begin{proof}
The $\beta$-reduction on $\asgn(\overline{t})$ performs substitution on
the leading arguments among $\overline{t}$ by the params $\overline{z}$.
By well-formedness, we are left with the binders $\overline{x}$, each of
$x_i$ gets replaced by the corresponding $\forall x_i$ in the context.
Finally, we get fully reduced $\asgn{}$ in the head.

By \Cref{lem:sol-shape}(b),
\begin{enumerate}
  \item along the branch that $\pi$ addresses, the $\beta$-reduced
        $\asgn(\overline{t})$ is the dual tree of $\HeadOf{c}{\pfx}$;
  \item each $\exists$-binder on that branch is named after a binder
        of $\Gamma$.
\end{enumerate}
Since the substituted terms also mention binders of $\Gamma$, the
reduction is capture-avoiding: it renames each $\exists$-binder on the
walked branch to a fresh name. \Navname\ emits no binders: only pairs, 
injections, witnesses, and hypothesis references. The context therefore 
never grows, and every judgment below is under the same $\Gamma$.

By induction on the suffix $\ctx'$ of $\restr{\Gamma}{\pfx}$, we prove
\[
  \Gamma \;\vdash\; \Nav{\ctx'}{\pred} \;:\; \pred,
\]
where $\pred$ is the subtree of the $\beta$-reduced
$\asgn(\overline{t})$ that $\ctx'$ addresses: $\Gamma =
\CtxOf{c}{\pfx \cdot \pi}$ matches $\pfx \cdot \pi$ entry by entry
(\Cref{sec:chc-props}), so $\restr{\Gamma}{\pfx}$ matches $\pi$.
The lemma is the instance
$\ctx' = \restr{\Gamma}{\pfx}$. Each \Navname\ equation
(\Cref{fig:nav}) is matched by one typing rule.

\begin{itemize}
\item \emph{Case} $\ctx' = \pLeft; \ctx''$ at $\pred_L \vee \pred_R$:
  the branch continues in $\pred_L$. The induction hypothesis
  certifies $\pred_L$, and $\vee$-introduction with $\mathsf{inl}$
  gives the disjunction; the untaken $\pred_R$ appears only in the
  type.
\item \emph{Case} $\ctx' = \pRight; \ctx''$: by symmetry.
\item \emph{Case} $\ctx' = \cBind{x}{b}; \ctx''$ at
  $\cExi{z}{b}{\pred}$: $z$ is the fresh rename of the walk binder
  $x$, and the witness substitution $\subasgn{\pred}{z}{x}$
  (\Cref{fig:nav}) undoes the renaming; on the walked branch the two
  substitutions compose to the identity. The induction hypothesis
  certifies the body, and $\exists$-introduction with the witness
  $x \in \Gamma$ gives the existential.
\item \emph{Case} $\ctx' = \cBind{\hyp}{g}; \ctx''$ at
  $\cAnd{g}{\pred}$: the substitutions leave $g$ unchanged (the slots
  are fresh, each $x_i$ replaced itself, the renamings above were
  undone), so the first conjunct is syntactically the guard named by
  $\hyp$: $\cBind{\hyp}{g} \in \Gamma$. $\wedge$-introduction pairs $\hyp$ with
  the recursive certificate.
\item \emph{Case} $\ctx' = \varepsilon$ at
  $\bigwedge_{i<\Rank{\kvar}} z_i = t_i$: this leaf belongs to the
  definition that supplied $\overline{t}$, so every conjunct reads
  $t_i = t_i$, and $\langle \rfl, \ldots, \rfl\rangle$ types it. When
  $\Rank{\kvar} = 0$, the empty conjunction is $\top$ and
  $\langle\rangle$ does.
\end{itemize}
\end{proof}

We write $\subasgn{\ctx}{\kvar}{\asgn}$ for the context $\ctx$ with
\asgn\ substituted in each hypothesis type; binders and marks are
unchanged.

\begin{lemma}[\Certname]\label{lem:cert}
For every prefix $\pi$, let $c^\ast = \HeadOf{c}{\pi}$ and
$\ctx = \CtxOf{c}{\pi}$. If
$\subasgn{\ctx}{\kvar}{\asgn} \vdash \pf : \Rem{\kvar}{\asgn}{c^\ast}$,
then
\[
  \subasgn{\ctx}{\kvar}{\asgn} \;\vdash\; \Cert{\ctx}{\pf}{c^\ast}
  \;:\; \subasgn{c^\ast}{\kvar}{\asgn}.
\]
\end{lemma}

\begin{proof}
By induction on $c^\ast$, pairing each equation of \Cref{fig:certify}
with the matching \Remname\ equation (\Cref{lem:red}) and one typing
rule.
\begin{itemize}
\item \emph{Case} $\cAll{x}{b}{c_0}$: both operators pass the
  quantifier through. Under $\ctx; \cBind{x}{b}$, the term $\pf\,x$
  proves $\Rem{\kvar}{\asgn}{c_0}$. The induction hypothesis and
  $\forall$-introduction give $\cLam{x}{\cdot}$.
\item \emph{Case} $\cImp{g}{c_0}$: both operators substitute the guard
  to $\subasgn{g}{\kvar}{\asgn}$. Given $\hyp : \subasgn{g}{\kvar}{\asgn}$,
  the term $\pf\,\hyp$ proves $\Rem{\kvar}{\asgn}{c_0}$. The induction
  hypothesis and $\to$-introduction give $\cLam{\hyp}{\cdot}$.
\item \emph{Case} $\cAnd{c_1}{c_2}$: the context extends by $\pLeft$,
  and $\pf.1$ proves the reduced $c_1$, so the induction hypothesis
  certifies $c_1$; symmetrically for $c_2$ with $\pRight$ and $\pf.2$.
  $\wedge$-introduction pairs the two certificates.
\item \emph{Case} $\kvapp{\kvar}{\overline{t}}$: here
  $\Rem{\kvar}{\asgn}{c^\ast} = \top$ and
  $\subasgn{c^\ast}{\kvar}{\asgn} = \asgn(\overline{t})$. The vacuous
  \pf\ is discarded, and $\Cert{\ctx}{\pf}{c^\ast} =
  \Nav{\restr{\ctx}{\pfx}}{\asgn(\overline{t})}$. The head lies inside
  $\HeadOf{c}{\pfx}$, so $\pi = \pfx \cdot \pi'$ for some path $\pi'$
  (\Cref{lem:prefix}). Every guard on this path is \kvar-free: above
  the LCA by the side condition of \Scopename, below it by acyclicity,
  as in \Cref{lem:sol-shape}(c). Hence
  $\subasgn{\ctx}{\kvar}{\asgn} = \ctx$, and \Cref{lem:nav} types the
  term at $\asgn(\overline{t})$.
\item \emph{Case} $\pred$, any other head: $\Rem{\kvar}{\asgn}{\pred} =
  \subasgn{\pred}{\kvar}{\asgn} = \pred$, and
  $\Cert{\ctx}{\pf}{\pred} = \pf$.
\end{itemize}
\end{proof}

Now, we state soundness of \Zaponame. The composition in
\Cref{app:zap-soundness} additionally needs a bound on the free
variables of the emitted proof term, which we record as a second
conclusion.

\begin{lemma}[\Zaponame]\label{lem:zap1}
If $\Zapo{\kvar}{c} = (c', \pf)$ then $\pf : c' \to \exists\kvar.\,c$
and $\mathrm{FV}(\pf) \subseteq \kvs{c} \setminus \{\kvar\}$.
\end{lemma}

\begin{proof}
\Zaponame\ (\Cref{fig:algo:zap}) emits
$\pf = \cLam{\hyp}{\langle \asgn,\ \Cert{\varepsilon}{\hyp}{c}\rangle}$, so
suppose $\hyp : c' = \Rem{\kvar}{\asgn}{c}$.

\emph{Typing.} Instantiate \Cref{lem:cert} at $\pi = \varepsilon$:
there $c^\ast = \HeadOf{c}{\varepsilon} = c$ and $\ctx =
\CtxOf{c}{\varepsilon} = \varepsilon$, so it gives
$\Cert{\varepsilon}{\hyp}{c} : \subasgn{c}{\kvar}{\asgn}$. The witness
\asgn\ is \kvar-free and of \kvar's predicate sort
(\Cref{lem:sol-shape}), so $\exists$-introduction yields
$\langle \asgn,\ \Cert{\varepsilon}{\hyp}{c}\rangle : \exists\kvar.\,c$,
and $\to$-introduction over $\hyp$ gives $\pf : c' \to \exists\kvar.\,c$.

\emph{Free variables.} Every binder in \pf\ (the outer $\hyp$ and the
$\lambda$s emitted by \Certname) is unannotated, and every variable
\Navname\ emits is bound by one of them. The only other constituent of
\pf\ is the witness \asgn, so
$\mathrm{FV}(\pf) \subseteq \mathrm{FV}(\asgn) \subseteq
\kvs{c} \setminus \{\kvar\}$ by \Cref{lem:sol-shape}(c).
\end{proof}

In the proof of \Cref{thm:zap:sound}, each step bridge delivers
$\exists\kvar_i.\,c_{i-1}$ while the accumulated bridge expects
$c_{i-1}$, so bridges must compose underneath an existential binder.
The next lemma justifies this. Here \kvar\ has a predicate sort
$b_1 \to \cdots \to b_n \to \Prop$, the case the $\exists$-rules of
\Cref{fig:zap-typing} admit.

\begin{lemma}[Existential lifting]\label{lem:emono}
Let \kvar\ be a variable of predicate sort, and let
\[
  \mathsf{lift}\,\pf \;\eqdef\;
  \cLam{\hyp_1}{\letin{\langle \kvar, \hyp_2\rangle = \hyp_1}{\langle \kvar,\
  \pf\,\hyp_2\rangle}}.
\]
If $\Gamma, \kvar \vdash \pf : A \to B$ then
$\Gamma \vdash \mathsf{lift}\,\pf :
(\exists\kvar.\,A) \to (\exists\kvar.\,B)$.
\end{lemma}

\begin{proof}
The \textsc{let} rule binds \kvar\ and $\hyp_2 : A$; then $\pf\,\hyp_2 : B$, and
$\exists$-introduction at \kvar\ gives $\langle \kvar, \pf\,\hyp_2\rangle :
\exists\kvar.\,B$, which does not mention the bound \kvar, discharging the
side condition.
\end{proof}

Lifting leaves the composed bridge's $\exists$-prefix in elimination
order; a final permutation restores the binder order of $\Cstr$, justified
with the following lemma:

\begin{lemma}[Prefix permutation]\label{lem:eperm}
Let $\overline{\kvar}'$ be a reordering of $\overline{\kvar}$, and let
\[
  \mathsf{perm} \;\eqdef\;
  \cLam{\hyp_1}{\letin{\langle \overline{\kvar},\, \hyp_2\rangle = \hyp_1}
  {\langle \overline{\kvar}';\ \hyp_2\rangle}}.
\]
Then $\vdash \mathsf{perm} :
(\exists\overline{\kvar}.\,c) \to (\exists\overline{\kvar}'.\,c)$.
\end{lemma}

\begin{proof}
The nested \textsc{let} binds the distinct witnesses $\overline{\kvar}$ and
exposes $\hyp_2 : c$; iterated $\exists$-introduction re-bundles the same
variables in the order $\overline{\kvar}'$.
\end{proof}

\begin{theorem}[Soundness, \Cref{sec:zap-correctness}]
\label{thm:zap-sound-app}
If $(\Cstr', \pf) = \Zap{\Cstr}$ then $\vdash \pf : \Cstr' \to \Cstr$, and
\pf\ is closed.
\end{theorem}

\begin{proof}
With $c_i$ as in \Cref{app:zap-elim}, $c' = c_n$ and $\Cstr' =
\ecstr{\kcyc{\kvar}}{c'}$. By \Cref{lem:zap1}, step $i$ emits
$\pf_i : c_i \to \exists\kvar_i.\,c_{i-1}$. Define $B_0 = \pfid$ and
$B_i = (\mathsf{lift}\,B_{i-1}) \circ \pf_i$, where $g \circ f$
abbreviates $\cLam{\hyp}{g\,(f\,\hyp)}$:
\[
\begin{array}{@{}r@{\;\;}l@{\qquad}l@{}}
\pf_i                  & : c_i \to \exists\kvar_i.\,c_{i-1} & \text{\Cref{lem:zap1}} \\
\mathsf{lift}\,B_{i-1} & : \exists\kvar_i.\,c_{i-1} \to \exists\kvar_i{\cdots}\kvar_1.\,c_0 & \text{\Cref{lem:emono}} \\
B_i                    & : c_i \to \exists\kvar_i{\cdots}\kvar_1.\,c_0 & \\
\end{array}
\]
All three judgments hold under the context of the variables still in
scope at step $i$: $\kvar_{i+1}, \ldots, \kvar_n$ and \kcycs{\kvar}. By
\Cref{lem:zap1}, that is where $\pf_i$ lives; \Cref{lem:emono} lifts
exactly the $\kvar_i$ that $\pf_i$ introduced. By \Cref{lem:afree}, each
$\asgn_j$ is cut-only, so no lift captures an acyclic variable.

At $i = n$, $B_n : c' \to \exists\kvar_n{\cdots}\kvar_1.\,c$ holds under
the cut variables alone. Applying \Cref{lem:emono} once per cut
variable, written $\mathsf{lift}^{\,\kcycs{\kvar}}$, binds the cut
prefix outside, since the $\asgn_j$ may mention cut variables. By
\Cref{lem:eperm}, $\mathsf{perm}$ restores $\Cstr$'s binder order:
\[
  \pf \;=\; \mathsf{perm} \circ \mathsf{lift}^{\,\kcycs{\kvar}} B_n
  \;:\; \ecstr{\kcyc{\kvar}}{c'} \to \ecstr{\kvar}{c}
  \;=\; \Cstr' \to \Cstr.
\]
This judgment holds under the empty context, so \pf\ is closed.

Up to flattening the intermediate $\kword{let}$-bindings, \pf\ is the
bridge \Cref{fig:algo:zap} assembles: \Zapsname's unpack/repack is
$\mathsf{lift}$, and \Zapname's top-level $\letin{}{}$ is the composite
of $\mathsf{lift}^{\,\kcycs{\kvar}}$ and $\mathsf{perm}$.
\end{proof}

\begin{theorem}[\Zapname-Equivalence]
\label{thm:zap:equiv}
\label{thm:zap-equiv-app}
If $(\Cstr', \cdot) = \Zap{\Cstr}$ then $\Cstr$ is satisfiable iff
$\Cstr'$ is satisfiable.
\end{theorem}

\begin{proof}
This is the \textsc{Fusion} result of \citet{Cosman17}: each \Remname\
step substitutes the strongest valid solution for an acyclic variable
and preserves satisfiability.
\end{proof}
 
\section{Fix: Soundness and Completeness}
\label{app:fix-proofs}

This appendix proves the correctness of the \Fixname\ algorithm of
\Cref{sec:fixpoint-algorithm}, whose two results are \emph{Soundness}
(\Cref{thm:fix:sound}) and \emph{Completeness} (\Cref{thm:fix:complete}).

\mypara{Overview}
Soundness is again a typing result. Completeness is the argument of
\citet{Rondon08}, with \sclean\ derivability in place of semantic
entailment.

\mypara{Setup}
Throughout, we fix a closed constraint
\[
  \Cstr \;\equiv\; \ecstr{\kcyc{\kvar}}{c}.
\]
In the pipeline, $\Cstr$ is \Zapname's residual and $\kcycs{\kvar}$ its
cut set, but nothing below requires this: predicate abstraction applies
to any closed constraint. We reuse the proof-term calculus of
\Cref{fig:zap-typing} and the conventions of \Cref{app:zap-prelim}.

\mypara{Organization}
\begin{itemize}
  \item \Cref{app:fix-prelim} recalls candidates, assignments, and the
        \Oraclename;
  \item \Cref{app:fix-fixpoint} characterizes the fixpoint computation;
  \item \Cref{app:fix-soundness} types \Certsname\ and the emitted
        bridge.
\end{itemize}

\subsection{Preliminaries}
\label{app:fix-prelim}

\mypara{Candidates and assignments}
Recall (\Cref{sec:predicate-abstraction}) that a candidate $\qual_\inst$ instantiates a
qualifier along an instance into a \kvar-free predicate of \kvar's type
(closed, qualifiers being top-level definitions), that an assignment \asgns\
maps each \kvar\ to a finite candidate set,
and that $\appasgn{\cdot}{\asgns}$ substitutes $\conjoinsol{\asgns(\kvar)}$
for each application of \kvar, extended homomorphically to atoms, contexts,
and constraints. The fold $\Rems{\kcycs{\kvar}}{c}$ of \Cref{fig:fix} agrees
with $\appasgn{c}{\asgns}$ except at heads, which it sends to $\top$;
because candidates are \kvar-free, the fold order is immaterial.

\mypara{Validity and order}
\asgns\ is \emph{valid} for $c$, written $\cSat{\asgns}{c}$, if at every
head-prefix $\pfx$ with $\HeadOf{c}{\pfx} = \kvapp{\kvar}{\overline{t}}$ the
\kvar-free goal $\appasgn{\CtxOf{c}{\pfx}}{\asgns} \vdash
\appasgn{\kvapp{\kvar}{\overline{t}}}{\asgns}$ is derivable. Assignments are
ordered by $\asgns \preceq \asgns'$ ($\asgns$ \emph{stronger}) iff
$\asgns'(\kvar) \subseteq \asgns(\kvar)$ for every \kvar: fewer candidates
is a weaker conjunction.

\mypara{The Oracle}
\Weakenname\ and \Certsname\ call an \Oraclename\ subject to the two laws of
\Cref{sec:pa-certification}: by \textit{(Oracle-Soundness)} a returned term proves
the queried goal, and by \textit{(Oracle-Completeness)} the \Oraclename\
succeeds on every derivable goal. Soundness uses \textit{(Oracle-Soundness)}
only; \textit{(Oracle-Completeness)} is used only for completeness, without
which the completeness claim weakens to the strongest
\emph{\Oraclename-certifiable} assignment.

\subsection{The fixpoint computation}
\label{app:fix-fixpoint}

\begin{lemma}[Monotonicity]\label{lem:fix-mono}\leavevmode
Let $\asgns \preceq \asgns'$.
\begin{enumerate}
\item[(a)] For every atom $\pred$, $\appasgn{\pred}{\asgns}$ entails
  $\appasgn{\pred}{\asgns'}$.
\item[(b)] For every context $\ctx$ and goal $\psi$, if
  $\appasgn{\ctx}{\asgns'} \vdash \psi$ then
  $\appasgn{\ctx}{\asgns} \vdash \psi$.
\end{enumerate}
\end{lemma}

\begin{proof}\leavevmode
\begin{enumerate}
\item[(a)] A \kvar-free atom is unchanged. At
  $\kvapp{\kvar}{\overline{t}}$ the two sides are conjunctions over
  $\asgns(\kvar) \supseteq \asgns'(\kvar)$, so the first contains every
  conjunct of the second; project and re-pair.
\item[(b)] Every hypothesis of $\appasgn{\ctx}{\asgns'}$ is an applied
  atom. By (a), each is derivable from the corresponding hypothesis of
  $\appasgn{\ctx}{\asgns}$; cut these derivations into the given one.
\end{enumerate}
\end{proof}

Part (c) below assumes \textit{(Oracle-Completeness)}; parts (a) and
(b) do not.

\begin{lemma}[Fixpoint]\label{lem:fix-fp}
Let $\asgns^\star = \Fixpoint{c}{\asgns_0}$ with
$\asgns_0 = [\kvar \mapsto \Init{\kvar}{\quals} \mid \kvar \in
\kcycs{\kvar}]$.
\begin{enumerate}
\item[(a)] \emph{Termination.} $\Fixpointname$ terminates.
\item[(b)] \emph{Validity.} $\cSat{\asgns^\star}{c}$.
\item[(c)] \emph{Strength.} $\asgns^\star \preceq \asgns'$ for every
  $\quals$-assignment $\asgns'$ with $\cSat{\asgns'}{c}$.
\end{enumerate}
\end{lemma}

\begin{proof}\leavevmode

\emph{(a) Termination.} $\asgns_0$ is finite, and each recursive call
strictly decreases $\sum_\kvar |\asgns(\kvar)| \in \Nat$: \Weakenname\
returns either a strictly smaller candidate set or $\bot$
(\Cref{fig:fixpoint}).

\emph{(b) Validity.} On exit, $\Weaken{c}{\asgns^\star}{\pfx} = \bot$
at every head-prefix $\pfx$: the surviving subset $I$ equals
$\asgns^\star(\kvar)$, so the \Oraclename\ succeeded on
$\qual_\inst(\overline{t})$ under
$\appasgn{\CtxOf{c}{\pfx}}{\asgns^\star}$ for every
$\qual_\inst \in \asgns^\star(\kvar)$. By \textit{(Oracle-Soundness)},
each success is a derivation of $\qual_\inst(\overline{t})$, and
$\wedge$-introduction assembles them into validity at $\pfx$.

\emph{(c) Strength.} By induction on the iterates $\asgns^k$, with
invariant $\asgns'(\kvar) \subseteq \asgns^k(\kvar)$ for every \kvar,
i.e.\ $\asgns^k \preceq \asgns'$. At the base, $\asgns_0$ contains all
instances. For the step, suppose \Weakenname\ at $\pfx$ drops some
$\qual_\inst \in \asgns'(\kvar)$: the \Oraclename\ failed on
$\qual_\inst(\overline{t})$ under
$\appasgn{\CtxOf{c}{\pfx}}{\asgns^k}$. By validity of $\asgns'$ and
$\wedge$-projection,
$\appasgn{\CtxOf{c}{\pfx}}{\asgns'} \vdash \qual_\inst(\overline{t})$.
By the invariant and \Cref{lem:fix-mono}(b), this derivation transports
to $\appasgn{\CtxOf{c}{\pfx}}{\asgns^k}$, and by
\textit{(Oracle-Completeness)} the \Oraclename\ succeeds, a
contradiction. Hence no candidate of $\asgns'$ is ever dropped, and
$\asgns^\star \preceq \asgns'$.
\end{proof}

\subsection{Soundness and completeness}
\label{app:fix-soundness}

\begin{lemma}[\Certsname]\label{lem:certs}
Let $\asgns^\star$ be as in \Cref{lem:fix-fp}, and suppose every
\Oraclename\ query made by \Certsname\ succeeds. For every prefix $\pi$,
let $c^\ast = \HeadOf{c}{\pi}$ and $\ctx = \CtxOf{c}{\pi}$. If
$\appasgn{\ctx}{\asgns^\star} \vdash \pf : \Rems{\kcycs{\kvar}}{c^\ast}$
then $\appasgn{\ctx}{\asgns^\star} \vdash \Certs{\ctx}{\pf}{c^\ast} :
\appasgn{c^\ast}{\asgns^\star}$.
\end{lemma}

\begin{proof}
By induction on $c^\ast$, as in \Cref{lem:cert}: the $\forall$, $\to$,
$\wedge$, and \kvar-free-head cases are verbatim, with both operators
substituting the guards identically.

At a head $c^\ast = \kvapp{\kvar}{\overline{t}}$,
$\Rems{\kcycs{\kvar}}{c^\ast} = \top$ and
$\appasgn{c^\ast}{\asgns^\star} =
(\conjoinsol{\asgns^\star(\kvar)})(\overline{t})$. \Certsname\
(\Cref{fig:certs}) discards the vacuous \pf\ and queries the
\Oraclename\ on this applied conjunction under the applied context. The
query succeeds by hypothesis, and by \textit{(Oracle-Soundness)} the
returned term proves
$(\conjoinsol{\asgns^\star(\kvar)})(\overline{t})$.
\end{proof}

\begin{theorem}[Soundness, \Cref{sec:pa-certification}]\label{thm:fix-sound-app}
If $\Fix{\Cstr}{\quals} = (\Cstr', \pf)$ then $\Cstr'$ is a VC and
$\vdash \pf : \Cstr' \to \Cstr$, and \pf\ is closed.
\end{theorem}

\begin{proof}
By \Cref{fig:fix}, $\Cstr' = \Rems{\kcycs{\kvar}}{c}$ and, writing
$\kcycs{\kvar} = [\kvar_1, \ldots, \kvar_r]$,
\[
  \pf \;=\; \cLam{\hyp}{\big\langle \conjoinsol{\asgns^\star(\kvar_1)},
  \ldots, \conjoinsol{\asgns^\star(\kvar_r)};\;
  \Certs{\varepsilon}{\hyp}{c}\big\rangle}.
\]
$\Cstr'$ is a VC: heads go to $\top$, and guards go to
$\appasgn{\cdot}{\asgns^\star}$ with \kvar-free candidates.

\emph{Typing.} Since \Fixname\ returned, every \Oraclename\ query made
by \Certsname\ succeeded, so \Cref{lem:certs} applies. Given
$\hyp : \Cstr'$, \Cref{lem:certs} at $\pi = \varepsilon$ gives
$\Certs{\varepsilon}{\hyp}{c} : \appasgn{c}{\asgns^\star}$, which is $c$
with each $\kvar_j$ replaced by the closed, \kvar-free predicate
$\conjoinsol{\asgns^\star(\kvar_j)}$ of $\kvar_j$'s sort. Iterated
$\exists$-introduction at predicate sort (\Cref{fig:zap-typing})
rebinds the prefix, yielding $\ecstr{\kcyc{\kvar}}{c} = \Cstr$, and
$\to$-introduction over $\hyp$ gives $\pf : \Cstr' \to \Cstr$. The
judgment holds under the empty context, so \pf\ is closed.
\end{proof}

Completeness assumes \textit{(Oracle-Completeness)}, as scoped in
\Cref{sec:pa-certification}.

\begin{theorem}[Completeness, \Cref{sec:pa-certification}]\label{thm:fix-complete-app}
If $\asgns^\star = \PredAbs{\Cstr}{\quals}$, then
$\asgns^\star \preceq \asgns'$ for every $\quals$-assignment $\asgns'$
with $\cSat{\asgns'}{c}$.
\end{theorem}

\begin{proof}
This is \Cref{lem:fix-fp}(c).
\end{proof}
 \section{Qualifier-Scraping Rules}
\label{app:scraping}

\sys implements qualifier scraping to generate candidates for
\Fixname. Here, we describe selected rules implemented in \sys;
their instances are validated by \PredAbsname.
Applying all rules at once can cause the candidate space to grow
rapidly. We therefore extend the qualifier bank \emph{incrementally},
moving to the next level when synthesis fails or the resulting
assignment fails its safety checks:
\begin{itemize}[leftmargin=*]
  \item \textbf{Bank:} Use the existing qualifier bank without scraping.
  \item \textbf{Basic:} Extract predicates from CHC heads.
  \item \textbf{Standard:} Add premise predicates, guarded arithmetic,
    and candidates from directly referenced definitions and source
    properties at $\kvar$-argument positions.
  \item \textbf{Deep:} Explore helper definitions and add further
    numeric and structural generalizations, including quantified bounds,
    accumulator completion, and affine relations.
\end{itemize}
Each level retains the candidates accumulated at earlier levels.

\mypara{Qualifier Bank.}
We provide 19 built-in qualifier templates.
The predicates below are abstracted over $x$ for unary templates
and $x,y$ for binary templates, with variables of the indicated type:
\[
\begin{array}{lll}
\text{Type} & \text{Unary: }\lambda x.\,p
            & \text{Binary: }\lambda x\,y.\,p \\[0.4em]
\Int
  & x=0,\; 0<x,\; 0\le x,\; x<0,\; x\le0
  & x=y,\; y<x,\; y\le x,\; x<y,\; x\le y,\; x\le y-1 \\[0.3em]
\Nat
  & x=0,\; 0<x
  & x=y,\; x<y,\; x\le y \\[0.3em]
\mathsf{Bool}
  & x=\mathsf{true},\; x=\mathsf{false}
  & x=y
\end{array}
\]

\mypara{Scope and Notation.}
Let $\Gamma$ contain caller locals, which remain fixed, and define
\[
 \mathsf{abs}_{\Gamma}(p) \doteq \cLam{\overline{x}}{p},
 \qquad \overline{x}=\mathrm{FV}(p)\setminus\domop(\Gamma).
\]
From a source predicate $p$, $\mathsf{abs}_{\Gamma}(p)$ forms a qualifier
template by abstracting the free variables outside $\Gamma$.
The abstraction must be well-scoped, \kvar-free, and contain no unresolved
metavariables. Its nonempty parameter list is bounded, and each parameter
has a closed data type; function types are allowed. We exclude $\top$,
$\bot$, and predicates with an outermost $\forall$ or $\to$.

\mypara{Constraint Atoms.}
To scrape qualifiers from atoms in the constraint, such as $x\le y$,
we traverse conjunctions, open universal binders, and substitute
let-bound values. Basic scraping selects heads; Standard and Deep
also select premise conjuncts. The following rule describes \emph{atom}
scraping for a selected predicate $p$:
\[
 \inferrule*[right=\rname{Atom}]
 {\qual=\mathsf{abs}_{\Gamma}(p)}
 {p\rightsquigarrow\qual}
\]
Here, $p\rightsquigarrow\qual$ means that scraping $p$ proposes qualifier
$\qual$. For example, $x\le y$ proposes $\cLam{x\,y}{x\le y}$
when $x$ and $y$ are clause-local variables.

\mypara{Guarded Arithmetic.}
We generate guarded qualifiers only when the whole constraint lies
in the supported arithmetic fragment.
Let $B$ contain one premise branch and its head, and let $A(B)$ collect
its selected arithmetic atoms, including bounds from equalities to literals.
We split disjunctive premises and open existential premises with fresh witnesses.
\[
 \inferrule*[right=\rname{Guard}]
 {a,b\in A(B)\quad \kvapp{\kvar}{\overline{t}}\in B
  \quad\mathrm{FV}(a),\mathrm{FV}(b)\ne\varnothing\\
  \mathrm{FV}(a)\ne\mathrm{FV}(b)\\
  \mathrm{FV}(a)\cup\mathrm{FV}(b)\subseteq\mathrm{FV}(\overline{t})
  \quad g\in\{a,\bar a\}\quad p\in\{b,\bar b\}}
 {B\rightsquigarrow\mathsf{abs}_{\Gamma}(\bar g\vee p)}
\]
Here $\bar a$ is the normalized complement, e.g.,
$\overline{s\le t}=t<s$ and $\overline{s<t}=t\le s$.
Thus $3\le y$ and $0\le x$, observed with $\kvar(x,y)$, can suggest
$\cLam{x\,y}{y<3\vee0\le x}$.

\mypara{Definition Calls and Tests.}
We inspect calls and branch conditions in the constraint and selected
defining equations. Let $T$ collect the visited terms, with binders opened.
Calls are inspected without recursive unfolding.
\begin{mathpar}
 \inferrule*[right=\rname{Call}]
 {f(\overline{t})\in T\quad f\text{ selected}\\
  f(\overline{t}):b\quad v:b\text{ fresh}}
 {T\rightsquigarrow\mathsf{abs}_{\Gamma}(v=f(\overline{t}))}
 \and
 \inferrule*[right=\rname{Test}]
 {(\mathbf{if}\ p\ \mathbf{then}\ t\ \mathbf{else}\ u)\in T\\
  a\in\mathsf{atoms}(p)\cup\mathsf{atoms}(\neg p)}
 {T\rightsquigarrow\mathsf{abs}_{\Gamma}(a)}
\end{mathpar}
Here, $b$ is closed and is neither a proposition, function type, nor sort.
The operation $\mathsf{atoms}$ splits conjunctions and substitutes
let-bound values without distributing negation.

\mypara{Argument Positions.}
Consider a source occurrence $\kvar(x_0,\ldots,x_n)$ whose arguments
are distinct variables of closed data types.
For a resolved, \kvar-free proposition $p$, we abstract the arguments
it uses and record their original positions:
\[
 \mathsf{at}_{\kvar}(p)\doteq\qual_{\inst},
 \qquad
 \qual=\cLam{z_0\cdots z_{r-1}}{p[z_j/x_{i_j}]_{j=0}^{r-1}},
 \quad \inst(j)=i_j.
\]
Here, $x_{i_0},\ldots,x_{i_{r-1}}$ are exactly the arguments occurring
free in $p$. We require $r>0$ and
$\mathrm{FV}(\qual)\subseteq\domop(\Gamma)$, so any remaining free
variables belong to the caller's context.
Unlike ordinary atoms, these candidates may retain quantifiers.

\mypara{Quantified Bounds.}
Let $O$ collect source observations at $\kvar(\overline{x})$.
A quantified property can suggest an invariant by replacing a free
numeric bound with a same-typed argument of $\kvar$:
\[
 \inferrule*[right=\rname{Bound}]
 {p\in O\quad p\text{ contains }\forall\\
  n\in\mathrm{FV}(p)\setminus\{\overline{x}\}
  \quad n,x_i:b\quad b\in\{\Int,\Nat\}}
 {O,\kvar(\overline{x})\rightsquigarrow
   \mathsf{at}_{\kvar}(p[x_i/n])}
\]
For example, replacing $n$ by the current index $i$ in
$\forall j.\ 0\le j<n\to m(j)\ne0$ proposes the corresponding
property over the processed prefix.

\mypara{Accumulator Completion.}
An observed final equality $a=E[xs]$ can suggest an invariant
relating accumulated and remaining work.
We combine $a$ with $E[r]$ for a compatible list argument $r$,
where $E[r]$ replaces all free occurrences of $xs$ by $r$ in $E[xs]$:
\[
 \inferrule*[right=\rname{Complete}]
 {(a=E[xs])\in O\quad a,r\in\overline{x}\\
  xs,r:\mathsf{List}\ b\quad xs\notin\overline{x}\\
  \mathrm{FV}(E[xs])\cap\{\overline{x}\}=\varnothing
  \quad a,E[xs]:\Int}
 {O,\kvar(\overline{x})\rightsquigarrow
   \mathsf{at}_{\kvar}(a+E[r]=E[xs])}
\]
For example, $a=\mathit{sum}(xs)$ suggests
$a+\mathit{sum}(r)=\mathit{sum}(xs)$.
Substitution avoids variable capture, and \PredAbsname must validate
the candidate; scraping assumes no additivity of $E$.
 \section{Additional Evaluation Results}
\label{app:evaluation-details}

\subsection{Elimination Proof Metrics}
\label{app:rq2-proof-metrics}

\Cref{tab:rq2-flux-agg-term} reports elimination certificate depth,
distinct constants, and kernel-checking time for \Cref{sec:rq2}.

\begin{table}[t]
\centering\footnotesize
\setlength{\tabcolsep}{2.5pt}
\begin{tabular}{@{}ll*{11}{r}@{}}
\toprule
 & & \multicolumn{2}{c}{CHCs} & \multicolumn{3}{c}{\zap} & \multicolumn{3}{c}{\fusiong} & \multicolumn{3}{c}{\fusiona} \\
\cmidrule(lr){3-4}\cmidrule(lr){5-7}\cmidrule(lr){8-10}\cmidrule(lr){11-13}
 & Suite & \shortstack[r]{Acyclic\\only} & Mixed & Depth & Consts & \shortstack[r]{Kernel\\($\mu$s)} & $\times$depth & $\times$consts & $\times$kernel & $\times$depth & $\times$consts & $\times$kernel \\
\midrule
& \fluxtests & 240 & 55 & 45.4 & 24.4 & 759 & 1.4$\times$ & 2.3$\times$ & 3.5$\times$ & 1.4$\times$ & 1.7$\times$ & 2.1$\times$ \\
& \fluxmedium & 36 & 44 & 75.6 & 27.8 & 2{,}180 & 1.5$\times$ & 2.9$\times$ & 5.5$\times$ & 1.4$\times$ & 1.7$\times$ & 2.3$\times$ \\
& \wavesuite & 26 & 9 & 99.4 & 34.6 & 2{,}828 & 1.6$\times$ & 2.2$\times$ & 3.9$\times$ & 1.5$\times$ & 1.6$\times$ & 2.6$\times$ \\
& \lfsuite & 50 & 7 & 47.6 & 20.0 & 1{,}041 & 1.6$\times$ & 2.8$\times$ & 5.7$\times$ & 1.8$\times$ & 2.0$\times$ & 4.2$\times$ \\
& Case studies & 27 & 20 & 166 & 47.1 & 14{,}496 & 1.3$\times$ & 1.9$\times$ & 2.8$\times$ & 1.4$\times$ & 1.9$\times$ & 2.5$\times$ \\
\midrule
\multirow{8}{*}{\rotatebox{90}{CHC-COMP}} & Eldarica & 0 & 97 & 402 & 42.3 & 253{,}140 & 1.3$\times$ & 2.8$\times$ & 3.5$\times$ & 1.5$\times$ & 1.7$\times$ & 5.6$\times$ \\
& HCAI & 0 & 104 & 183 & 42.9 & 11{,}718 & 1.0$\times$ & 1.3$\times$ & 1.0$\times$ & 1.0$\times$ & 1.2$\times$ & 1.0$\times$ \\
& HoPV & 32 & 46 & 161 & 30.6 & 10{,}791 & 1.3$\times$ & 2.6$\times$ & 4.0$\times$ & 1.4$\times$ & 1.6$\times$ & 3.8$\times$ \\
& Kind2 & 0 & 395 & 492 & 44.2 & 95{,}315 & -- & -- & -- & -- & -- & -- \\
& LIA-XS & 0 & 6 & 64.8 & 32.4 & 1{,}972 & 1.5$\times$ & 3.0$\times$ & 6.2$\times$ & 2.0$\times$ & 1.9$\times$ & 7.9$\times$ \\
& LLReve & 2 & 9 & 335 & 41.8 & 59{,}332 & 1.9$\times$ & 3.6$\times$ & 14.0$\times$ & -- & -- & -- \\
& RustHorn & 2 & 8 & 140 & 26.0 & 22{,}640 & -- & -- & -- & -- & -- & -- \\
& Synthesis & 0 & 60 & 270 & 36.0 & 609{,}511 & 1.0$\times$ & 1.7$\times$ & 1.5$\times$ & 1.8$\times$ & 2.0$\times$ & 7.6$\times$ \\
\midrule
& \textbf{All} & 415 & 860 & 168 & 34.4 & 14{,}179 & 1.4$\times$ & 2.4$\times$ & 3.7$\times$ & 1.4$\times$ & 1.7$\times$ & 2.5$\times$ \\
\bottomrule
\end{tabular}
\caption{Elimination certificate depth, distinct constants, and kernel-checking time,
including generated helper proofs. Values are geometric means; ratios compare
each variant with \zap{} on jointly reduced CHCs.}
\label{tab:rq2-flux-agg-term}
\end{table}

\subsection{RQ3 Results}
\label{app:rq3-details}

\Cref{tab:rq3a-time} gives the tactic-time distribution for
\hyperref[sec:rq3a]{the refinement-typechecker experiment}.

\begin{table}[t]
\centering\footnotesize
\setlength{\tabcolsep}{4pt}
\begin{tabular}{@{}ll r rrrrrr@{}}
\toprule
 & & & \multicolumn{6}{c}{Tactic time (s)} \\
\cmidrule(lr){4-9}
Configuration & Suite & Proved & Mean & P25 & Median & P75 & P90 & Max \\
\midrule
\multirow{5}{*}{\zap{} + \Fixname{}} & \fluxtests & 333 & 0.42 & 0.03 & 0.05 & 0.30 & 1.04 & 14.7 \\
 & \fluxmedium & 94 & 1.59 & 0.07 & 0.32 & 0.77 & 2.02 & 49.3 \\
 & \wavesuite & 35 & 1.95 & 0.06 & 0.31 & 0.75 & 2.83 & 24.0 \\
 & \lfsuite & 67 & 0.75 & 0.03 & 0.04 & 0.14 & 0.70 & 25.2 \\
 & \textbf{All} & 529 & 0.77 & 0.03 & 0.08 & 0.42 & 1.12 & 49.3 \\
\midrule
\multirow{5}{*}{\sys{} Spacer} & \fluxtests & 310 & 0.24 & 0.09 & 0.12 & 0.22 & 0.47 & 7.2 \\
 & \fluxmedium & 91 & 0.49 & 0.16 & 0.23 & 0.37 & 0.56 & 8.4 \\
 & \wavesuite & 35 & 0.38 & 0.17 & 0.27 & 0.41 & 0.87 & 1.4 \\
 & \lfsuite & 65 & 0.69 & 0.10 & 0.12 & 0.19 & 0.48 & 30.2 \\
 & \textbf{All} & 501 & 0.36 & 0.10 & 0.15 & 0.27 & 0.51 & 30.2 \\
\midrule
\multirow{5}{*}{\sys{} Eldarica} & \fluxtests & 320 & 1.32 & 0.66 & 0.78 & 1.10 & 1.67 & 87.7 \\
 & \fluxmedium & 91 & 1.73 & 0.81 & 1.05 & 1.34 & 1.84 & 39.5 \\
 & \wavesuite & 36 & 2.20 & 1.00 & 1.44 & 2.53 & 4.48 & 8.8 \\
 & \lfsuite & 65 & 1.13 & 0.68 & 0.73 & 0.82 & 1.63 & 8.7 \\
 & \textbf{All} & 512 & 1.43 & 0.69 & 0.81 & 1.23 & 1.90 & 87.7 \\
\bottomrule
\end{tabular}
\caption{Tactic time in seconds over the CHCs each configuration of \Cref{tab:rq3a-coverage}
fully proves. Percentiles are linearly interpolated.}
\label{tab:rq3a-time}
\end{table}

\Cref{tab:rq3b-chccomp} gives the per-suite results for
\hyperref[sec:rq3b]{the CHC-COMP experiment}.

\begin{table}[t]
\centering\footnotesize
\setlength{\tabcolsep}{4pt}
\begin{tabular}{@{}lr*{6}{c}@{}}
\toprule
 & & & & \multicolumn{2}{c}{\zap{}+\Fixname{}} & & \\
\cmidrule(lr){5-6}
Suite & CHCs & \multicolumn{1}{c}{Native Spacer} & \multicolumn{1}{c}{Native Eldarica} & \multicolumn{1}{c}{Full} & \multicolumn{1}{c}{Partial} & \multicolumn{1}{c}{\sys{} Spacer} & \multicolumn{1}{c}{\sys{} Eldarica} \\
\midrule
VMT       & 474 & 412 & 183 &  0 &   0 & 301 & 151 \\
Kind2     & 395 & 267 & 196 &  0 & 395 & 199 & 152 \\
Eldarica  & 138 & 108 & 121 & 24 &  83 & 101 & 108 \\
Synthesis & 109 &  95 &  67 & 22 &  54 &  88 &  63 \\
HoPV      & 108 & 100 & 107 & 76 &  23 &  99 & 107 \\
HCAI      & 104 &  92 &  90 & 23 &  81 &  91 &  85 \\
LLReve    & 104 &  90 &  96 & 56 &   6 &  85 &  95 \\
LIA-XS    &  55 &  16 &  36 & 17 &   3 &  15 &  31 \\
AEval     &  55 &   7 &  27 &  2 &   0 &   7 &  20 \\
RustHorn  &  10 &   8 &  10 &  1 &   9 &   8 &  10 \\
Ultimate  &   7 &   3 &   4 &  0 &   0 &   3 &   0 \\
\midrule
\textbf{All} & 1{,}559 & 1{,}198 & 937 & 221 & 654 & 997 & 822 \\
\bottomrule
\end{tabular}
\caption{Coverage on 1{,}559 distinct CHC-COMP constraints.
\emph{Native Spacer} and \emph{Native Eldarica}: CHCs the solver solves on its own.
\emph{\zap{}+\Fixname{}}: CHCs \zap{} and \Fixname{} prove outright (Full), and CHCs with
a partial result (Partial): \zap{} eliminates at least one predicate, or
\Fixname{} proves every constraint involving a predicate \zap{} cannot eliminate.
\emph{\sys{} Spacer} and \emph{\sys{} Eldarica}: CHCs \sys{} fully proves using that
solver as its backend.}
\label{tab:rq3b-chccomp}
\end{table}

\mypara{Failure Breakdown.}
Of the 1,198 CHCs that Spacer solves natively, \sys{} Spacer proves 983.
Of the remaining 215, 31 yield no model after the round trip through \sclean,
10 fail during model certification, and two exhaust memory.
The other 172 have all predicates instantiated but leave proof obligations.
Of the 937 CHCs that Eldarica solves natively, \sys{} Eldarica proves 799;
of the remaining 138, 16 yield no model, five fail during model certification,
one exhausts memory, and 116 leave proof obligations.
Almost all of these fail at head elimination (170 of 172 for Spacer and
115 of 116 for Eldarica): their imported invariants are large
LIA formulas that \lean{omega} and \texttt{lean-smt} cannot prove within
the budget.
\Cref{tab:rq3b-failures} gives the per-suite breakdown.

\begin{table}[t]
\centering\footnotesize
\setlength{\tabcolsep}{5pt}
\begin{tabular}{@{}l*{8}{r}@{}}
\toprule
 & \multicolumn{4}{c}{\sys{} Spacer} & \multicolumn{4}{c}{\sys{} Eldarica} \\
\cmidrule(lr){2-5}\cmidrule(lr){6-9}
Suite & Solver & Cert. & Disch. & Process & Solver & Cert. & Disch. & Process \\
\midrule
VMT       &  3 &  1 & 108 & 0 &  5 & 0 &  40 & 0 \\
Kind2     & 18 &  4 &  47 & 2 &  6 & 4 &  40 & 1 \\
Eldarica  &  2 &  0 &   6 & 0 &  2 & 0 &  11 & 0 \\
Synthesis &  4 &  2 &   6 & 0 &  2 & 0 &   3 & 0 \\
HoPV      &  1 &  0 &   0 & 0 &  0 & 0 &   0 & 0 \\
HCAI      &  0 &  1 &   1 & 0 &  0 & 0 &   5 & 0 \\
LLReve    &  3 &  0 &   3 & 0 &  1 & 0 &   2 & 0 \\
LIA-XS    &  0 &  1 &   0 & 0 &  0 & 0 &   5 & 0 \\
AEval     &  0 &  1 &   1 & 0 &  0 & 1 &   6 & 0 \\
RustHorn  &  0 &  0 &   0 & 0 &  0 & 0 &   0 & 0 \\
Ultimate  &  0 &  0 &   0 & 0 &  0 & 0 &   4 & 0 \\
\midrule
\textbf{All} & 31 & 10 & 172 & 2 & 16 & 5 & 116 & 1 \\
\bottomrule
\end{tabular}
\caption{CHCs that Spacer (resp.\ Eldarica) solves natively but \sys{} Spacer
(resp.\ \sys{} Eldarica) does not fully prove, by the furthest stage reached.
Solver: no model; Cert.: model-certification failure, including timeouts;
Disch.: proof obligations left undischarged, including predicate heads;
Process: out of memory.}
\label{tab:rq3b-failures}
\end{table}

\clearpage
\subsection{RQ4 Scraping Results}
\label{app:scraping-ablation}

The 391 CHCs comprise 274 cyclic-only and 117 mixed constraints.
We ran each CHC once per configuration. Tactic times exclude startup, imports,
and final kernel checking. 

\begin{table}[htbp]
\centering\small
\setlength{\tabcolsep}{4pt}
\begin{tabular}{@{}lr*{5}{r}@{}}
\toprule
Corpus & CHCs & Bank & Basic & Standard & Deep & Incremental \\
\midrule
Refinement-typechecker & 181 & 154 & 166 & 158 & 166 & 181 \\
CHC-COMP & 192 & 130 & 130 & 157 & 178 & 192 \\
RQ1 case studies & 18 & 1 & 9 & 18 & 18 & 18 \\
\midrule
\textbf{Total} & \textbf{391} & \textbf{285} & \textbf{305} & \textbf{333} & \textbf{362} & \textbf{391} \\
\bottomrule
\end{tabular}
\caption{Fully proved CHCs by scraping configuration.}
\label{tab:scraping-coverage}
\end{table}

\begin{table}[htbp]
\centering\small
\setlength{\tabcolsep}{4pt}
\begin{tabular}{@{}lr*{6}{r}@{}}
\toprule
Configuration & $n$ & Mean & P25 & Median & P75 & P90 & Max \\
\midrule
Bank & 285 & 1.53 & 0.18 & 0.37 & 1.07 & 3.64 & 29.04 \\
Basic & 305 & 1.51 & 0.19 & 0.38 & 1.19 & 3.60 & 29.12 \\
Standard & 333 & 6.12 & 0.38 & 1.16 & 5.38 & 17.59 & 99.80 \\
Deep & 362 & 7.05 & 0.44 & 1.56 & 6.99 & 21.46 & 101.82 \\
Incremental & 391 & 3.19 & 0.21 & 0.50 & 1.80 & 8.82 & 48.93 \\
\bottomrule
\end{tabular}
\caption{Tactic time in seconds over fully proved CHCs ($n$).
Quantiles use linear interpolation.}
\label{tab:scraping-times}
\end{table}

\begin{table}[htbp]
\centering\small
\setlength{\tabcolsep}{12pt}
\begin{tabular}{@{}lrr@{}}
\toprule
Level & CHCs & Share of 391 \\
\midrule
Bank & 285 & 72.9\% \\
Basic & 22 & 5.6\% \\
Standard & 52 & 13.3\% \\
Deep & 32 & 8.2\% \\
\midrule
\textbf{Total} & \textbf{391} & \textbf{100\%} \\
\bottomrule
\end{tabular}
\caption{First level where Incremental accepted an invariant.}
\label{tab:scraping-escalation}
\end{table}

One Standard and three Deep runs hit the 300\,s process watchdog.
Standard and Deep each eliminate all predicates in three further CHCs
but leave residual proof obligations.
  

\begin{thebibliography}{59}



\ifx \showCODEN    \undefined \def \showCODEN     #1{\unskip}     \fi
\ifx \showISBNx    \undefined \def \showISBNx     #1{\unskip}     \fi
\ifx \showISBNxiii \undefined \def \showISBNxiii  #1{\unskip}     \fi
\ifx \showISSN     \undefined \def \showISSN      #1{\unskip}     \fi
\ifx \showLCCN     \undefined \def \showLCCN      #1{\unskip}     \fi
\ifx \shownote     \undefined \def \shownote      #1{#1}          \fi
\ifx \showarticletitle \undefined \def \showarticletitle #1{#1}   \fi
\ifx \showURL      \undefined \def \showURL       {\relax}        \fi
\providecommand\bibfield[2]{#2}
\providecommand\bibinfo[2]{#2}
\providecommand\natexlab[1]{#1}
\providecommand\showeprint[2][]{arXiv:#2}

\bibitem[Astrauskas et~al\mbox{.}(2022)]{prusti}
\bibfield{author}{\bibinfo{person}{Vytautas Astrauskas}, \bibinfo{person}{Aurea
  B\'{\i}l\'{a}}, \bibinfo{person}{Jon\'{a}\v{s} Fiala},
  \bibinfo{person}{Zachary Grannan}, \bibinfo{person}{Christoph Matheja},
  \bibinfo{person}{Peter M\"{u}ller}, \bibinfo{person}{Federico Poli}, {and}
  \bibinfo{person}{Alexander~J. Summers}.} \bibinfo{year}{2022}\natexlab{}.
\newblock \showarticletitle{The Prusti Project: Formal Verification for Rust}.
  In \bibinfo{booktitle}{\emph{NASA Formal Methods: 14th International
  Symposium, NFM 2022, Pasadena, CA, USA, May 24–27, 2022, Proceedings}}
  (Pasadena, CA, USA). \bibinfo{publisher}{Springer-Verlag},
  \bibinfo{address}{Berlin, Heidelberg}, \bibinfo{pages}{88–108}.
\newblock
\showISBNx{978-3-031-06772-3}
\href{https://doi.org/10.1007/978-3-031-06773-0_5}{doi:\nolinkurl{10.1007/978-3-031-06773-0_5}}


\bibitem[Barbosa et~al\mbox{.}(2022)]{cvc5}
\bibfield{author}{\bibinfo{person}{Haniel Barbosa}, \bibinfo{person}{Clark
  Barrett}, \bibinfo{person}{Martin Brain}, \bibinfo{person}{Gereon Kremer},
  \bibinfo{person}{Hanna Lachnitt}, \bibinfo{person}{Makai Mann},
  \bibinfo{person}{Abdalrhman Mohamed}, \bibinfo{person}{Mudathir Mohamed},
  \bibinfo{person}{Aina Niemetz}, \bibinfo{person}{Andres N{\"o}tzli},
  \bibinfo{person}{Alex Ozdemir}, \bibinfo{person}{Mathias Preiner},
  \bibinfo{person}{Andrew Reynolds}, \bibinfo{person}{Ying Sheng},
  \bibinfo{person}{Cesare Tinelli}, {and} \bibinfo{person}{Yoni Zohar}.}
  \bibinfo{year}{2022}\natexlab{}.
\newblock \showarticletitle{cvc5: A Versatile and Industrial-Strength SMT
  Solver}. In \bibinfo{booktitle}{\emph{Tools and Algorithms for the
  Construction and Analysis of Systems}},
  \bibfield{editor}{\bibinfo{person}{Dana Fisman} {and}
  \bibinfo{person}{Grigore Rosu}} (Eds.). \bibinfo{publisher}{Springer
  International Publishing}, \bibinfo{address}{Cham},
  \bibinfo{pages}{415--442}.
\newblock
\showISBNx{978-3-030-99524-9}


\bibitem[Barbosa et~al\mbox{.}(2023)]{smt-certificates}
\bibfield{author}{\bibinfo{person}{Haniel Barbosa}, \bibinfo{person}{Clark
  Barrett}, \bibinfo{person}{Byron Cook}, \bibinfo{person}{Bruno Dutertre},
  \bibinfo{person}{Gereon Kremer}, \bibinfo{person}{Hanna Lachnitt},
  \bibinfo{person}{Aina Niemetz}, \bibinfo{person}{Andres N\"{o}tzli},
  \bibinfo{person}{Alex Ozdemir}, \bibinfo{person}{Mathias Preiner},
  \bibinfo{person}{Andrew Reynolds}, \bibinfo{person}{Cesare Tinelli}, {and}
  \bibinfo{person}{Yoni Zohar}.} \bibinfo{year}{2023}\natexlab{}.
\newblock \showarticletitle{Generating and Exploiting Automated Reasoning Proof
  Certificates}.
\newblock \bibinfo{journal}{\emph{Commun. ACM}} \bibinfo{volume}{66},
  \bibinfo{number}{10} (\bibinfo{date}{Sept.} \bibinfo{year}{2023}),
  \bibinfo{pages}{86–95}.
\newblock
\showISSN{0001-0782}
\href{https://doi.org/10.1145/3587692}{doi:\nolinkurl{10.1145/3587692}}


\bibitem[B{\'a}rtek et~al\mbox{.}(2025)]{bártek2025vampirediary}
\bibfield{author}{\bibinfo{person}{Filip B{\'a}rtek}, \bibinfo{person}{Ahmed
  Bhayat}, \bibinfo{person}{Robin Coutelier}, \bibinfo{person}{Márton Hajdu},
  \bibinfo{person}{Matthias Hetzenberger}, \bibinfo{person}{Petra Hozzová},
  \bibinfo{person}{Laura Kovács}, \bibinfo{person}{Jakob Rath},
  \bibinfo{person}{Michael Rawson}, \bibinfo{person}{Giles Reger},
  \bibinfo{person}{Martin Suda}, \bibinfo{person}{Johannes Schoisswohl}, {and}
  \bibinfo{person}{Andrei Voronkov}.} \bibinfo{year}{2025}\natexlab{}.
\newblock \bibinfo{title}{The Vampire Diary}.
\newblock
\showeprint[arxiv]{2506.03030}~[cs.LO]
\urldef\tempurl \url{https://arxiv.org/abs/2506.03030}
\showURL{\tempurl}


\bibitem[Bj{\o}rner et~al\mbox{.}(2015)]{bjorner2015horn}
\bibfield{author}{\bibinfo{person}{Nikolaj Bj{\o}rner}, \bibinfo{person}{Arie
  Gurfinkel}, \bibinfo{person}{Ken McMillan}, {and} \bibinfo{person}{Andrey
  Rybalchenko}.} \bibinfo{year}{2015}\natexlab{}.
\newblock \showarticletitle{Horn Clause Solvers for Program Verification}.
\newblock In \bibinfo{booktitle}{\emph{Fields of Logic and Computation II}}.
  \bibinfo{publisher}{Springer}, \bibinfo{pages}{24--51}.
\newblock
\href{https://doi.org/10.1007/978-3-319-23534-9_2}{doi:\nolinkurl{10.1007/978-3-319-23534-9_2}}


\bibitem[Blicha et~al\mbox{.}(2025)]{Golem2025}
\bibfield{author}{\bibinfo{person}{Martin Blicha}, \bibinfo{person}{Konstantin
  Britikov}, {and} \bibinfo{person}{Natasha Sharygina}.}
  \bibinfo{year}{2025}\natexlab{}.
\newblock \showarticletitle{{Golem}: a flexible and efficient solver for
  constrained {Horn} clauses}.
\newblock \bibinfo{journal}{\emph{Formal Methods in System Design}}
  \bibinfo{volume}{67} (\bibinfo{year}{2025}), \bibinfo{pages}{143--160}.
\newblock
\href{https://doi.org/10.1007/s10703-025-00470-9}{doi:\nolinkurl{10.1007/s10703-025-00470-9}}


\bibitem[Bodingbauer et~al\mbox{.}(2026)]{vampire-lean}
\bibfield{author}{\bibinfo{person}{Jonas Bodingbauer}, \bibinfo{person}{Márton
  Hajdu}, \bibinfo{person}{Laura Kovács}, \bibinfo{person}{Axel Polaczek},
  {and} \bibinfo{person}{Michael Rawson}.} \bibinfo{year}{2026}\natexlab{}.
\newblock \bibinfo{title}{Lean on Vampire Proofs (Short Paper)}.
\newblock
\showeprint[arxiv]{2603.26342}~[cs.LO]
\urldef\tempurl \url{https://arxiv.org/abs/2603.26342}
\showURL{\tempurl}


\bibitem[Bradley(2011)]{IC3}
\bibfield{author}{\bibinfo{person}{Aaron~R Bradley}.}
  \bibinfo{year}{2011}\natexlab{}.
\newblock \showarticletitle{SAT-based model checking without unrolling}. In
  \bibinfo{booktitle}{\emph{Verification, Model Checking, and Abstract
  Interpretation (VMCAI)}}. Springer, \bibinfo{pages}{70--87}.
\newblock
\href{https://doi.org/10.1007/978-3-642-18275-4_7}{doi:\nolinkurl{10.1007/978-3-642-18275-4_7}}


\bibitem[Britikov and Bajczi(2026)]{chccomp26}
\bibfield{author}{\bibinfo{person}{Konstantin Britikov} {and}
  \bibinfo{person}{Levente Bajczi}.} \bibinfo{year}{2026}\natexlab{}.
\newblock \bibinfo{title}{{CHC-COMP} 2026: Constrained Horn Clause Solver
  Competition}.
\newblock
  \bibinfo{howpublished}{\url{https://chc-comp.github.io/chc-comp-2026/}}.
\newblock
\newblock
\shownote{Results presented at HCVS 2026. Accessed: 2026-11-11}.


\bibitem[Carneiro(2025)]{lean4lean}
\bibfield{author}{\bibinfo{person}{Mario Carneiro}.}
  \bibinfo{year}{2025}\natexlab{}.
\newblock \bibinfo{title}{Lean4Lean: Verifying a Typechecker for Lean, in
  Lean}.
\newblock
\showeprint[arxiv]{2403.14064}~[cs.PL]
\urldef\tempurl \url{https://arxiv.org/abs/2403.14064}
\showURL{\tempurl}


\bibitem[Champion et~al\mbox{.}(2018)]{Holce}
\bibfield{author}{\bibinfo{person}{Adrien Champion}, \bibinfo{person}{Naoki
  Kobayashi}, {and} \bibinfo{person}{Ryosuke Sato}.}
  \bibinfo{year}{2018}\natexlab{}.
\newblock \showarticletitle{HoIce: An ICE-Based Non-linear Horn Clause Solver}.
  In \bibinfo{booktitle}{\emph{Programming Languages and Systems - 16th Asian
  Symposium, APLAS 2018, Proceedings}} \emph{(\bibinfo{series}{Lecture Notes in
  Computer Science (including subseries Lecture Notes in Artificial
  Intelligence and Lecture Notes in Bioinformatics)})},
  \bibfield{editor}{\bibinfo{person}{Sukyoung Ryu}} (Ed.).
  \bibinfo{publisher}{Springer Verlag}, \bibinfo{address}{Germany},
  \bibinfo{pages}{146--156}.
\newblock
\showISBNx{9783030027674}
\href{https://doi.org/10.1007/978-3-030-02768-1\_8}{doi:\nolinkurl{10.1007/978-3-030-02768-1\_8}}
\newblock
\shownote{Publisher Copyright: {\textcopyright} 2018, Springer Nature
  Switzerland AG.; 16th Asian Symposium on Programming Languages and Systems,
  APLAS 2018 ; Conference date: 02-12-2018 Through 06-12-2018}.


\bibitem[Chen et~al\mbox{.}(2010)]{fine}
\bibfield{author}{\bibinfo{person}{Juan Chen}, \bibinfo{person}{Ravi Chugh},
  {and} \bibinfo{person}{Nikhil Swamy}.} \bibinfo{year}{2010}\natexlab{}.
\newblock \showarticletitle{Type-preserving compilation of end-to-end
  verification of security enforcement}. In
  \bibinfo{booktitle}{\emph{Proceedings of the 2010 {ACM} {SIGPLAN} Conference
  on Programming Language Design and Implementation, {PLDI} 2010, Toronto,
  Ontario, Canada, June 5-10, 2010}},
  \bibfield{editor}{\bibinfo{person}{Benjamin~G. Zorn} {and}
  \bibinfo{person}{Alex Aiken}} (Eds.). \bibinfo{publisher}{{ACM}},
  \bibinfo{pages}{412--423}.
\newblock
\href{https://doi.org/10.1145/1806596.1806643}{doi:\nolinkurl{10.1145/1806596.1806643}}


\bibitem[Cook et~al\mbox{.}(2026)]{10.1007/978-3-032-28079-4_19}
\bibfield{author}{\bibinfo{person}{Byron Cook}, \bibinfo{person}{Remi Delmas},
  \bibinfo{person}{Zyad Hassan}, \bibinfo{person}{Bart Jacobs},
  \bibinfo{person}{Ranjit Jhala}, \bibinfo{person}{Rahul Kumar},
  \bibinfo{person}{Felipe~R. Monteiro}, \bibinfo{person}{Thanh Nguyen},
  \bibinfo{person}{Rebecca Rumbul}, \bibinfo{person}{Michael Tautschnig},
  \bibinfo{person}{Celina Val}, {and} \bibinfo{person}{Carolyn Zech}.}
  \bibinfo{year}{2026}\natexlab{}.
\newblock \showarticletitle{Verifying the Rust Standard Library}. In
  \bibinfo{booktitle}{\emph{NASA Formal Methods: 18th International Symposium,
  NFM 2026, Los Angeles, CA, USA, May 5–7, 2026, Proceedings}} (Los Angeles,
  CA, USA). \bibinfo{publisher}{Springer-Verlag}, \bibinfo{address}{Berlin,
  Heidelberg}, \bibinfo{pages}{415–435}.
\newblock
\showISBNx{978-3-032-28078-7}
\href{https://doi.org/10.1007/978-3-032-28079-4_19}{doi:\nolinkurl{10.1007/978-3-032-28079-4_19}}


\bibitem[Cosman and Jhala(2017)]{Cosman17}
\bibfield{author}{\bibinfo{person}{Benjamin Cosman} {and}
  \bibinfo{person}{Ranjit Jhala}.} \bibinfo{year}{2017}\natexlab{}.
\newblock \showarticletitle{Local refinement typing}.
\newblock \bibinfo{journal}{\emph{Proc. ACM Program. Lang.}}
  \bibinfo{volume}{1}, \bibinfo{number}{ICFP}, Article \bibinfo{articleno}{26}
  (\bibinfo{date}{Aug.} \bibinfo{year}{2017}), \bibinfo{numpages}{27}~pages.
\newblock
\href{https://doi.org/10.1145/3110270}{doi:\nolinkurl{10.1145/3110270}}


\bibitem[Cuoq et~al\mbox{.}(2012)]{framac}
\bibfield{author}{\bibinfo{person}{Pascal Cuoq}, \bibinfo{person}{Florent
  Kirchner}, \bibinfo{person}{Nikolai Kosmatov}, \bibinfo{person}{Virgile
  Prevosto}, \bibinfo{person}{Julien Signoles}, {and} \bibinfo{person}{Boris
  Yakobowski}.} \bibinfo{year}{2012}\natexlab{}.
\newblock \showarticletitle{Frama-C: a software analysis perspective}. In
  \bibinfo{booktitle}{\emph{Proceedings of the 10th International Conference on
  Software Engineering and Formal Methods}} (Thessaloniki, Greece)
  \emph{(\bibinfo{series}{SEFM'12})}. \bibinfo{publisher}{Springer-Verlag},
  \bibinfo{address}{Berlin, Heidelberg}, \bibinfo{pages}{233–247}.
\newblock
\showISBNx{9783642338250}
\href{https://doi.org/10.1007/978-3-642-33826-7_16}{doi:\nolinkurl{10.1007/978-3-642-33826-7_16}}


\bibitem[de~Moura and Bj{\o}rner(2008)]{z3}
\bibfield{author}{\bibinfo{person}{Leonardo de Moura} {and}
  \bibinfo{person}{Nikolaj Bj{\o}rner}.} \bibinfo{year}{2008}\natexlab{}.
\newblock \showarticletitle{Z3: An Efficient SMT Solver}. In
  \bibinfo{booktitle}{\emph{Tools and Algorithms for the Construction and
  Analysis of Systems}}, \bibfield{editor}{\bibinfo{person}{C.~R. Ramakrishnan}
  {and} \bibinfo{person}{Jakob Rehof}} (Eds.). \bibinfo{publisher}{Springer
  Berlin Heidelberg}, \bibinfo{address}{Berlin, Heidelberg},
  \bibinfo{pages}{337--340}.
\newblock
\showISBNx{978-3-540-78800-3}


\bibitem[Denis et~al\mbox{.}(2022)]{Creusot}
\bibfield{author}{\bibinfo{person}{Xavier Denis},
  \bibinfo{person}{Jacques-Henri Jourdan}, {and} \bibinfo{person}{Claude
  March\'{e}}.} \bibinfo{year}{2022}\natexlab{}.
\newblock \showarticletitle{Creusot: A Foundry for the Deductive Verification
  of Rust Programs}. In \bibinfo{booktitle}{\emph{Formal Methods and Software
  Engineering: 23rd International Conference on Formal Engineering Methods,
  ICFEM 2022, Madrid, Spain, October 24–27, 2022, Proceedings}} (Madrid,
  Spain). \bibinfo{publisher}{Springer-Verlag}, \bibinfo{address}{Berlin,
  Heidelberg}, \bibinfo{pages}{90–105}.
\newblock
\showISBNx{978-3-031-17243-4}
\href{https://doi.org/10.1007/978-3-031-17244-1_6}{doi:\nolinkurl{10.1007/978-3-031-17244-1_6}}


\bibitem[Dietsch et~al\mbox{.}(2019)]{Dietsch_2019}
\bibfield{author}{\bibinfo{person}{Daniel Dietsch}, \bibinfo{person}{Matthias
  Heizmann}, \bibinfo{person}{Jochen Hoenicke}, \bibinfo{person}{Alexander
  Nutz}, {and} \bibinfo{person}{Andreas Podelski}.}
  \bibinfo{year}{2019}\natexlab{}.
\newblock \showarticletitle{Ultimate TreeAutomizer (CHC-COMP Tool
  Description)}.
\newblock \bibinfo{journal}{\emph{Electronic Proceedings in Theoretical
  Computer Science}}  \bibinfo{volume}{296} (\bibinfo{date}{July}
  \bibinfo{year}{2019}), \bibinfo{pages}{42–47}.
\newblock
\showISSN{2075-2180}
\href{https://doi.org/10.4204/eptcs.296.7}{doi:\nolinkurl{10.4204/eptcs.296.7}}


\bibitem[Dijkstra(1976)]{dijkstra1976}
\bibfield{author}{\bibinfo{person}{Edsger~W. Dijkstra}.}
  \bibinfo{year}{1976}\natexlab{}.
\newblock \bibinfo{booktitle}{\emph{A Discipline of Programming}}.
\newblock \bibinfo{publisher}{Prentice-Hall}, \bibinfo{address}{Englewood
  Cliffs, N.J.}
\newblock
\showISBNx{978-0-13-215871-8}


\bibitem[Ekici et~al\mbox{.}(2017)]{SMTCoq}
\bibfield{author}{\bibinfo{person}{Burak Ekici}, \bibinfo{person}{Alain
  Mebsout}, \bibinfo{person}{Cesare Tinelli}, \bibinfo{person}{Chantal Keller},
  \bibinfo{person}{Guy Katz}, \bibinfo{person}{Andrew Reynolds}, {and}
  \bibinfo{person}{Clark Barrett}.} \bibinfo{year}{2017}\natexlab{}.
\newblock \showarticletitle{{SMTCoq: A plug-in for integrating SMT solvers into
  Coq}}. In \bibinfo{booktitle}{\emph{{Computer Aided Verification - 29th
  International Conference}}}. \bibinfo{address}{Heidelberg, Germany}.
\newblock
\urldef\tempurl \url{https://hal.science/hal-01669345}
\showURL{\tempurl}


\bibitem[Filli\^{a}tre and Paskevich(2013)]{Why3}
\bibfield{author}{\bibinfo{person}{Jean-Christophe Filli\^{a}tre} {and}
  \bibinfo{person}{Andrei Paskevich}.} \bibinfo{year}{2013}\natexlab{}.
\newblock \showarticletitle{Why3: where programs meet provers}. In
  \bibinfo{booktitle}{\emph{Proceedings of the 22nd European Conference on
  Programming Languages and Systems}} (Rome, Italy)
  \emph{(\bibinfo{series}{ESOP'13})}. \bibinfo{publisher}{Springer-Verlag},
  \bibinfo{address}{Berlin, Heidelberg}, \bibinfo{pages}{125–128}.
\newblock
\showISBNx{9783642370359}
\href{https://doi.org/10.1007/978-3-642-37036-6_8}{doi:\nolinkurl{10.1007/978-3-642-37036-6_8}}


\bibitem[Flanagan(2004)]{flanagan2004automatic}
\bibfield{author}{\bibinfo{person}{Cormac Flanagan}.}
  \bibinfo{year}{2004}\natexlab{}.
\newblock \showarticletitle{Automatic software model checking via constraint
  logic}.
\newblock \bibinfo{journal}{\emph{Science of Computer Programming}}
  \bibinfo{volume}{50}, \bibinfo{number}{1-3} (\bibinfo{year}{2004}),
  \bibinfo{pages}{253--270}.
\newblock


\bibitem[Flanagan and Leino(2001)]{Houdini}
\bibfield{author}{\bibinfo{person}{Cormac Flanagan} {and}
  \bibinfo{person}{K.~Rustan~M. Leino}.} \bibinfo{year}{2001}\natexlab{}.
\newblock \showarticletitle{Houdini, an Annotation Assistant for ESC/Java}. In
  \bibinfo{booktitle}{\emph{{FME} 2001: Formal Methods for Increasing Software
  Productivity, International Symposium of Formal Methods Europe, Berlin,
  Germany, March 12-16, 2001, Proceedings}} \emph{(\bibinfo{series}{Lecture
  Notes in Computer Science}, Vol.~\bibinfo{volume}{2021})},
  \bibfield{editor}{\bibinfo{person}{Jos{\'{e}}~Nuno Oliveira} {and}
  \bibinfo{person}{Pamela Zave}} (Eds.). \bibinfo{publisher}{Springer},
  \bibinfo{pages}{500--517}.
\newblock
\href{https://doi.org/10.1007/3-540-45251-6\_29}{doi:\nolinkurl{10.1007/3-540-45251-6\_29}}


\bibitem[Flanagan et~al\mbox{.}(2002)]{escjava}
\bibfield{author}{\bibinfo{person}{Cormac Flanagan},
  \bibinfo{person}{K.~Rustan~M. Leino}, \bibinfo{person}{Mark Lillibridge},
  \bibinfo{person}{Greg Nelson}, \bibinfo{person}{James~B. Saxe}, {and}
  \bibinfo{person}{Raymie Stata}.} \bibinfo{year}{2002}\natexlab{}.
\newblock \showarticletitle{Extended static checking for Java}. In
  \bibinfo{booktitle}{\emph{Proceedings of the ACM SIGPLAN 2002 Conference on
  Programming Language Design and Implementation}} (Berlin, Germany)
  \emph{(\bibinfo{series}{PLDI '02})}. \bibinfo{publisher}{Association for
  Computing Machinery}, \bibinfo{address}{New York, NY, USA},
  \bibinfo{pages}{234–245}.
\newblock
\showISBNx{1581134630}
\href{https://doi.org/10.1145/512529.512558}{doi:\nolinkurl{10.1145/512529.512558}}


\bibitem[Floyd(1967)]{Floyd1967}
\bibfield{author}{\bibinfo{person}{Robert~W. Floyd}.}
  \bibinfo{year}{1967}\natexlab{}.
\newblock \showarticletitle{Assigning {{Meanings}} to {{Programs}}}. In
  \bibinfo{booktitle}{\emph{Mathematical {{Aspects}} of {{Computer Science}}}},
  \bibfield{editor}{\bibinfo{person}{J.~T. Schwartz}} (Ed.),
  Vol.~\bibinfo{volume}{19}. \bibinfo{publisher}{American Mathematical
  Society}, \bibinfo{pages}{19--32}.
\newblock


\bibitem[Gladshtein et~al\mbox{.}(2026)]{Loom}
\bibfield{author}{\bibinfo{person}{Vladimir Gladshtein},
  \bibinfo{person}{George P\^{\i}rlea}, \bibinfo{person}{Qiyuan Zhao},
  \bibinfo{person}{Vitaly Kurin}, {and} \bibinfo{person}{Ilya Sergey}.}
  \bibinfo{year}{2026}\natexlab{}.
\newblock \showarticletitle{Foundational Multi-Modal Program Verifiers}.
\newblock \bibinfo{journal}{\emph{Proc. ACM Program. Lang.}}
  \bibinfo{volume}{10}, \bibinfo{number}{POPL}, Article \bibinfo{articleno}{77}
  (\bibinfo{date}{Jan.} \bibinfo{year}{2026}), \bibinfo{numpages}{32}~pages.
\newblock
\href{https://doi.org/10.1145/3776719}{doi:\nolinkurl{10.1145/3776719}}


\bibitem[Graf and Saidi(1997)]{Graf97}
\bibfield{author}{\bibinfo{person}{Susanne Graf} {and} \bibinfo{person}{Hassen
  Saidi}.} \bibinfo{year}{1997}\natexlab{}.
\newblock \showarticletitle{Construction of abstract state graphs with PVS}. In
  \bibinfo{booktitle}{\emph{Computer Aided Verification}},
  \bibfield{editor}{\bibinfo{person}{Orna Grumberg}} (Ed.).
  \bibinfo{publisher}{Springer Berlin Heidelberg}, \bibinfo{address}{Berlin,
  Heidelberg}, \bibinfo{pages}{72--83}.
\newblock
\showISBNx{978-3-540-69195-2}


\bibitem[Gurfinkel et~al\mbox{.}(2015)]{SeaHorn}
\bibfield{author}{\bibinfo{person}{Arie Gurfinkel}, \bibinfo{person}{Temesghen
  Kahsai}, \bibinfo{person}{Anvesh Komuravelli}, {and}
  \bibinfo{person}{Jorge~A. Navas}.} \bibinfo{year}{2015}\natexlab{}.
\newblock \showarticletitle{The SeaHorn Verification Framework}. In
  \bibinfo{booktitle}{\emph{International Conference on Computer Aided
  Verification}}. \bibinfo{pages}{343--361}.
\newblock
\href{https://doi.org/10.1007/978-3-319-21690-4_20}{doi:\nolinkurl{10.1007/978-3-319-21690-4_20}}


\bibitem[Hamza et~al\mbox{.}(2019)]{stainless}
\bibfield{author}{\bibinfo{person}{Jad Hamza}, \bibinfo{person}{Nicolas
  Voirol}, {and} \bibinfo{person}{Viktor Kun\v{c}ak}.}
  \bibinfo{year}{2019}\natexlab{}.
\newblock \showarticletitle{System FR: formalized foundations for the stainless
  verifier}.
\newblock \bibinfo{journal}{\emph{Proc. ACM Program. Lang.}}
  \bibinfo{volume}{3}, \bibinfo{number}{OOPSLA}, Article
  \bibinfo{articleno}{166} (\bibinfo{date}{Oct.} \bibinfo{year}{2019}),
  \bibinfo{numpages}{30}~pages.
\newblock
\href{https://doi.org/10.1145/3360592}{doi:\nolinkurl{10.1145/3360592}}


\bibitem[Henzinger et~al\mbox{.}(2002b)]{blastpcc}
\bibfield{author}{\bibinfo{person}{Thomas~A. Henzinger},
  \bibinfo{person}{Ranjit Jhala}, \bibinfo{person}{Rupak Majumdar},
  \bibinfo{person}{George~C. Necula}, \bibinfo{person}{Gr\'{e}goire Sutre},
  {and} \bibinfo{person}{Westley Weimer}.} \bibinfo{year}{2002}\natexlab{b}.
\newblock \showarticletitle{Temporal-Safety Proofs for Systems Code}. In
  \bibinfo{booktitle}{\emph{Proceedings of the 14th International Conference on
  Computer Aided Verification}} \emph{(\bibinfo{series}{CAV '02})}.
  \bibinfo{publisher}{Springer-Verlag}, \bibinfo{address}{Berlin, Heidelberg},
  \bibinfo{pages}{526–538}.
\newblock
\showISBNx{3540439978}


\bibitem[Henzinger et~al\mbox{.}(2002a)]{BLAST}
\bibfield{author}{\bibinfo{person}{Thomas~A. Henzinger},
  \bibinfo{person}{Ranjit Jhala}, \bibinfo{person}{Rupak Majumdar}, {and}
  \bibinfo{person}{Gr\'{e}goire Sutre}.} \bibinfo{year}{2002}\natexlab{a}.
\newblock \showarticletitle{Lazy abstraction}. In
  \bibinfo{booktitle}{\emph{Proceedings of the 29th ACM SIGPLAN-SIGACT
  Symposium on Principles of Programming Languages}} (Portland, Oregon)
  \emph{(\bibinfo{series}{POPL '02})}. \bibinfo{publisher}{Association for
  Computing Machinery}, \bibinfo{address}{New York, NY, USA},
  \bibinfo{pages}{58–70}.
\newblock
\showISBNx{1581134509}
\href{https://doi.org/10.1145/503272.503279}{doi:\nolinkurl{10.1145/503272.503279}}


\bibitem[Hoare(1969)]{Hoare1969}
\bibfield{author}{\bibinfo{person}{Charles Anthony~Richard Hoare}.}
  \bibinfo{year}{1969}\natexlab{}.
\newblock \showarticletitle{An axiomatic basis for computer programming}.
\newblock \bibinfo{journal}{\emph{Commun. ACM}} \bibinfo{volume}{12},
  \bibinfo{number}{10} (\bibinfo{year}{1969}), \bibinfo{pages}{576--580}.
\newblock


\bibitem[Hojjat and Rümmer(2018)]{ELDARICA}
\bibfield{author}{\bibinfo{person}{Hossein Hojjat} {and}
  \bibinfo{person}{Philipp Rümmer}.} \bibinfo{year}{2018}\natexlab{}.
\newblock \showarticletitle{The ELDARICA Horn Solver}. In
  \bibinfo{booktitle}{\emph{2018 Formal Methods in Computer Aided Design
  (FMCAD)}}. \bibinfo{pages}{1--7}.
\newblock
\href{https://doi.org/10.23919/FMCAD.2018.8603013}{doi:\nolinkurl{10.23919/FMCAD.2018.8603013}}


\bibitem[Johnson et~al\mbox{.}(2023)]{wave}
\bibfield{author}{\bibinfo{person}{Evan Johnson}, \bibinfo{person}{Evan
  Laufer}, \bibinfo{person}{Zijie Zhao}, \bibinfo{person}{Dan Gohman},
  \bibinfo{person}{Shravan Narayan}, \bibinfo{person}{Stefan Savage},
  \bibinfo{person}{Deian Stefan}, {and} \bibinfo{person}{Fraser Brown}.}
  \bibinfo{year}{2023}\natexlab{}.
\newblock \showarticletitle{WaVe: a verifiably secure WebAssembly sandboxing
  runtime}. In \bibinfo{booktitle}{\emph{2023 IEEE Symposium on Security and
  Privacy (SP)}}. IEEE, \bibinfo{pages}{2940--2955}.
\newblock


\bibitem[Karp(1972)]{Karp72}
\bibfield{author}{\bibinfo{person}{Richard~M. Karp}.}
  \bibinfo{year}{1972}\natexlab{}.
\newblock \showarticletitle{Reducibility Among Combinatorial Problems}. In
  \bibinfo{booktitle}{\emph{Proceedings of a symposium on the Complexity of
  Computer Computations, held March 20-22, 1972, at the {IBM} Thomas J. Watson
  Research Center, Yorktown Heights, New York, {USA}}}
  \emph{(\bibinfo{series}{The {IBM} Research Symposia Series})},
  \bibfield{editor}{\bibinfo{person}{Raymond~E. Miller} {and}
  \bibinfo{person}{James~W. Thatcher}} (Eds.). \bibinfo{publisher}{Plenum
  Press, New York}, \bibinfo{pages}{85--103}.
\newblock
\href{https://doi.org/10.1007/978-1-4684-2001-2\_9}{doi:\nolinkurl{10.1007/978-1-4684-2001-2\_9}}


\bibitem[Komuravelli et~al\mbox{.}(2016)]{Spacer1}
\bibfield{author}{\bibinfo{person}{Anvesh Komuravelli}, \bibinfo{person}{Arie
  Gurfinkel}, {and} \bibinfo{person}{Sagar Chaki}.}
  \bibinfo{year}{2016}\natexlab{}.
\newblock \showarticletitle{SMT-based model checking for recursive programs}.
\newblock \bibinfo{journal}{\emph{Form. Methods Syst. Des.}}
  \bibinfo{volume}{48}, \bibinfo{number}{3} (\bibinfo{date}{June}
  \bibinfo{year}{2016}), \bibinfo{pages}{175–205}.
\newblock
\showISSN{0925-9856}
\href{https://doi.org/10.1007/s10703-016-0249-4}{doi:\nolinkurl{10.1007/s10703-016-0249-4}}


\bibitem[Lattuada et~al\mbox{.}(2023)]{verus}
\bibfield{author}{\bibinfo{person}{Andrea Lattuada}, \bibinfo{person}{Travis
  Hance}, \bibinfo{person}{Chanhee Cho}, \bibinfo{person}{Matthias Brun},
  \bibinfo{person}{Isitha Subasinghe}, \bibinfo{person}{Yi Zhou},
  \bibinfo{person}{Jon Howell}, \bibinfo{person}{Bryan Parno}, {and}
  \bibinfo{person}{Chris Hawblitzel}.} \bibinfo{year}{2023}\natexlab{}.
\newblock \showarticletitle{Verus: Verifying Rust Programs using Linear Ghost
  Types}.
\newblock \bibinfo{journal}{\emph{Proc. ACM Program. Lang.}}
  \bibinfo{volume}{7}, \bibinfo{number}{OOPSLA1}, Article
  \bibinfo{articleno}{85} (\bibinfo{date}{April} \bibinfo{year}{2023}),
  \bibinfo{numpages}{30}~pages.
\newblock
\href{https://doi.org/10.1145/3586037}{doi:\nolinkurl{10.1145/3586037}}


\bibitem[Lehmann et~al\mbox{.}(2023)]{FluxRS}
\bibfield{author}{\bibinfo{person}{Nico Lehmann}, \bibinfo{person}{Adam~T.
  Geller}, \bibinfo{person}{Niki Vazou}, {and} \bibinfo{person}{Ranjit Jhala}.}
  \bibinfo{year}{2023}\natexlab{}.
\newblock \showarticletitle{Flux: Liquid Types for Rust}.
\newblock \bibinfo{journal}{\emph{Proc. ACM Program. Lang.}}
  \bibinfo{volume}{7}, \bibinfo{number}{PLDI}, Article \bibinfo{articleno}{169}
  (\bibinfo{date}{June} \bibinfo{year}{2023}), \bibinfo{numpages}{25}~pages.
\newblock
\href{https://doi.org/10.1145/3591283}{doi:\nolinkurl{10.1145/3591283}}


\bibitem[Leino(2010)]{dafny}
\bibfield{author}{\bibinfo{person}{K.~Rustan~M. Leino}.}
  \bibinfo{year}{2010}\natexlab{}.
\newblock \showarticletitle{Dafny: an automatic program verifier for functional
  correctness}. In \bibinfo{booktitle}{\emph{Proceedings of the 16th
  International Conference on Logic for Programming, Artificial Intelligence,
  and Reasoning}} (Dakar, Senegal) \emph{(\bibinfo{series}{LPAR'10})}.
  \bibinfo{publisher}{Springer-Verlag}, \bibinfo{address}{Berlin, Heidelberg},
  \bibinfo{pages}{348–370}.
\newblock
\showISBNx{3642175104}


\bibitem[Leino and Pit-Claudel(2016)]{leino2016trigger}
\bibfield{author}{\bibinfo{person}{K.~Rustan~M. Leino} {and}
  \bibinfo{person}{Cl{\'{e}}ment Pit-Claudel}.}
  \bibinfo{year}{2016}\natexlab{}.
\newblock \showarticletitle{Trigger Selection Strategies to Stabilize Program
  Verifiers}. In \bibinfo{booktitle}{\emph{Computer Aided Verification - 28th
  International Conference, {CAV} 2016, Toronto, ON, Canada, July 17-23, 2016,
  Proceedings, Part I}} \emph{(\bibinfo{series}{Lecture Notes in Computer
  Science}, Vol.~\bibinfo{volume}{9779})}. \bibinfo{publisher}{Springer},
  \bibinfo{pages}{361--381}.
\newblock
\href{https://doi.org/10.1007/978-3-319-41528-4_20}{doi:\nolinkurl{10.1007/978-3-319-41528-4_20}}


\bibitem[Luckham et~al\mbox{.}(1979)]{Luckham1979}
\bibfield{author}{\bibinfo{person}{D.~C. Luckham}, \bibinfo{person}{S.~M.
  German}, \bibinfo{person}{F.~W. von Henke}, \bibinfo{person}{R.~A. Karp},
  \bibinfo{person}{P.~W. Milne}, \bibinfo{person}{D.~C. Oppen},
  \bibinfo{person}{W. Polak}, {and} \bibinfo{person}{W.~L. Scherlis}.}
  \bibinfo{year}{1979}\natexlab{}.
\newblock \bibinfo{booktitle}{\emph{Stanford Pascal Verifier user manual}}.
\newblock \bibinfo{type}{Technical Report} STAN-CS-79-731.
  \bibinfo{institution}{Stanford University, Department of Computer Science}.
\newblock


\bibitem[Matsushita et~al\mbox{.}(2021)]{matsushita2021rusthorn}
\bibfield{author}{\bibinfo{person}{Yusuke Matsushita}, \bibinfo{person}{Takeshi
  Tsukada}, {and} \bibinfo{person}{Naoki Kobayashi}.}
  \bibinfo{year}{2021}\natexlab{}.
\newblock \showarticletitle{RustHorn: CHC-Based Verification for Rust
  Programs}.
\newblock \bibinfo{journal}{\emph{ACM Transactions on Programming Languages and
  Systems (TOPLAS)}} \bibinfo{volume}{43}, \bibinfo{number}{4}
  (\bibinfo{year}{2021}), \bibinfo{pages}{1--47}.
\newblock
\href{https://doi.org/10.1145/3462205}{doi:\nolinkurl{10.1145/3462205}}


\bibitem[McMillan(2006)]{IMPACT}
\bibfield{author}{\bibinfo{person}{Kenneth~L McMillan}.}
  \bibinfo{year}{2006}\natexlab{}.
\newblock \showarticletitle{Lazy abstraction with interpolants}. In
  \bibinfo{booktitle}{\emph{International Conference on Computer Aided
  Verification (CAV)}}. Springer, \bibinfo{pages}{123--136}.
\newblock
\href{https://doi.org/10.1007/11817963_14}{doi:\nolinkurl{10.1007/11817963_14}}


\bibitem[Mohamed et~al\mbox{.}(2025)]{lean-smt}
\bibfield{author}{\bibinfo{person}{Abdalrhman Mohamed}, \bibinfo{person}{Tomaz
  Mascarenhas}, \bibinfo{person}{Harun Khan}, \bibinfo{person}{Haniel Barbosa},
  \bibinfo{person}{Andrew Reynolds}, \bibinfo{person}{Yicheng Qian},
  \bibinfo{person}{Cesare Tinelli}, {and} \bibinfo{person}{Clark Barrett}.}
  \bibinfo{year}{2025}\natexlab{}.
\newblock \showarticletitle{lean-smt: An SMT Tactic for Discharging Proof Goals
  in Lean}. In \bibinfo{booktitle}{\emph{Computer Aided Verification}},
  \bibfield{editor}{\bibinfo{person}{Ruzica Piskac} {and}
  \bibinfo{person}{Zvonimir Rakamari{\'{c}}}} (Eds.).
  \bibinfo{publisher}{Springer Nature Switzerland}, \bibinfo{address}{Cham},
  \bibinfo{pages}{197--212}.
\newblock
\showISBNx{978-3-031-98682-6}


\bibitem[Moura and Ullrich(2021)]{Lean4}
\bibfield{author}{\bibinfo{person}{Leonardo~de Moura} {and}
  \bibinfo{person}{Sebastian Ullrich}.} \bibinfo{year}{2021}\natexlab{}.
\newblock \showarticletitle{The Lean 4 Theorem Prover and Programming
  Language}. In \bibinfo{booktitle}{\emph{Automated Deduction – CADE 28: 28th
  International Conference on Automated Deduction, Virtual Event, July 12–15,
  2021, Proceedings}}. \bibinfo{publisher}{Springer-Verlag},
  \bibinfo{address}{Berlin, Heidelberg}, \bibinfo{pages}{625–635}.
\newblock
\showISBNx{978-3-030-79875-8}
\href{https://doi.org/10.1007/978-3-030-79876-5_37}{doi:\nolinkurl{10.1007/978-3-030-79876-5_37}}


\bibitem[Necula(1998)]{necula1998compiling}
\bibfield{author}{\bibinfo{person}{George~Ciprian Necula}.}
  \bibinfo{year}{1998}\natexlab{}.
\newblock \emph{\bibinfo{title}{Compiling with Proofs}}.
\newblock \bibinfo{thesistype}{Ph.\,D. Dissertation}. \bibinfo{school}{Carnegie
  Mellon University}, \bibinfo{address}{Pittsburgh, PA, USA}.
\newblock
\urldef\tempurl \url{https://www.cs.cmu.edu/~rwh/students/necula.pdf}
\showURL{\tempurl}
\newblock
\shownote{Available as Technical Report CMU-CS-98-154}.


\bibitem[Nipkow et~al\mbox{.}(2002)]{Isabelle}
\bibfield{author}{\bibinfo{person}{Tobias Nipkow}, \bibinfo{person}{Lawrence~C.
  Paulson}, {and} \bibinfo{person}{Markus Wenzel}.}
  \bibinfo{year}{2002}\natexlab{}.
\newblock \bibinfo{booktitle}{\emph{Isabelle/HOL --- A Proof Assistant for
  Higher-Order Logic}}. \bibinfo{series}{Lecture Notes in Computer Science},
  Vol.~\bibinfo{volume}{2283}.
\newblock \bibinfo{publisher}{Springer}.
\newblock
\href{https://doi.org/10.1007/3-540-45949-9}{doi:\nolinkurl{10.1007/3-540-45949-9}}


\bibitem[Norell(2009)]{Norell2009}
\bibfield{author}{\bibinfo{person}{Ulf Norell}.}
  \bibinfo{year}{2009}\natexlab{}.
\newblock \bibinfo{booktitle}{\emph{Dependently Typed Programming in Agda}}.
\newblock \bibinfo{publisher}{Springer Berlin Heidelberg},
  \bibinfo{address}{Berlin, Heidelberg}, \bibinfo{pages}{230--266}.
\newblock
\showISBNx{978-3-642-04652-0}
\href{https://doi.org/10.1007/978-3-642-04652-0_5}{doi:\nolinkurl{10.1007/978-3-642-04652-0_5}}


\bibitem[Ogawa et~al\mbox{.}(2025)]{Thrust}
\bibfield{author}{\bibinfo{person}{Hiromi Ogawa}, \bibinfo{person}{Taro
  Sekiyama}, {and} \bibinfo{person}{Hiroshi Unno}.}
  \bibinfo{year}{2025}\natexlab{}.
\newblock \showarticletitle{Thrust: A Prophecy-Based Refinement Type System for
  Rust}.
\newblock \bibinfo{journal}{\emph{Proc. ACM Program. Lang.}}
  \bibinfo{volume}{9}, \bibinfo{number}{PLDI}, Article \bibinfo{articleno}{230}
  (\bibinfo{date}{June} \bibinfo{year}{2025}), \bibinfo{numpages}{25}~pages.
\newblock
\href{https://doi.org/10.1145/3729333}{doi:\nolinkurl{10.1145/3729333}}


\bibitem[Podelski and Rybalchenko(2007)]{podelski2007armc}
\bibfield{author}{\bibinfo{person}{Andreas Podelski} {and}
  \bibinfo{person}{Andrey Rybalchenko}.} \bibinfo{year}{2007}\natexlab{}.
\newblock \showarticletitle{{ARMC}: The Logical Choice for Software Model
  Checking with Abstraction Refinement}. In \bibinfo{booktitle}{\emph{Practical
  Aspects of Declarative Languages (PADL)}} \emph{(\bibinfo{series}{Lecture
  Notes in Computer Science}, Vol.~\bibinfo{volume}{4354})}.
  \bibinfo{publisher}{Springer}, \bibinfo{pages}{245--259}.
\newblock
\href{https://doi.org/10.1007/978-3-540-69611-7_16}{doi:\nolinkurl{10.1007/978-3-540-69611-7_16}}


\bibitem[Rondon et~al\mbox{.}(2008)]{Rondon08}
\bibfield{author}{\bibinfo{person}{Patrick~M. Rondon}, \bibinfo{person}{Ming
  Kawaguchi}, {and} \bibinfo{person}{Ranjit Jhala}.}
  \bibinfo{year}{2008}\natexlab{}.
\newblock \showarticletitle{Liquid types}. In
  \bibinfo{booktitle}{\emph{Proceedings of the 29th ACM SIGPLAN Conference on
  Programming Language Design and Implementation}} (Tucson, AZ, USA)
  \emph{(\bibinfo{series}{PLDI '08})}. \bibinfo{publisher}{Association for
  Computing Machinery}, \bibinfo{address}{New York, NY, USA},
  \bibinfo{pages}{159–169}.
\newblock
\showISBNx{9781595938602}
\href{https://doi.org/10.1145/1375581.1375602}{doi:\nolinkurl{10.1145/1375581.1375602}}


\bibitem[Schuermann et~al\mbox{.}(2025)]{tock-proper}
\bibfield{author}{\bibinfo{person}{Leon Schuermann}, \bibinfo{person}{Brad
  Campbell}, \bibinfo{person}{Branden Ghena}, \bibinfo{person}{Philip Levis},
  \bibinfo{person}{Amit Levy}, {and} \bibinfo{person}{Pat Pannuto}.}
  \bibinfo{year}{2025}\natexlab{}.
\newblock \showarticletitle{Tock: From Research to Securing 10 Million
  Computers}. In \bibinfo{booktitle}{\emph{Proceedings of the ACM SIGOPS 31st
  Symposium on Operating Systems Principles}}. \bibinfo{pages}{36--49}.
\newblock


\bibitem[Swamy et~al\mbox{.}(2013)]{F*-dijkstra-monads}
\bibfield{author}{\bibinfo{person}{Nikhil Swamy}, \bibinfo{person}{Joel
  Weinberger}, \bibinfo{person}{Cole Schlesinger}, \bibinfo{person}{Juan Chen},
  {and} \bibinfo{person}{Benjamin Livshits}.} \bibinfo{year}{2013}\natexlab{}.
\newblock \showarticletitle{Verifying higher-order programs with the dijkstra
  monad}. In \bibinfo{booktitle}{\emph{Proceedings of the 34th ACM SIGPLAN
  Conference on Programming Language Design and Implementation}} (Seattle,
  Washington, USA) \emph{(\bibinfo{series}{PLDI '13})}.
  \bibinfo{publisher}{Association for Computing Machinery},
  \bibinfo{address}{New York, NY, USA}, \bibinfo{pages}{387–398}.
\newblock
\showISBNx{9781450320146}
\href{https://doi.org/10.1145/2491956.2491978}{doi:\nolinkurl{10.1145/2491956.2491978}}


\bibitem[Tarjan(1972)]{Tarjan72}
\bibfield{author}{\bibinfo{person}{Robert Tarjan}.}
  \bibinfo{year}{1972}\natexlab{}.
\newblock \showarticletitle{Depth-First Search and Linear Graph Algorithms}.
\newblock \bibinfo{journal}{\emph{SIAM J. Comput.}} \bibinfo{volume}{1},
  \bibinfo{number}{2} (\bibinfo{year}{1972}), \bibinfo{pages}{146--160}.
\newblock
\href{https://doi.org/10.1137/0201010}{doi:\nolinkurl{10.1137/0201010}}


\bibitem[{The Coq Development Team}(2024)]{Coq-refman}
\bibfield{author}{\bibinfo{person}{{The Coq Development Team}}.}
  \bibinfo{year}{2024}\natexlab{}.
\newblock \bibinfo{title}{The {Coq} Reference Manual -- Release 8.19.0}.
\newblock
  \bibinfo{howpublished}{\url{https://coq.inria.fr/doc/V8.19.0/refman}}.
\newblock


\bibitem[Tsukada and Unno(2024)]{Spacer2}
\bibfield{author}{\bibinfo{person}{Takeshi Tsukada} {and}
  \bibinfo{person}{Hiroshi Unno}.} \bibinfo{year}{2024}\natexlab{}.
\newblock \showarticletitle{Inductive Approach to Spacer}.
\newblock \bibinfo{journal}{\emph{Proc. ACM Program. Lang.}}
  \bibinfo{volume}{8}, \bibinfo{number}{PLDI}, Article \bibinfo{articleno}{227}
  (\bibinfo{date}{June} \bibinfo{year}{2024}), \bibinfo{numpages}{24}~pages.
\newblock
\href{https://doi.org/10.1145/3656457}{doi:\nolinkurl{10.1145/3656457}}


\bibitem[Vazou et~al\mbox{.}(2014)]{LiquidHaskell}
\bibfield{author}{\bibinfo{person}{Niki Vazou}, \bibinfo{person}{Eric~L.
  Seidel}, \bibinfo{person}{Ranjit Jhala}, \bibinfo{person}{Dimitrios
  Vytiniotis}, {and} \bibinfo{person}{Simon Peyton-Jones}.}
  \bibinfo{year}{2014}\natexlab{}.
\newblock \showarticletitle{Refinement types for Haskell}. In
  \bibinfo{booktitle}{\emph{Proceedings of the 19th ACM SIGPLAN International
  Conference on Functional Programming}} (Gothenburg, Sweden)
  \emph{(\bibinfo{series}{ICFP '14})}. \bibinfo{publisher}{Association for
  Computing Machinery}, \bibinfo{address}{New York, NY, USA},
  \bibinfo{pages}{269–282}.
\newblock
\showISBNx{9781450328739}
\href{https://doi.org/10.1145/2628136.2628161}{doi:\nolinkurl{10.1145/2628136.2628161}}


\bibitem[Zee et~al\mbox{.}(2008)]{jahob}
\bibfield{author}{\bibinfo{person}{Karen Zee}, \bibinfo{person}{Viktor Kuncak},
  {and} \bibinfo{person}{Martin Rinard}.} \bibinfo{year}{2008}\natexlab{}.
\newblock \showarticletitle{Full functional verification of linked data
  structures}. In \bibinfo{booktitle}{\emph{Proceedings of the 29th ACM SIGPLAN
  Conference on Programming Language Design and Implementation}} (Tucson, AZ,
  USA) \emph{(\bibinfo{series}{PLDI '08})}. \bibinfo{publisher}{Association for
  Computing Machinery}, \bibinfo{address}{New York, NY, USA},
  \bibinfo{pages}{349–361}.
\newblock
\showISBNx{9781595938602}
\href{https://doi.org/10.1145/1375581.1375624}{doi:\nolinkurl{10.1145/1375581.1375624}}


\bibitem[Zhou et~al\mbox{.}(2023)]{mariposa}
\bibfield{author}{\bibinfo{person}{Yi Zhou}, \bibinfo{person}{Jay Bosamiya},
  \bibinfo{person}{Yoshiki Takashima}, \bibinfo{person}{Jessica Li},
  \bibinfo{person}{Marijn Heule}, {and} \bibinfo{person}{Bryan Parno}.}
  \bibinfo{year}{2023}\natexlab{}.
\newblock \showarticletitle{Mariposa: Measuring SMT Instability in Automated
  Program Verification}. In \bibinfo{booktitle}{\emph{2023 Formal Methods in
  Computer-Aided Design (FMCAD)}}. \bibinfo{pages}{178--188}.
\newblock
\href{https://doi.org/10.34727/2023/isbn.978-3-85448-060-0_26}{doi:\nolinkurl{10.34727/2023/isbn.978-3-85448-060-0_26}}


\end{thebibliography}
\end{document}